\documentclass[twocolumn,pra,longbibliography,aps,superscriptaddress,10pt,floatfix]{revtex4-2}
\usepackage[utf8]{inputenc}
\usepackage[english]{babel}
\usepackage{fontenc}
\usepackage{anyfontsize}
\usepackage[toc,page]{appendix}
\usepackage{braket}
\usepackage{qcircuit}
\usepackage{amsmath}
\usepackage{amssymb}
\usepackage{amsthm}
\usepackage[mathlines]{lineno}
\usepackage{algorithm}
\usepackage{algpseudocode}
\usepackage{algcompatible}
\usepackage{array}
\usepackage{makecell}
\usepackage{xfrac}
\usepackage[letterpaper,margin=.6in,tmargin=.8in,bmargin=.9in]{geometry}

\floatname{algorithm}{Protocol}
\newlength{\continueindent}
\usepackage{etoolbox}
\makeatletter
\newcommand*{\ALG@customparshape}{\parshape 2 \leftmargin \linewidth \dimexpr\ALG@tlm+\continueindent\relax \dimexpr\linewidth+\leftmargin-\ALG@tlm-\continueindent\relax}
\apptocmd{\ALG@beginblock}{\ALG@customparshape}{}{\errmessage{failed to patch}}
\makeatother

\usepackage{graphicx}
\usepackage{wrapfig}
\usepackage[colorlinks = true, linkcolor = blue, urlcolor=blue, citecolor = red, anchorcolor = green]{hyperref}
\usepackage{tikz}
\usepackage{caption}
\usepackage{subcaption}

\usepackage{dcolumn}
\usepackage{bm}
\usepackage{comment}

\usepackage{mathtools}

\newtheorem{theorem}{Theorem}

\newtheorem{corollary}{Corollary}

\newtheorem{algorithmmain}{Algorithm}

\let\oldtextcolor\textcolor
\renewcommand{\textcolor}[2]{%
  \begingroup
  \edef\tempa{#1}%
  \edef\tempb{green}%
  \ifx\tempa\tempb
    \endgroup #2%
  \else
    \endgroup \oldtextcolor{#1}{#2}%
  \fi
}

\allowdisplaybreaks

\newcommand{\mmw}[1] {{\color{blue}MMW: [#1]}}

\begin{document}
\preprint{APS/123-QED}

\title{Digital Quantum Simulations of the Non-Resonant Open Tavis--Cummings Model}

\author{Aidan N. Sims} 
\affiliation{School of Applied and Engineering Physics,
Cornell University, Ithaca, New York 14850, USA}
\affiliation{Department of Computer Science, Cornell University, Ithaca, New York 14850, USA}

\author{Dhrumil Patel}
\affiliation{Department of Computer Science, Cornell University, Ithaca, New York 14850, USA}

\author{Aby Philip}
\affiliation{School of Applied and Engineering Physics,
Cornell University, Ithaca, New York 14850, USA}
\affiliation{Institute of Fundamental Technological Research, Polish Academy of Sciences, Pawi\'{n}skiego 5B, 02-106 Warsaw, Poland}

\author{Alex H. Rubin}
\affiliation{Department of Electrical and Computer Engineering, University of California, Davis, California 95616, USA}
\affiliation{Department of Physics and Astronomy, University of California, Davis, California
95616, USA}

\author{Rahul Bandyopadhyay}
\affiliation{Department of Electrical and Computer Engineering, University of California, Davis, California 95616, USA}
\affiliation{Applied Quantum Algorithms, Leiden University, Leiden, 2333 CA, The Netherlands}

\author{Marina Radulaski}
\affiliation{Department of Electrical and Computer Engineering, University of California, Davis, California 95616, USA}

\author{Mark M. Wilde}
\affiliation{School of Electrical and Computer Engineering, Cornell University, Ithaca, New York
14850, USA}
\affiliation{School of Applied and Engineering Physics,
Cornell University, Ithaca, New York 14850, USA}


\begin{abstract}
    The open Tavis--Cummings model consists of $N$ quantum emitters interacting with a common cavity mode, accounts for losses and decoherence, and is frequently explored for quantum information processing and designing quantum devices. As $N$ increases, it becomes harder to simulate the open Tavis--Cummings model using traditional methods. To address this problem, we implement two quantum algorithms for simulating the dynamics of this model in the inhomogeneous, non-resonant regime, with up to three excitations in the cavity. We show that the implemented algorithms have gate complexities that scale polynomially, as $O(N^2)$ and $O(N^3)$, while the number of qubits used by these algorithms (space complexity) scales linearly as $O(N)$. One of these algorithms is the sampling-based wave matrix Lindbladization algorithm, for which we propose two protocols to implement its system-independent fixed interaction, resolving key open questions of [Patel~and~Wilde, \textit{Open Sys.~\&~Info.~Dyn.},~30:2350014~(2023)]. \textcolor{green}{We benchmark our results against a classical differential equation solver in a variety of scenarios and demonstrate that our algorithms accurately reproduce the expected dynamics.} 
\end{abstract}

\maketitle

\tableofcontents

\section{Introduction}

\subsection{Motivation}

    A foundational object of study in quantum optics is a linear cavity coupled to one or more two-level systems, representing atoms or quantum emitters. The well-known single-emitter case, described by the Jaynes--Cummings model and its variants~\cite{larson2021jaynes}, captures the physics underlying various quantum technologies, including cavity quantum electrodynamics (QED) experiments~\cite{kimble1998strong}, circuit QED systems~\cite{wallraff2004strong}, and quantum dots in photonic crystals~\cite{yoshie2004vacuum}, among many others. The Tavis--Cummings (TC) model extends this framework to $N$ quantum emitters interacting with a common cavity mode~\cite{PhysRev.170.379}. This extension is particularly relevant for modeling optical quantum devices based on color centers or atoms, where many emitters can easily occupy a single cavity due to their intrinsically small size.
    
    Coupling $N$ emitters to a cavity can enhance their collective coupling by a factor of $\sqrt{N}$, which can be beneficial for emitters with small dipole moments such as color centers~\cite{zhong2017interfacing, lukin2023two}, and it also opens the door to physical phenomena that are not present in the single-emitter case, such as super- and sub-radiance, novel types of photon blockade~\cite{dicke1954coherence, gross1982superradiance, radulaski2017photon, PhysRevLett.122.243602, white2022enhancing}, and collectively induced transparency~\cite{lei2023many}. These effects have potential applications in quantum technologies, particularly in developing enhanced light-matter interfaces, efficient single-photon sources, and optical quantum memories~\cite{PhysRevA.80.012317}. To explore this many-body physics, and to design experiments and devices making use of its collective effects, it is necessary to model the behavior of open TC systems that can exchange excitations with their environments (among other decoherence processes).

    For these reasons, among others, there has recently been a shift in focus to simulating open systems rather than closed systems. Closed models are governed by Schr\"{o}dinger's equation, and their dynamics are governed by a Hamiltonian. On the other hand, open systems that have Markovian dynamics are governed by the Lindblad master equation~\cite{cmp/1103899849, gorini1976completely}. Open models are of greater physical relevance because almost all physical models contain noisy or non-unitary interactions.


    It is well known that quantum systems are generally difficult to simulate classically. Indeed, the naive approach to solving an $N$-body Lindblad master equation using Liouville operators (based on a vectorized density matrix) requires memory and runtime that scale exponentially in $N$. There are a variety of classical techniques that improve on this scaling by making various assumptions, which all cut down the size of the Hilbert space to be simulated. These include the use of an effective Hamiltonian~\cite{radulaski2017nonclassical}, only valid in the single-excitation regime, scattering matrix methods~\cite{PhysRevLett.122.243602}, which focus on the dynamics of few-photon states of the scattered field, and the use of quantum trajectories~\cite{PhysRevA.89.052133}, which does not have a closed-form solution if the Hamiltonian does not conserve the number of excitations. Another result showed that quantum inverse methods can be used to find solutions; however, it is difficult to extract quantities of interest using these methods~\cite{ExactTC, QIM}. In addition to the aforementioned methods, variational techniques inspired by the density matrix renormalization group (DMRG) and matrix product states (MPS) have been leveraged to study the photon statistical dynamics of open quantum systems akin to the open TC model~\cite{Dhar_2018, Zens_2021, tiwary2024protectinginformationparametricallydriven}.    


    The challenge of simulating quantum systems was the original impetus behind Feynman's proposal for quantum computers~\cite{Feynman:1981tf}.
    Digital quantum computers have made great strides over the last two decades, with appreciable increases in qubit count and coherence times~\cite{QC_review}. A wide variety of physical models have also been successfully mapped onto qubits, including special cases of the Tavis-Cummings model~\cite{huang2022qubitizationbosons, Tudorovskaya_2024, Chiew_2023, Derby_2021, Marinkovic_2023, rubin2024}. 
    In addition, quantum simulations are one of the most promising near-term applications of digital quantum computers~\cite{RevModPhys.86.153, Clinton_2024}. 
     

    Significant progress has recently been made in developing new open-systems quantum simulation algorithms~\cite{ cleve2019efficient, childsopen, PhysRevResearch.4.023216,   PRXQuantum.3.010320,Suri2023twounitary,WMLpartI, WMLpartII} (see~\cite{openReview} for a review). These algorithms have applications in condensed matter physics~\cite{Olmos_2012, Prosen_2011, PhysRevE.86.061118}, quantum chemistry~\cite{10.1093/oso/9780198529798.001.0001chembook, molecularsystemsbook}, quantum optics~\cite{quantumnoisebook, RevModPhys.70.101}, entanglement preparation~\cite{PhysRevLett.117.040501, PhysRevA.78.042307, PhysRevLett.106.090502}, and other fields~\cite{PhysRevA.87.012324, verstraete2008quantumcomputationquantumstate, kastoryano2016quantumgibbssamplerscommuting}. 

\subsection{Contributions}
    
    In this paper, we implement the wave matrix Lindbladization (WML) algorithm~\cite{WMLpartI, WMLpartII} and a variant of the algorithm from~\cite{cleve2019efficient}, which we refer to as the Split $J$-Matrix algorithm, to simulate the open TC model. These two algorithms differ in their input model; the WML algorithm assumes sample access to program states that encode Lindblad operators in a set $\left\{L_i\right\}_i$, whereas the Split $J$-Matrix algorithm assumes that all these operators are available in matrix form. We show the results of using these algorithms to simulate the open TC model and compare their performances \cite{code_repo}.
    
    Our paper contains several key contributions. First, we resolve an open question from~\cite{WMLpartI} by designing two protocols for implementing the fixed interaction in the WML algorithm. Both of these protocols are based on the linear combination of unitaries method for channels from~\cite[Sections~3 \& 4]{cleve2019efficient}. We show that this fixed interaction is independent of the system being simulated and easily scales to larger systems, for the case in which the Lindblad operators are local and act on a constant number of qubits. Second, we show that the gate complexities---the number of one- and two-qubit gates---for the WML and Split J-Matrix algorithms scale quadratically and cubically with the number of emitters, $N$, respectively,  while the number of qubits scales linearly in $N$. This is an exponential improvement over the time and space required by typical classical Lindblad equation solvers. Finally, our results show that our quantum algorithms can be used to simulate non-resonant and inhomogeneous regimes of the open TC model, both of which are inaccessible to standard classical simulation techniques.

    \subsection{Paper Organization}

     The rest of the paper is structured as follows. In Section~\ref{subsec:Notation}, we explain the notation we use, and then Section~\ref{sec:Notation_preliminaries} provides more background on the non-resonant open TC model, the algorithm proposed in~\cite{cleve2019efficient}, which we refer to as the $J$-Matrix algorithm, and the WML algorithm. We then present, in Section~\ref{sec:methods_split_j_matrix}, an improved version of the $J$-Matrix algorithm, i.e., the Split $J$-Matrix algorithm.
     In Section~\ref{sec:wml-practical}, along with Appendices~\ref{app:WML_channel} and~\ref{app:reduce-aux-overhead}, we present two protocols for implementing the fixed interaction of the WML algorithm.
     In Section~\ref{sec:Methods}, we demonstrate how to map the excitation-number states of the cavity and emitters to qubits so that we can employ our quantum algorithms for simulating the open TC model.
     Section~\ref{sec:wml-program-states-tc} describes the program states that encode the Hamiltonian and Lindblad operators of the open TC model.
     In Section~\ref{sec:cc}, we investigate the gate complexities of our algorithms. Next, in Section~\ref{sec:Results},  we provide the results of using these algorithms to numerically simulate the behavior of the TC model. We conclude the paper in Sections~\ref{sec:Discussion} and \ref{sec:Conclusion} by summarizing our results and detailing questions for future research.

\section{Notation}

\label{subsec:Notation}

    We start by establishing some basic mathematical notation used throughout the rest of the paper (see~\cite{WMLpartI} for similar notation). First, let the Hilbert space of a $d$-dimensional system associated with the quantum system $S$ be denoted by $\mathcal{H}_S$. The set of quantum states acting on $\mathcal{H}_S$ is denoted by $\mathcal{D}(\mathcal{H}_S)$. The trace of a matrix $X$ is denoted by $\operatorname{Tr}[X]$, and the conjugate transpose or adjoint of $X$ is denoted by $X^\dag$. The partial trace over systems $B$ and $C$ in a joint state $\rho_{ABC}$ of systems $ABC$ is denoted by $\operatorname{Tr}_{BC}[\rho_{ABC}]$.

    To analyze the performance of the algorithms in this paper, we define various norms of an operator. For  all $p\in[1,\infty)$, the Schatten-$p$ norm of an operator $X$ is defined as
    \begin{equation} \label{eqn:p-norm}
        \left\|X\right\|_p \coloneqq \left(\operatorname{Tr}\!\left[\left(X^\dag X\right)^\frac{p}{2}\right]\right)^\frac{1}{p}.
    \end{equation}
    We primarily use $p = 1$, called the trace norm, $p = 2$, called the Hilbert--Schmidt norm, and $p = \infty$, called the operator norm. Note that the operator norm of a matrix corresponds to its maximum singular value. For notational convenience, we omit the subscript `$\infty$' when referring to the operator norm.

    The normalized diamond distance between two quantum channels $\mathcal{N}$ and $\mathcal{M}$ is defined as follows:
    \begin{multline}\label{def:diamond-dist}
        \frac{1}{2}\left\Vert \mathcal{N} - \mathcal{M} \right\Vert_{\diamond} 
        \coloneqq\\ \sup_{\rho \in \mathcal{D}(\mathcal{H}_{R} \otimes \mathcal{H}_{S})}  \frac{1}{2}\left \Vert (\mathcal{I}_{R} \otimes \mathcal{N}  )(\rho) - (\mathcal{I}_{R} \otimes \mathcal{M}  )(\rho)  \right \Vert_{1},
    \end{multline}
    where $R$ is a reference system (of arbitrarily large dimension) and $\mathcal{I}_{R}$ is the identity channel.

    We employ the unitary SWAP operator throughout this paper,  defined as follows:
    \begin{equation}\label{def:swap}
        \operatorname{SWAP} \coloneqq \sum_{i,j} |i\rangle\!\langle j| \otimes |j\rangle\!\langle i|.
    \end{equation}
     Note that a SWAP operation between registers of multiple qubits can be represented as the tensor product of pairwise SWAP operations. 
     A related operator is the unnormalized maximally entangled operator, represented as $|\Gamma\rangle\!\langle\Gamma|$, where
    \begin{equation}\label{def:gamma}
        \ket{\Gamma} \coloneqq \sum_{i} \ket{i}\!\ket{i}.
    \end{equation}

    Finally, the commutator of operators $A$ and $B$ is denoted by $[A,B] \coloneqq AB - BA$, the anti-commutator by $\{A,B\} \coloneqq AB+BA$, and we use the notation $[M]$ to denote the set $\{1, 2, \ldots, M\}$.

\section{Review} 

\label{sec:Notation_preliminaries}

    In this section, we provide a brief review of the non-resonant open TC model and some important background on the algorithms we will use to simulate this model. Specifically, we will discuss Trotterization, the Wave Matrix Lindbladization algorithm~\cite{WMLpartI,WMLpartII}, and the $J$-Matrix algorithm~\cite{cleve2019efficient,pocrnic2024lindbladian}. 

\subsection{Non-Resonant Open Tavis–Cummings Model}\label{subsec:TC_Model}

    The TC model involves a cavity coupled to~$N$ two-level emitters, whose dynamics are governed by the following Hamiltonian:
    \begin{equation}\label{def:TC_Hamil}
        H_{\operatorname{TC}}\coloneqq \omega_{C}a^{\dag}a+\sum_{j=1}^{N}\omega_j\,\sigma_{j}^{+}\sigma_{j}^- + g_{j}\left(\sigma_{j}^{+}a+\sigma_{j}^-a^{\dag}\right),
    \end{equation}
    where we have set $\hbar=1$ here and throughout, $a$ is the annihilation operator corresponding to the cavity, $\omega_{C} > 0$ is the frequency of the cavity, $\sigma_{j}^+ \coloneqq |1\rangle\!\langle 0|$ is the creation operator of the $j$-th emitter, $\sigma_{j}^-\coloneqq |0\rangle\!\langle 1|$ is the annihilation operator of the $j$-th emitter, $\omega_j > 0$ is the frequency of the $j$-th emitter, and $g_j > 0$ is the coupling strength between the cavity and the $j$-th emitter. The interaction between the cavity and $j$-th emitter is governed by the term $\sigma_{j}^{+} a+\sigma_{j}^- a^{\dag}$. This interaction allows the emitters and cavity to exchange excitations. 
    
    Additionally, a coherent photon pump can be attached to the system during evolution, so that excitations can be pumped into the system. This pump can be modeled by adding the following term to the Hamiltonian $H_{\text{TC}}$:
    \begin{equation}\label{def:TC_pump}
        E_P\left(ae^{i\omega_{P}t} + a^\dagger e^{-i\omega_{P}t}\right)\!,
    \end{equation}
    where $\omega_{P} > 0$ is the frequency of the pump and $E_P > 0$ is the power of the coherent pump, having the same units as~$\omega_{C}$. 
    
    We can also understand the Hamiltonian by considering how different parts of the system are affected by different terms in the Hamiltonian. This is shown in Table~\ref{table:tavis_cummings}.
    \begin{table} [ht]
        \centering
        \hfill
        \begin{center}
            \renewcommand{\arraystretch}{1.5} 
            \begin{tabular}{||c | c||} 
            \hline
            \makecell{Hamiltonian} & \makecell{Systems} \\ 
            \hline\hline
            \makecell{$\omega_C a^\dag a$} & \makecell{Cavity} \\ 
            \hline
            \makecell{$E_P\left(ae^{i\omega_{P}t} + a^\dagger e^{-i\omega_{P}t}\right)$} & \makecell{Cavity}  \\
            \hline
            \makecell{$\omega_{j}\,\sigma_{j}^{+}\sigma_{j}^-$} & \makecell{Emitter $j$}  \\
            \hline
            \makecell{$g_j\left( \sigma_j^+ a + a^\dagger \sigma_j^-\right)$} & \makecell{Cavity \& Emitter $j$}  \\
            \hline
            \end{tabular}
        \end{center}
        \caption{Different terms in the TC Hamiltonian, $H_{\operatorname{TC}}$, and the systems upon which they act. }
        \label{table:tavis_cummings}
    \end{table}
    
    Realistically, excitations can decay out of the cavity and emitters. The evolution of the system, when accounting for these decay processes, is governed by the following Lindblad master equation:
    \begin{equation}\label{eqn:open_TC_master}
       \frac{\partial\rho}{\partial t} = -i[H_{\text{TC}},\rho] + \kappa\mathcal{L}_a(\rho)+\sum_{j=1}^{N}\gamma\mathcal{L}_{\sigma_j^-}(\rho), 
    \end{equation}
    where $\rho$ is the combined state of the cavity and all $N$ emitters, $\kappa$ and $\gamma$ represent the rates of excitation loss by the cavity and the emitters, respectively, and 
    \begin{equation} \label{eqn:wml_lindbl}
    \kappa\mathcal{L}_a(\rho)+\sum_{j=1}^{N}\gamma\mathcal{L}_{\sigma^-_j}(\rho)
    \end{equation} 
    is the term that governs the dissipative part of the dynamics. Here, the Lindbladians $\mathcal{L}_a, \mathcal{L}_{\sigma^-_1}, \ldots, \mathcal{L}_{\sigma^-_N}$ are defined through the following superoperator:
    \begin{equation}\label{eqn:lindbladian}
    \mathcal{L}_L(\rho) \coloneqq L\rho L^\dagger - \frac{1}{2}\left\{L^\dagger L,\rho\right\},
    \end{equation}
    for every Lindblad operator $L \in \left\{a, \sigma^-_1, \ldots, \sigma^-_N  \right\}$.
    
    By simulating~\eqref{eqn:open_TC_master}, we want to study the behavior of two important quantities: population and the second-order photon correlation, which is denoted by $g^{(2)}(0)$. Specifically, we first want to estimate the population within the cavity and emitters at any given time $t$. The population refers to the number of excitations in a certain part of the system. The expected value of the population within the cavity is obtained using the following formula:
    \begin{equation}
        \operatorname{Tr}\!\left[a^{\dag} a\rho\right]\label{eq:pop-cavity},
    \end{equation}
    and the expected value of the population within the $j$-th emitter is obtained using 
    \begin{equation}
        \operatorname{Tr}\!\left[\sigma^+_{j} \sigma_{j}^-\rho\right]\label{eq:pop-emitter}.
    \end{equation}
    
    The second quantity of interest is  $g^{(2)}(0)$, the second-order photon correlation, when the cavity is in the steady-state regime ($\dot{\rho} = 0$). This quantity is defined as follows~\cite{STEVENS201325}:
    \begin{equation} \label{eqn:g_2_coherence}
        g^{(2)}(0) \coloneqq \frac{\operatorname{Tr}\!\left[a^\dag a^\dag aa\rho\right]}{\left(\operatorname{Tr}\!\left[a^\dag a\rho\right]\right)^2}.
    \end{equation}

    The relations between the frequency of the cavity, $\omega_{C}$, the frequency of the emitters, $\omega_{j}$, and the coupling strength between the cavity and emitter, $g_{j}$, are important factors to consider within the TC model. If the cavity and an emitter have the same frequency, then the cavity-emitter pair is considered resonant. If each emitter has the same frequency and the same coupling strength to the cavity, the system is considered homogeneous. When all the emitters and the cavity are resonant and homogenous, the system is classically tractable~\cite{Marinkovic_2023}. In this paper, we simulate non-resonant and inhomogeneous systems along with lossy cavities and emitters. 

\subsection{Excitation-Number State to Qubit Mapping for Open TC Model}\label{sec:Methods}

In this section, we demonstrate how to map the excitation-number states of the cavity and emitters to qubits so that we can employ the Split $J$-Matrix algorithm and the WML algorithm for simulating the open TC model.

To achieve this, we model the cavity as a two-qubit system, while each emitter is represented as a one-qubit system. This configuration enables us to simulate up to three excitations in the cavity. Accordingly, we represent each excitation-number state of the cavity in the following manner:
    \begin{equation}
    \label{def:photon_mapping}
        \begin{aligned}
        |00\rangle\!\langle00| & \implies 0 \textrm{ excitations,}  \\
        |01\rangle\!\langle01| & \implies 1 \textrm{ excitation,}  \\
        |10\rangle\!\langle10| & \implies 2 \textrm{ excitations,}  \\
        |11\rangle\!\langle11| & \implies 3 \textrm{ excitations}. 
    \end{aligned}
    \end{equation}
    The annihilation operator $a$ of the cavity can be then written as
    \begin{equation} \label{eqn:annihil_cavity}
        a  = |00\rangle\!\langle01|+\sqrt{2}\,|01\rangle\!\langle10|+\sqrt{3}\,|10\rangle\!\langle11|.
    \end{equation}
    Similarly, the annihilation operator $\sigma_j^{-}$ of the $j$-th emitter can be written as
    \begin{equation} \label{eqn:annihil_emitter}
        \sigma^-_j = |0\rangle\!\langle1|_j.
    \end{equation}
    From~\eqref{eqn:annihil_cavity} and~\eqref{eqn:annihil_emitter}, we obtain 
    \begin{align}
        \sigma_{j}^{+}\sigma_{j}^{-} & = |1\rangle\!\langle1|_{j}\label{eqn:number_op_emitter},\\
         a^{\dag}a & = |01\rangle\!\langle01| + 2\,|10\rangle\!\langle10| + 3\,|11\rangle\!\langle11|.\label{eqn:number_op_cavity}
    \end{align}

\subsection{Background on Trotterization}

\label{subsec:background_trotter}

    Trotterization is a technique for Hamiltonian simulation that leverages the idea that most physically relevant Hamiltonians are sums of smaller Hamiltonians, each acting locally on a constant number of qubits~\cite{lloyd1996universal}. The goal of Hamiltonian simulation is to implement the unitary evolution $e^{-iHt}$, where $H$ is the Hamiltonian of interest. However, it can be challenging to find a sequence of one- and two-qubit gates that realize this unitary evolution exactly. 
    
    To circumvent this, first note that $H$ can be written as a sum of local Hamiltonians. For instance, consider a Hamiltonian $H$ that can be expressed as $H \coloneqq H_1 + H_2 + H_3$, where $H_1$,  $H_2$, and $H_3$ are local Hamiltonians. In such a scenario, one can approximate the unitary $e^{-iHt}$ with the following compositions of unitaries:
    \begin{equation}\label{eqn:trotter_first_order}
        \left(e^{-iH_1\frac{t}{r}}e^{-iH_2\frac{t}{r}}e^{-iH_3\frac{t}{r}}\right)^r,
    \end{equation} 
    where $r \in \mathbb{N}$.
    As $r$ tends to infinity, the distance between the above sequence of unitaries and $e^{-iHt}$ goes to zero. In the literature, such a technique is known as first-order Trotterization, owing to the fact that the aforementioned sequence of unitaries implements the zeroth and first orders of the Taylor expansion of $e^{-iHt}$. 
    
    The second-order Trotter approach is similar, but in addition to applying the aforementioned sequence, i.e., $e^{-iH_1\tau}e^{-iH_2\tau}e^{-iH_3\tau}$, for time $\tau\coloneqq\frac{t}{2r}$, one also applies it in the reverse order for the same amount of time. For example, for $H = H_1 + H_2 + H_3$, the expression
    \begin{equation}\label{eqn:trotter_second_order}
        \left(e^{-iH_1\frac{t}{2r}}e^{-iH_2\frac{t}{2r}}e^{-iH_3\frac{t}{2r}}e^{-iH_3\frac{t}{2r}}e^{-iH_2\frac{t}{2r}}e^{-iH_1\frac{t}{2r}}\right)^r
    \end{equation}
    represents a second-order Trotterization. Note that the above discussion on the first-order and second-order Trotterization can be easily extended for Hamiltonians with an arbitrary number of summands, i.e., $H = \sum_j H_j$~\cite{Suzuki1991}.

\section{\texorpdfstring{$J$}{J}-Matrix and Split \texorpdfstring{$J$}{J}-Matrix Algorithms}

\subsection{Background on \texorpdfstring{$J$}{J}-Matrix Algorithm}

\label{subsec:background_j_matrix}

The authors of~\cite{cleve2019efficient,pocrnic2024lindbladian} proposed an algorithm, here called the $J$-matrix algorithm, to simulate the Lindbladian evolution of a finite-dimensional quantum system, which is in state $\rho$, for time $t$. This evolution is governed by the following Lindblad master equation:
    \begin{equation}\label{eqn:general_master}
        \frac{d\rho}{dt} = \mathcal{L}(\rho) \coloneqq  -i[H, \rho] + \sum_{k=1}^{K}\left(L_k\rho L^\dagger_k - \frac{1}{2} \left\{L^\dagger_k L_k,\rho\right\}\right),  
    \end{equation}
    where $H$ is a Hamiltonian and $L_1, L_2, \ldots, L_K$ are Lindblad operators (not necessarily the Lindblad operators in the open TC model). The superoperator $\mathcal{L}$ is a general Lindbladian, and note that the superoperator $\mathcal{L}_{L}$, as defined in~\eqref{eqn:lindbladian}, is a special case of $\mathcal{L}$ with no Hamiltonian and only one Lindblad operator. The Hamiltonian $H$ is Hermitian, but there is no constraint on the Lindblad operators $L_1, L_2, \ldots, L_K$.
    The $J$-matrix algorithm assumes that the Lindblad operators are embedded in a larger Hermitian matrix in the following way:
    \begin{equation} \label{eqn:sample_j_matrix}
        J\coloneqq\left[\begin{matrix}
            0 & L_1^{\dagger} & L_2^{\dagger} & \cdots & L_K^{\dagger}\\
            L_1  & 0 & 0 & \cdots & 0 \\
            L_2  & 0 & 0 & \cdots & 0 \\
            \vdots  & \vdots & \vdots & \ddots & \vdots \\
            L_K  & 0 & 0 & \cdots & 0 \\
        \end{matrix}\right].
    \end{equation}
    The core idea of the $J$-matrix algorithm is to simulate Lindbladian evolution by performing Hamiltonian evolution on a larger system that includes both the system qubits and some auxiliary qubits. $\lceil\log_2(K+1)\rceil$ auxiliary qubits suffice for simulating evolution with $K$ Lindblad operators. Below, we present pseudocode of the $J$-Matrix algorithm.    
    \begin{algorithmmain}[$J$-Matrix]\label{algo:j_matrix_intro} Set $n \coloneqq O\!\left(\frac{t^2}{\varepsilon}\right)$, where $t\geq 0$ is the simulation time and $\varepsilon \in [0,1]$ is the desired final error in normalized diamond distance. Repeat the following steps $n$ times:  
        \begin{enumerate}
            \item Initialize the auxiliary qubits to the state $|0\rangle\!\langle 0|^{\otimes \lceil \log_2(K+1)\rceil}$.
            \item Apply the unitary $e^{-iJ\sqrt{\frac{t}{n}}}$ to the auxiliary qubits and the system qubits. \label{protstep:j_matrix_unitaries}
            \item Trace out the auxiliary qubits.
            \item Apply the unitary $e^{-iH\!\frac{t}{n}}$ to the system qubits.
        \end{enumerate}
    \end{algorithmmain}
    Note that the unitary operator $e^{-iJ\sqrt{\frac{t}{n}}}$, in Step~\ref{protstep:j_matrix_unitaries} of the above algorithm, acts on both the system qubits and the auxiliary qubits. Additionally, we did not specify in the above algorithm how to decompose this unitary operator into smaller unitary gates that each act on a constant number of qubits. In general, if some structure of this unitary is not known in advance, it may require an exponential number of such gates to implement it. 
    
    Note that, in many physically relevant models, the Lindblad operators $L_1, L_2, \ldots, L_K$ are local operators that each act on only a constant number of qubits, and the Hamiltonian $H$ is a sum of local Hamiltonians, each of which also act on a constant number of qubits. To use this structure to our advantage, we propose an improved version of the standard $J$-matrix algorithm, which we call the Split $J$-Matrix algorithm. A similar algorithm has been analyzed in~\cite{pocrnic2024lindbladian}. The Split $J$-Matrix algorithm requires only $K$ auxiliary qubits, but not an auxiliary environment, which is potentially much larger than the system of interest, like in~\cite{pocrnic2024lindbladian}. In Section~\ref{sec:methods_split_j_matrix}, we explain the Split $J$-Matrix algorithm in more detail. In this algorithm, we employ Trotterization to decompose the unitary $e^{-iJ\sqrt{\frac{t}{n}}}$ into unitary operators that act only on a constant number of qubits. For this reason, we briefly overview the concept of Trotterization in the following subsection.

\subsection{Split \texorpdfstring{$J$}{J}-Matrix}~\label{sec:methods_split_j_matrix}

In this section, we present the Split $J$-Matrix algorithm for simulating the Lindbladian $\mathcal{L}$ as defined in~\eqref{eqn:general_master}. We begin by rewriting this Lindbladian as a sum of the following Lindbladians to simplify the gate complexity analysis for this algorithm, which we present later in Section~\ref{sec:gate-compl-split-J-matrix}:
\begin{equation}
    \mathcal{L}(\rho) =  \underbrace{\mathcal{H}(\rho) + \mathcal{H}'(\rho)}_{\text{coherent}} + \underbrace{\mathcal{N}(\rho)}_{\text{dissipative}}\label{eq:sum_coh_dis},
\end{equation}
where
\begin{align}
    \mathcal{N}(\rho) & \coloneqq \sum_{k=1}^{K}\left(\underbrace{L_k\rho L^\dagger_k - \frac{1}{2} \left\{L^\dagger_k L_k,\rho\right\}}_{\eqqcolon\mathcal{N}_k(\rho)}\right),\\
    \mathcal{H}(\rho) & \coloneqq \sum_{p=1}^{P}\underbrace{-i[H_{p}, \rho]}_{\eqqcolon \mathcal{H}_{p}(\rho)},\\
    \mathcal{H}'(\rho) & \coloneqq \sum_{q=1}^{Q} \underbrace{-i[H'_{q}, \rho]}_{\eqqcolon \mathcal{H}'_{q}(\rho)}.
\end{align}
Here, the coherent part of~\eqref{eq:sum_coh_dis} is composed of the following Hamiltonians: the mutually commuting local Hamiltonians $\{H_{p}\}_p$, and the mutually non-commuting local Hamiltonians $\{H'_{q}\}_q$. Furthermore, we assume that the Hamiltonians in these sets act only on a constant number of qubits. In the dissipative part of~\eqref{eq:sum_coh_dis}, we assume that the Lindblad operators $L_{1}, \ldots, L_{K}$ commute with each other. Both these assumptions are quite common and also hold for the open TC model.

Recall that the naive implementation of the $J$-Matrix algorithm (Algorithm~\ref{algo:j_matrix_intro}) requires the Lindblad operators to be embedded in a larger Hermitian operator, $J$, as shown in~\eqref{eqn:sample_j_matrix}. Furthermore, it involves applying the unitary $e^{iJ\sqrt{\tau}}$ for some small amount of time $\tau$ on both the system qubits and the auxiliary qubits (see Step~\ref{protstep:j_matrix_unitaries} of Algorithm~\ref{algo:j_matrix_intro}). It is simple to see that applying this unitary naively suffers from a critical drawback---the gate complexity for implementing it in general, without any assumption on the Lindblad operators, scales exponentially with the number of system qubits.

We can mitigate this issue for the case that we are considering, that is, the case where the Lindblad operators are local operators acting on a constant number of qubits, and these operators are mutually commuting operators. Hence, we can break the larger matrix $J$  into smaller matrices, namely, $J_{L_1}, \ldots, J_{L_K}$, each encoding only one Lindblad operator at a time. These smaller matrices are defined as follows:
\begin{align}\label{eq:small-J-matrices}
        J_{L_k}& \coloneqq\left[\begin{matrix}
            0  & L_k^{\dagger} \\
            L_k & 0 \\
        \end{matrix}\right],
\end{align}
for all $k \in [K]$. The key idea is then to apply easier-to-implement local unitaries $e^{-iJ_{L_1}\sqrt{\tau}}, \ldots, e^{-iJ_{L_K}\sqrt{\tau}}$ in parallel.  This approach allows us to achieve the same dynamics as applying the larger unitary $e^{-iJ\sqrt{\tau}}$ to the entire system.

In a similar vein, the naive implementation of the $J$-Matrix algorithm involves applying the unitary
\begin{equation}
    e^{-i\left(\sum_{p=1}^{P} H_p + \sum_{q=1}^{Q} H'_q\right)\tau}\label{eq:split-J-full-H-U}
\end{equation}
to simulate the coherent part of~\eqref{eq:sum_coh_dis}. Here as well if there is no structure to the Hamiltonian
\begin{equation}
    \sum_{p=1}^{P} H_p + \sum_{q=1}^{Q} H'_q,
\end{equation}
then the gate complexity for implementing the above unitary is exponential in the number of qubits in general. However, for our case, the above Hamiltonian is the sum of local Hamiltonians that act on a constant number of qubits. Therefore, we apply the second-order Trotterization to the unitary in~\eqref{eq:split-J-full-H-U} in order to decompose it into a product of easier-to-implement local unitaries. Refer to~\eqref{eqn:trotter_second_order} for an explanation of second-order Trotterization, along with an example. With the above notions in place, we now present pseudocode for the Split $J$-Matrix algorithm below.
    \begin{algorithmmain}[Split~$J$-Matrix]\label{protocol:split_j_matrix}
       Set 
       \begin{equation}
           n\coloneqq O\!\left(\frac{(K^2 + Q^2)\lambda_{\max}^2 t^2}{\varepsilon}\right)\label{eq:comp-split-J-mat},
       \end{equation}
       where $t\geq 0$ is the simulation time, $\varepsilon\in (0,1)$ is the desired final error in normalized diamond distance, and
       \begin{multline}\lambda_{\max}\coloneqq\max\Big\{\left\Vert H_{1}\right\Vert, \ldots, \left\Vert H_{P}\right\Vert, \left\Vert H'_{1}\right\Vert, \ldots, \left\Vert H'_{Q}\right\Vert,\\ \left\Vert L_1\right\Vert^2, \ldots, \left\Vert L_K\right\Vert^2 \Big\}\label{eq:split-lambda}.
        \end{multline}
        Repeat the following steps $n$ times:
       \begin{enumerate}
           \item Initialize the auxiliary qubits to the state $|0\rangle^{\otimes K}$.
           \item Apply the local unitaries $e^{-iJ_{L_k}\sqrt{\frac{t}{n}}}$ in parallel, where the unitary $e^{-iJ_{L_k}\sqrt{\frac{t}{n}}}$ acts on the $k$-th auxiliary qubit and the system qubits that $L_k$ acts on non-trivially.
           \item Trace out all the $K$ auxiliary qubits.
           \item Apply the local unitaries $e^{-iH_p\frac{t}{2n}}$ in parallel, where the unitary $e^{-iH_p\frac{t}{2n}}$ acts on the system qubits that $H_p$ acts on non-trivially.
           \item Apply the local unitaries $e^{-iH'_{q}\frac{t}{2N}}$ sequentially, where the unitary $e^{-iH'_{q}\frac{t}{2N}}$ acts on the system qubits that $H'_q$ acts on non-trivially.
           \item Repeat Steps 4 and 5 in the reverse order\textcolor{green}{, with the order of the emitters also reversed.} 
       \end{enumerate}
    \end{algorithmmain}

To understand the expression for the number of iterations, $n$, in~\eqref{eq:comp-split-J-mat}, refer to Section~\ref{sec:gate-compl-split-J-matrix} for an in-depth analysis of the gate complexity of the above algorithm.

\section{Wave Matrix Lindbladization}

\subsection{Background on Wave Matrix Lindbladization}

\label{subsec:background_DME_WML}

    Wave Matrix Lindbladization (WML) \cite{WMLpartI,WMLpartII} is related conceptually to  Density Matrix Exponentiation (DME)~\cite{Lloyd2014QuantumAnalysis}, the latter of which is used to simulate Hamiltonian dynamics when the Hamiltonian is made available in the form of quantum states. See also~\cite{Kimmel_2017,go2024densitymatrixexponentiationsamplebased} for further exposition of DME. While DME is used to simulate closed system dynamics, WML is used to simulate Lindbladian dynamics~\cite{WMLpartI,WMLpartII}. Under the umbrella term of WML, there are two algorithms for simulating Lindbladian dynamics: the sampling-based WML algorithm and the Trotter-like WML algorithm. For our purposes, we focus on the sampling-based algorithm, as we will use it later to simulate the open TC model. For the sake of brevity, we will henceforth refer to the sampling-based WML algorithm simply as the WML algorithm.

 The WML algorithm assumes that the Hamiltonian $H$ is given as a linear combination of mixed states $\left\{\sigma_{j}\right\}_{j=1}^{J}$:
 \begin{equation}\label{eqn:program_states}
        H = \sum_{j=1}^J c_j\sigma_{j},
\end{equation}
where each $c_j \in \mathbb{R}$.
This algorithm also assumes that each Lindblad operator $L_k$ is a local operator, acting on a constant number of qubits, and is given encoded in a pure state $\ket{\psi_k}$ in the following way:
    \begin{equation} \label{def:WML_prog_state}
        \ket{\psi_k} \coloneqq \frac{(L_k \otimes I)\ket{\Gamma}}{\left\|L_k\right\|_2^2},
    \end{equation}
 where $|\Gamma\rangle$ is the unnormalized maximally entangled vector, defined in~\eqref{def:gamma}, and that    we have sample access to multiple copies of $\sigma_{j}$ for all $j \in [J]$ and $\psi_k \coloneqq |\psi_k\rangle\!\langle \psi_k|$ for all $k \in [K]$. In~\cite{WMLpartI,WMLpartII}, the authors referred to these states as program states, a term we will adopt in this paper. The WML algorithm consists of two registers: the system register, initialized in the $d$-dimensional quantum state $\rho$, and the program register. Pseudocode for this algorithm is as follows.

    \begin{algorithmmain}[WML]\label{algo:wml}
    Set $n \coloneqq O\!\left(\frac{c^2 t^2}{\varepsilon}\right)$ and $\Delta \coloneqq \frac{ct}{n}$, where
    \begin{align}\label{def:WML_c}
        c \coloneqq \sum_{j = 1}^J |c_j| + \sum_{k = 1}^K \left\Vert L_k\right\Vert^2_2,
    \end{align}
    $t\geq 0$ is the simulation time, and $\varepsilon\in(0,1)$ is the desired final error in normalized diamond distance. Repeat the following steps $n$ times:
    \begin{enumerate}
        \item \label{protstep:wml_algo_step1} Randomly sample a Hamiltonian program state $\sigma_j$ or a Lindbladian program state $\psi_k$, where $\sigma_j$ has probability $\frac{|c_j|}{c}$ of being sampled and $\psi_k$ has probability $\frac{\left\|L_k\right\|^2_2}{c}$ of being sampled.
        \item Initialize the program register to the state sampled above.
        \item If a Hamiltonian program state $\sigma_j$ is sampled in Step~\ref{protstep:wml_algo_step1}, apply the unitary $e^{-\textnormal{sgn}(c_j)i\,\operatorname{SWAP}\Delta}$ on both the system and program registers. Here, $\operatorname{sgn}(x)$ evaluates to $1$ if $x$ is non-negative and $-1$ otherwise. 
        \item \label{protstep:wml_algo_step4} If a Lindbladian program state $\psi_k$ is sampled in Step~\ref{protstep:wml_algo_step1} instead, apply the quantum channel $e^{\mathcal{M}\Delta}$ on both the program register and the system registers on which $L_k$ acts non-trivially. Here,
        $\mathcal{M}$ is a single-operator Lindbladian:
        \begin{equation}\label{eq:intro-M}
            \mathcal{M} (\cdot) \coloneqq M(\cdot) M^{\dagger} - \frac{1}{2} \left \{M^{\dagger}M, \cdot \right\},
        \end{equation}
        with Lindblad operator
        \begin{equation}
            M \coloneqq \frac{1}{\sqrt{Q}}\left(I_1\otimes |\Gamma\rangle\! \langle \Gamma |_{23}\right) \left( \operatorname{SWAP}_{12} \otimes I_3\right),
            \label{eqn:fixed_interaction_WML}
        \end{equation}
        where register 1 is the system register, registers 2 and 3 jointly represent the program register, and $Q$ is the dimension of the system registers on which $L_k$ acts non-trivially.
            \item Trace out the program register. 
        \end{enumerate}
    \end{algorithmmain}

\subsection{Realizing the Fixed Interaction \texorpdfstring{$e^{\mathcal{M}\Delta}$}{exp(M Delta)}}\label{sec:wml-practical}

In this section, we answer the following question: How can we realize the fixed interaction, that is, the quantum channel $e^{\mathcal{M}\Delta}$ (see Step~\ref{protstep:wml_algo_step4} of Algorithm~\ref{algo:wml}), of the WML algorithm? An answer to this question will resolve one of the key open problems of~\cite{WMLpartII}. Note that this answer applies more broadly to the case where we use the WML algorithm for simulating a general Lindbladian evolution where the Lindblad operators are local operators; that is, it is not limited to simulating the open TC model.

To this end, we employ the LCU-based Lindbladian simulation algorithm proposed in~\cite{cleve2019efficient} to realize the quantum channel $e^{\mathcal{M}\Delta}$. Note that this algorithm assumes an input model where the Lindblad operators are represented as linear combinations of unitaries. 

For our case, we have the Lindbladian $\mathcal{M}$ with a single Lindblad operator $M$ as defined in~\eqref{eqn:fixed_interaction_WML}. Using this LCU-based algorithm, we implement a map $\mathcal{M}_\Delta$ that approximates $e^{\mathcal{M}\Delta}$, where
    \begin{equation}\label{eqn:lcu_approx_wml}
        \mathcal{M}_\Delta(\rho) = \sum^1_{j=0}A_j \rho A_j^\dagger,
    \end{equation}
    $A_0 = I -\frac{\Delta}{2}M^\dagger M$, and $A_1 = \sqrt{\Delta} M$.

    We begin by finding a representation for the Lindblad operator $M$ as a linear combination of unitaries, which, as mentioned before, is the input model for the LCU-based algorithm.
    To simplify things, let us consider the case when the system of interest is a single qubit. Now, since the expression for $M$ consists of operators such as $|\Gamma\rangle\!\langle\Gamma|$ (defined in~\eqref{def:gamma}) and  $\operatorname{SWAP}$ (defined in~\eqref{def:swap}), we first represent these operators as a linear combination of Pauli matrices:   \begin{align}\label{eqn:gamma_swap_as_sum_unitary}
        |\Gamma\rangle\!\langle\Gamma| &= \textcolor{green}{\frac{1}{2}\left(\right.} I\otimes I+ X\otimes X-Y\otimes Y+Z\otimes Z  \textcolor{green}{\left.\right)}, \\
        \operatorname{SWAP}&= \textcolor{green}{\frac{1}{2}\left(\right.} I\otimes I+ X\otimes X+Y\otimes Y+Z\otimes Z \textcolor{green}{\left.\right)}.
    \end{align}
    Plugging the above equations into~\eqref{eqn:fixed_interaction_WML}, we can express $M$ as a linear combination of the 16 Pauli matrices.  Using these 16 Pauli matrices, we can directly use the procedure described in \cite{cleve2019efficient} to realize the fixed interaction $e^{\mathcal{M}\Delta}$. In Appendix~\ref{app:WML_fixed_inter_unitaries}, we show how to extend the implementation of $e^{\mathcal{M}\Delta}$ beyond a single qubit to multiple qubits.
    
    A direct implementation involves 16 controlled unitaries, and each unitary would require up to six control qubits. However, we can reduce the number of controlled unitaries using symmetries inherent to $M$. We provide a detailed circuit diagram and step-by-step procedure to implement $e^{\mathcal{M}\Delta}$ using these symmetries and the LCU method in Appendix~\ref{app:WML_channel}. 

    It is crucial to also note that the LCU-based algorithm requires a number of auxiliary qubits that scale logarithmically with the number of Pauli matrices needed to express $M$. 
    In Appendix~\ref{app:reduce-aux-overhead}, we outline how the aforementioned 16 Pauli matrices can be combined so that $M$ can be expressed as a sum of four unitaries. This quadratic improvement halves the number of required auxiliary qubits. Although this is a constant improvement, it is important for the actual implementation of the algorithm. Additionally, in Appendix~\ref{app:reduce-aux-overhead}, we provide detailed pseudocode of the LCU-based algorithm for approximately implementing the map $e^{\mathcal{M}\Delta}$ using this fewer number of auxiliary qubits.

\subsection{WML Program States for Open TC Model}\label{sec:wml-program-states-tc}

To employ the WML algorithm for simulating the open TC model, governed by the Lindblad master equation in~\eqref{eqn:open_TC_master}, we first need to answer the following question related to the input model of this algorithm: What are choices for program states that encode the Lindblad operators $\sqrt{\kappa}a, \sqrt{\gamma}\sigma^-_1, \ldots, \sqrt{\gamma}\sigma^-_N  $ and the Hamiltonian $H_{\operatorname{TC}}$ of the open TC master equation? We answer this question in what follows.

Recall that the Hamiltonian for the open TC model with a coherent drive is as follows:
    \begin{multline} \omega_{C}a^{\dag}a+\sum_{j=1}^{N}\omega_i\,\sigma_{j}^{+}\sigma_{j}^- + g_{j}\left(\sigma_{j}^{+}a+\sigma_{j}^-a^{\dag}\right)\\
         + E_P\left(ae^{i\omega_{P}t} + a^\dagger e^{-i\omega_{P}t}\right).\!
    \end{multline}
    Now, let us break down this Hamiltonian into program states.
    From~\eqref{eqn:number_op_emitter} and~\eqref{eqn:number_op_cavity}, it is clear that the program state corresponding to the Hamiltonian term $\sigma_{j}^{+}\sigma_{j}^-$ is $|1\rangle\!\langle1|_{j}$ and those corresponding to $a^{\dag}a$ are $|01\rangle\!\langle01|$, $|10\rangle\!\langle10|$, and $|11\rangle\!\langle11|$. Applying a similar analysis on the interaction terms of the Hamiltonian, we get:
    \begin{multline}\label{eqn:inter_hamil_prog_state}
        a \otimes \sigma_{j}^{+}+ a^{\dag}\otimes \sigma_{j}^- =\Psi_{1}-\Psi_{2} + \sqrt{2}\left(\Psi_{3}-\Psi_{4}\right)  \\ + \sqrt{3}\left(\Psi_{5}-\Psi_{6}\right),
    \end{multline}
    where $\Psi_{p} \equiv|\Psi_{p}\rangle\!\langle\Psi_{p}|$, with $p\in\{1,\ldots, 6\}$, are the program states and
    \begin{align}
        &|\Psi_{1}\rangle \coloneqq \frac{1}{\sqrt{2}}\left(|001\rangle+|010\rangle\right),\quad|\Psi_{2}\rangle \coloneqq  I\otimes Z \otimes I\, |\Psi_{1}\rangle,\notag \\
        &|\Psi_{3}\rangle \coloneqq \frac{1}{\sqrt{2}}\left(|011\rangle+|100\rangle\right),\quad|\Psi_{4}\rangle \coloneqq  Z \otimes I \otimes I\, |\Psi_{3}\rangle,\notag \\
        &|\Psi_{5}\rangle \coloneqq \frac{1}{\sqrt{2}}\left(|101\rangle+|110\rangle\right),\quad |\Psi_{6}\rangle \coloneqq  I \otimes Z \otimes I\,|\Psi_{5}\rangle.
    \end{align}
    Likewise, the coherent cavity drive term, $ae^{i\omega_{C}t} + a^\dagger e^{-i\omega_{C}t}$, can be expressed as follows:
    \begin{equation}\label{eqn:coher_drive_prog_state}
         \left(\Phi_{1}-\Phi_{2}\right) + \sqrt{2}\left(\Phi_{3}-\Phi_{4}\right) + \sqrt{3}\left(\Phi_{5}-\Phi_{6}\right),
    \end{equation}
    where $\Phi_{p} \equiv|\Phi_{p}\rangle\!\langle\Phi_{p}|$, with $p\in\{1,\ldots, 6\}$, are the program states and
    \begin{align}
        &|\Phi_{1}\rangle \coloneqq \frac{1}{\sqrt{2}}\left(|00\rangle+e^{-i\omega_Ct}|01\rangle\right),\quad |\Phi_{2}\rangle \coloneqq  I\otimes Z\, |\Phi_{1}\rangle,\notag \\
        &|\Phi_{3}\rangle \coloneqq \frac{1}{\sqrt{2}}\left(|01\rangle+e^{-i\omega_Ct}|10\rangle\right),\quad |\Phi_{4}\rangle \coloneqq  Z \otimes I\, |\Phi_{3}\rangle,\notag \\
        &|\Phi_{5}\rangle \coloneqq \frac{1}{\sqrt{2}}\left(|10\rangle+e^{-i\omega_Ct}|11\rangle\right),\quad |\Phi_{6}\rangle \coloneqq  I \otimes Z\,|\Phi_{5}\rangle.
    \end{align}
    
For the program states associated with the Lindblad operators, we can apply the operators in~\eqref{eqn:annihil_cavity} and~\eqref{eqn:annihil_emitter} to the definition of program states in~\eqref{def:WML_prog_state} to obtain
\begin{align}
    (a \otimes I)|\Gamma\rangle & = |0\rangle|1\rangle,\\
    (\sigma_j^{-} \otimes I)|\Gamma\rangle & = |0\rangle|1\rangle,
\end{align}
for all $j \in [N]$. 

It is important to note that all the program states mentioned in this section are easy to prepare.

\section{Gate Complexity}\label{sec:cc}
In this section, we investigate the gate complexities of the WML and Split $J$-Matrix algorithms for implementing the map $e^{\mathcal{L}t}$, where $\mathcal{L}$ is the Lindbladian defined in~\eqref{eqn:general_master}. By gate complexity, we mean the total number of one- and two-qubit gates required to implement these algorithms. Note that the results of this section are applicable for general Lindbladians, beyond the open TC model.

\subsection{Gate Complexity of the Split \texorpdfstring{$J$}{J}-Matrix Algorithm}\label{sec:gate-compl-split-J-matrix}

With the notations and definitions introduced in Section~\ref{sec:methods_split_j_matrix}, we can rewrite the Split $J$-Matrix algorithm in the quantum channel form in the following way:
\begin{multline}
    \Bigg(\left(\prod_{p=1}^{P} e^{\mathcal{H}_{p}t/2n}\right)\circ \left(\prod_{q=1}^{Q} e^{\mathcal{H}'_{ q}t/n}\right) \circ \left(\prod_{p=1}^{P} e^{\mathcal{H}_{p}t/2n}\right)\\ \circ \mathcal{J}_{1}(t/n)\circ\cdots\circ\mathcal{J}_{K}(t/n) \Bigg)^{\circ n}\label{eq:J-matrix-channel-form},
\end{multline}
where
    \begin{align}
        \mathcal{J}_{k}(t/n)(\rho)& \coloneqq\operatorname{Tr}_{A}\!\left[e^{-iJ_{k}\sqrt{t/n}}\rho e^{iJ_{k}\sqrt{t/n}}\right],\label{def:general_h_for_j}\\
        J_{k}& \coloneqq L^\dagger_{k}\otimes |0\rangle\!\langle 1|_{A} + L_{k}\otimes |1\rangle\!\langle 0|_{A}, \label{def:general_j_for_l}
    \end{align}  
for all $k \in [K]$, and $A$ is a single-qubit auxiliary system.
Note that for simulating the coherent part of the Lindbladian $\mathcal{L}$ in~\eqref{eq:sum_coh_dis}, the Split $J$-Matrix algorithm employs the second-order Trotterization for first splitting $\mathcal{H} + \mathcal{H}'$ and then the first-order Trotterization for splitting the summands in $\mathcal{H}$ and in $\mathcal{H}'$. However, for our purposes and to simplify the analysis, we employ the first-order Trotterization for all three aforementioned splittings. Note that a similar analysis can be performed for the second-order Trotterization, or any higher-order Trotterization, but we leave that for future work. To this end, the Split $J$-Matrix algorithm in quantum channel form is
\begin{multline}
    \Bigg(\left(\prod_{p=1}^{P} e^{\mathcal{H}_{p}t/n}\right)\circ \left(\prod_{q=1}^{Q} e^{\mathcal{H}'_{ q}t/n}\right)\\ \circ \mathcal{J}_{1}(t/n)\circ\cdots\circ\mathcal{J}_{K}(t/n) \Bigg)^{\circ n}\label{eq:sJm-channel-1},
\end{multline}
which can be written more compactly as follows:
\begin{align}
    \left(e^{\mathcal{H}t/n}\circ \left(\prod_{q=1}^{Q} e^{\mathcal{H}'_{ q}t/n}\right) \circ \mathcal{J}_{1}(t/n)\circ\cdots\circ\mathcal{J}_{K}(t/n) \right)^{\circ n}\label{eq:sJm-channel}
\end{align}
due to the following fact:
\begin{equation}
    e^{\mathcal{H}t/n} = \prod_{p=1}^{P} e^{\mathcal{H}_{p}t/n}.
\end{equation}

\begin{theorem}[Gate complexity of the Split $J$-Matrix algorithm]\label{thm:split_j_matrix} 
        Let $\mathcal{L}$ be a Lindbladian, as defined in~\eqref{eq:sum_coh_dis} such that the Lindblad operators $L_{1}, L_{2}, \ldots, L_{K}$ commute with each other. The Split $J$-Matrix algorithm, represented as a quantum channel in~\eqref{eq:sJm-channel}, uses the following number of one- and two-qubit gates such that it is $\varepsilon$-close to the target channel $e^{\mathcal{L}t}$ in normalized diamond distance: 
        \begin{equation}
            O\!\left(\frac{ (P+Q+K)(K^2 + Q^2)\lambda_{\max}^2 t^2}{\varepsilon}\right),
        \end{equation}
        where $\lambda_{\max}$ is defined in~\eqref{eq:split-lambda}.
    \end{theorem}
    \begin{proof}
        See Appendix~\ref{app:split-J-matrix-analysis}.
    \end{proof}

Specifically, for the open TC model, we have $P=Q=K=N$, where $N$ is the number of emitters. Therefore, the above theorem directly implies the following result:
\begin{corollary}
    The Split $J$-Matrix algorithm uses the following number of one- and two-qubit gates to approximate the open TC model dynamics with one cavity and $N$ emitters:
\begin{equation}
    O\!\left(\frac{ N^3\lambda_{\max}^2 t^2}{\varepsilon}\right),
\end{equation}
where $\varepsilon$ is the approximation error in normalized diamond distance.
\end{corollary}

\subsection{Gate Complexity of the WML Algorithm}

Recall that the WML algorithm assumes that the Hamiltonian $H$ is given as a linear combination of quantum states $\{\sigma_j\}_{j=1}^J$, as shown in~\eqref{eqn:program_states}, and that each Lindblad operator $L_k$ is given encoded in a pure state $\ket{\psi_k}$, as shown in~\eqref{def:WML_prog_state}. Step 1 of the WML algorithm is the sampling step where the state $\sigma_j$ is sampled with probability $\frac{c_j}{c}$ (Case~1), the state $\sigma_j$ is sampled with probability $\frac{(-c_j)}{c}$ (Case~2), and the state $\psi_k$ is sampled with probability $\frac{\left \Vert L_k\right\Vert^2_2}{c}$ (Case~3), where $c$ is defined in~\eqref{def:WML_c}. 
Step~2 is simply initializing the program register with the sampled state. Note that the system register is in the state $\rho$. Depending on the case, Step~2 can be represented as the following appending channels, defined for all $j\in [J]$ and $k \in [K]$:
\begin{align}
    \operatorname{(Case\ 1\ and\ Case\ 2):} & \quad \mathcal{P}_{1, j}(\rho) \coloneqq \rho \otimes \sigma_j\\
    \operatorname{(Case\ 3):} & \quad \mathcal{P}_{2, k}(\rho) \coloneqq \rho \otimes \psi_k.
\end{align}
Step~3 of the algorithm involves applying one of the following three quantum channels jointly to the system and program registers, also depending on the case:
\begin{align}
    \operatorname{(Case\ 1):}& \quad e^{\mathcal{N}_1 c\tau}(\rho \otimes \sigma_j) \\
    \operatorname{(Case\ 2):}& \quad e^{\mathcal{N}_2 c\tau}(\rho \otimes \sigma_j)\\
    \operatorname{(Case\ 3):}& \quad e^{\mathcal{M} c\tau}(\rho \otimes \psi_k),
\end{align}
where 
\begin{align}
    \mathcal{N}_1(\cdot) & \coloneqq -i[\operatorname{SWAP}, \cdot]\\
    \mathcal{N}_2(\cdot) & \coloneqq i[\operatorname{SWAP}, \cdot]\\
    \mathcal{M}(\cdot) & \coloneqq M(\cdot)M^{\dagger} -\frac{1}{2}\left\{M^{\dagger}M, \cdot\right\},
\end{align}
with the Lindblad operator $M$ defined as:
\begin{equation}
    M = \frac{1}{\sqrt{Q}}\left(I_{1}\otimes |\Gamma\rangle\!\langle\Gamma|_{23}\right)\left(\operatorname{SWAP}_{12}\otimes I_3\right),
\end{equation}
where $Q \coloneqq 2^q$. Finally, Step~4 of the algorithm is to trace out the program register, and we repeat all the above-mentioned steps $n$ times. 

We represent each iteration of the above algorithm, i.e., Steps 1 to 4, as a quantum channel $\mathcal{A}_{\operatorname{WML, \tau}}^{\operatorname{(ideal)}}$, where $\tau \coloneqq t/n$. This channel is defined as follows:
\begin{align}
    \mathcal{A}_{\operatorname{WML, \tau}}^{\operatorname{(ideal)}} & \coloneqq
    \sum_{j: c_j>0} \frac{c_j}{c} \operatorname{Tr}_{2}\circ~ e^{\mathcal{N}_1 c\tau} \circ \mathcal{P}_{1, j} \notag \\
    & \qquad+  \sum_{j: c_j<0} \frac{(-c_j)}{c} \operatorname{Tr}_{2}\circ~ e^{\mathcal{N}_2 c\tau} \circ \mathcal{P}_{1, j}\notag\\
    & \qquad + \sum_{k} \frac{\left \Vert L_k\right\Vert_2^2}{c} \operatorname{Tr}_{23}\circ~ e^{\mathcal{M} c\tau} \circ \mathcal{P}_{2, k}.\label{eq:ideal-WML-iterate} 
\end{align}
This implies that the entire algorithm can be expressed as the composition of the above channel $n$ times: 
\begin{equation}
    \left (\mathcal{A}_{\operatorname{WML, \tau}}^{\operatorname{(ideal)}} \right)^{\!\circ n}\label{eq:ideal-WML}.
\end{equation}
Note that we use a superscript ``ideal" because we assume that the channels $e^{\mathcal{N}_1 c\tau}, e^{\mathcal{N}_2 c\tau}, $ and $e^{\mathcal{M} c\tau}$ can be implemented exactly without any errors. However, this assumption is not practical. While the channels $e^{\mathcal{N}_1 c\tau}$ and $ e^{\mathcal{N}_2 c\tau}$ can be implemented exactly in principle because they are unitary channels, the same cannot be said for the non-unitary Lindbladian channel $e^{\mathcal{M} c\tau}$.

As mentioned in Section~\ref{sec:wml-practical}, we implement the Lindbladian channel $e^{\mathcal{M} c\tau}$ using an LCU-based algorithm introduced in~\cite{cleve2019efficient}. Let us represent this algorithm as a quantum channel $\mathcal{R}_{c\tau}$. Additionally, we represent  this version of the WML algorithm that employs algorithm $\mathcal{R}_{c\tau}$ as a subroutine for implementing $e^{\mathcal{M} c\tau}$ as
\begin{equation}
    \left (\mathcal{A}_{\operatorname{WML, \tau}}^{\operatorname{(LCU)}} \right)^{\!\circ n},\label{eq:practical-WML}
\end{equation}
where
\begin{align}
    \mathcal{A}_{\operatorname{WML, \tau}}^{\operatorname{(LCU)}} & \coloneqq
    \sum_{j: c_j>0} \frac{c_j}{c} \operatorname{Tr}_{2}\circ~ e^{\mathcal{N}_1 c\tau} \circ \mathcal{P}_{1, j} \notag \\
    & \qquad +  \sum_{j: c_j<0} \frac{(-c_j)}{c} \operatorname{Tr}_{2}\circ~ e^{\mathcal{N}_2 c\tau} \circ \mathcal{P}_{1, j} \notag \\
    & \qquad+ \sum_{k} \frac{\left \Vert L_k\right\Vert_2^2}{c} \operatorname{Tr}_{23}\circ~\mathcal{R}_{c\tau} \circ \mathcal{P}_{2, k}\label{eq:practical-wml-iterate}
\end{align}

\begin{theorem}[Gate complexity of the LCU-based WML algorithm]
\label{thm:gate-comp-WML}
Let $\mathcal{L}$ be a Lindbladian as defined in~\eqref{eqn:general_master}. The LCU-based WML algorithm, represented as a quantum channel in~\eqref{eq:practical-WML}, uses the following number of one- and two-qubit gates such that it is $\varepsilon$-close to the target channel $e^{\mathcal{L}t}$ in normalized diamond distance:
\begin{equation}
    O\!\left(\frac{c^2 t^2 \ln^2(ct/\varepsilon)}{\varepsilon \ln\ln(ct /\varepsilon)}\right),
\end{equation}
where $c$ is defined in~\eqref{def:WML_c}.
\end{theorem} 
\begin{proof}
See Appendix~\ref{app:wml_complexity}.
\end{proof}

The above theorem and the relation between $c$ and $N$ (the number of emitters), which we show to be $c = O(N)$ in Appendix~\ref{app:constants}, implies the following result:
\begin{corollary} The LCU-based WML algorithm uses the following number of one- and two-qubit gates to approximate the open TC model dynamics with one cavity and $N$ emitters:
\begin{equation}
    O\!\left(\frac{N^2 t^2 \ln^2(Nt/\varepsilon)}{\varepsilon \ln\ln(Nt /\varepsilon)}\right),
\end{equation}
where $\varepsilon$ is the approximation error in normalized diamond distance.
\end{corollary}

\section{Results}

\label{sec:Results}

    We developed our simulations using Qiskit v.0.45 and ran them on the QASM simulator from Qiskit-Aer, which is a noiseless quantum computer simulator.
    We then compare these simulations with simulations using the classical Lindblad master equation solver of QuTiP~\cite{QuTiP_1,QuTiP_2}. Finally, we use the Matplotlib Pyplot library to generate the figures in this section.

\subsection{Population Plots}

\label{subsec:population_plots}

    We first demonstrate that our algorithms accurately model the populations of the cavity and emitters over a given time interval. To generate plots, we first select equally spaced times over this interval.  At each selected time, we calculate the populations of the cavity and emitters using one of our quantum algorithms. To understand how to calculate these, refer to the paragraph surrounding~\eqref{eq:pop-cavity} and~\eqref{eq:pop-emitter}. 

    First, we consider a system consisting of a cavity and a single emitter ($N=1$) with $\omega_C = \omega_{E,1} = 245$ THz, $\kappa =24.5$ GHz, $\gamma=0.4$ GHz, and $g_1=100$ GHz, evolving according to the Lindblad master equation, as defined in~\eqref{eqn:open_TC_master}, from time $t_1=0$ ns to $t_2=0.25$ ns. At $t=0$ ns,
    there are two excitations in the cavity but none in the emitter. We begin by selecting 250 equally spaced times over the time interval $[0 , 0.25]$ ns. For each selected time, we ran the $J$-Matrix algorithm starting from $t_1=0$ ns to this time to calculate the populations of the cavity and emitter. For each run of this algorithm, we employed $n=100$ steps. We plot these results in Figure~\ref{fig:2_j_plots}, where the top plot corresponds to the population plots produced using the $J$-Matrix algorithm and the bottom plot corresponds to the population plots produced using the classical solver of QuTiP.

    \begin{figure}[ht]
        \centering
        \includegraphics[width=\columnwidth]{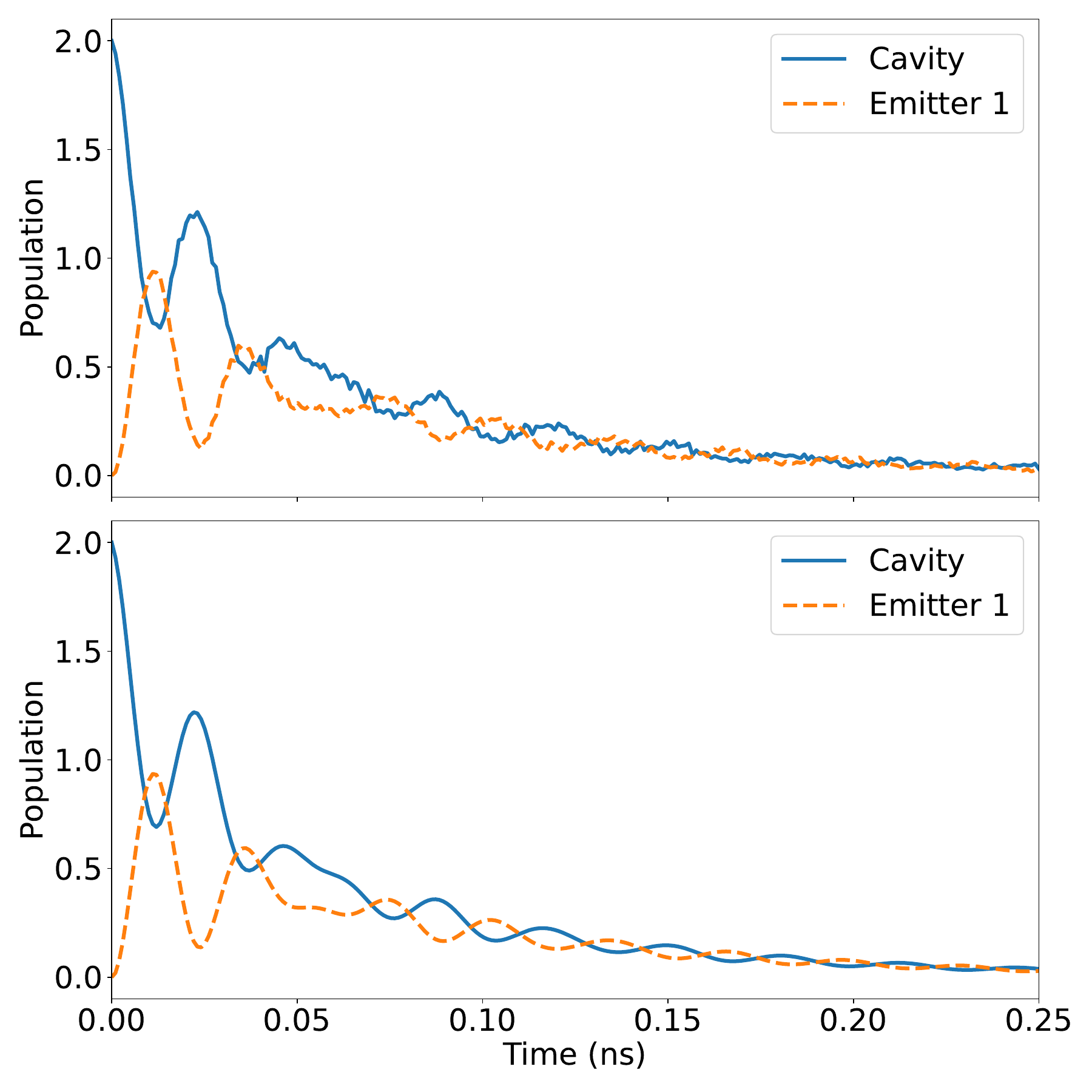}
        \caption{Population of a resonant single-emitter system initialized with two excitations between $t = 0$ and $t = 0.25$ ns. The cavity is coupled to a single resonant emitter ($\omega_C = \omega_{E,1} = 245$ THz) and the system parameters are $(\kappa, \gamma, g_1) = (24.5, 0.4, 100)$ GHz. The top plot represents the result of the $J$-matrix quantum algorithm run on the QASM simulator, while the bottom plot represents the classical solution simulated in QuTiP. The simulation used 1000 shots, giving statistical shot noise of approximately $1/\sqrt{1000} \approx 0.03$.}
        \label{fig:2_j_plots}
    \end{figure}
    
    In Figure~\ref{fig:pumped_j_plots}, we consider the same resonant system as above, but initialized with one cavity excitation instead of two. In addition, we add a coherent drive of strength $E_P = \kappa/2$ (i.e., the cavity is pumped with excitations at a rate of $50\%$ at which they decay) and plot the populations of the cavity and emitter over the same times as we did previously.
    
    \begin{figure}[ht]
        \centering
        \includegraphics[width=\columnwidth]{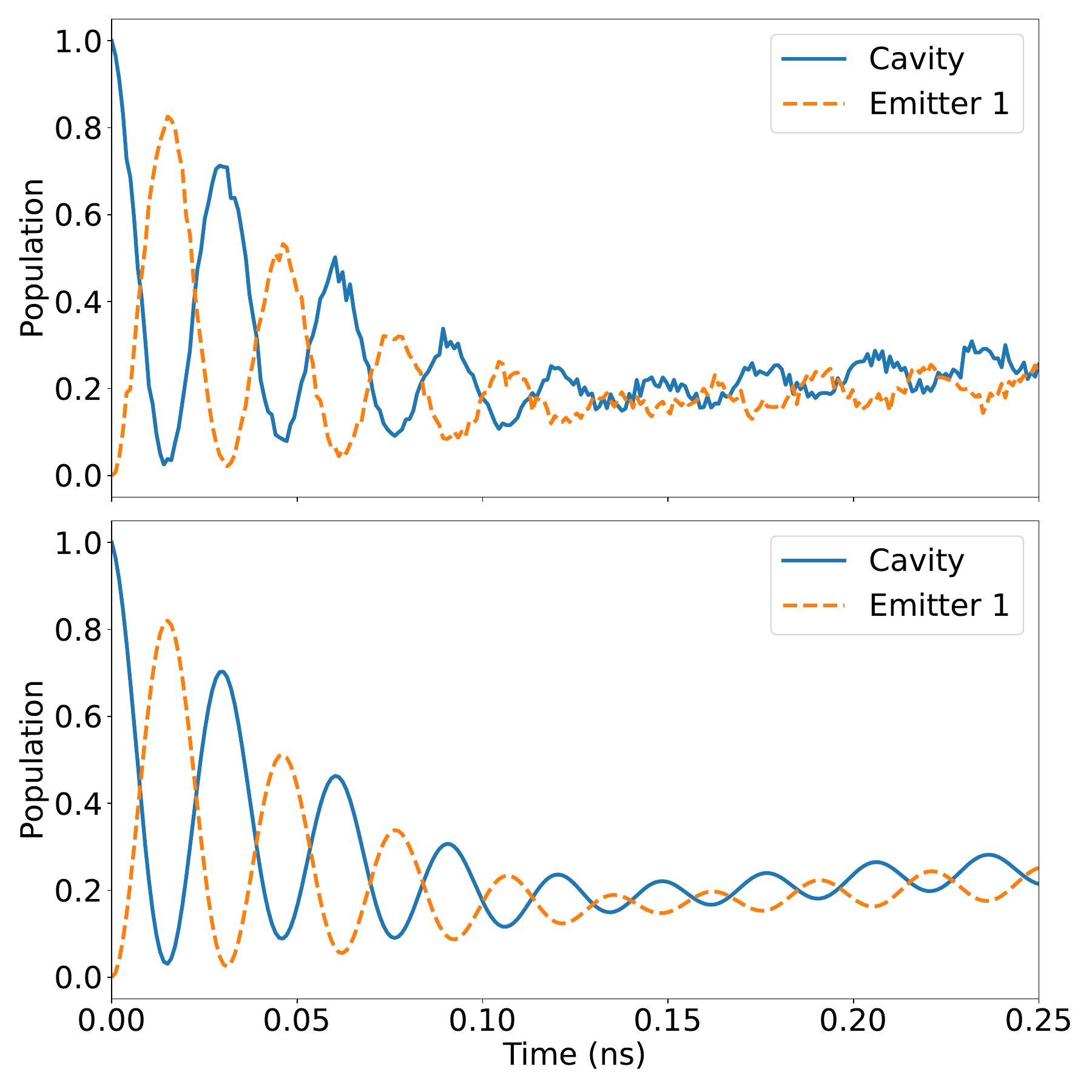}
        \caption{Population of a driven resonant single-emitter system initialized with one photon between $t = 0$ and $t = 0.25$ ns. The cavity is coupled to a single resonant emitter ($\omega_C = \omega_{E,1} = 245$ THz) and the system parameters are $(\kappa, \gamma, g_1) = (24.5, 0.4, 100)$ GHz. The system also has a coherent drive of power $E_P=\kappa/2$.
        The top plot represents the result of the $J$-matrix quantum algorithm run on the QASM simulator, while the bottom plot represents the classical solution simulated in QuTiP.
        The simulation used 1000 shots, giving statistical shot noise of approximately $1/\sqrt{1000} \approx 0.03$.}
        \label{fig:pumped_j_plots}
    \end{figure}

    It is important to note that our algorithms are not limited to homogenous and resonant systems. To demonstrate the inhomogenous case, we consider a system consisting of a cavity with a single excitation and four emitters ($N=4$). We set the cavity frequency to be the same as we set previously, i.e., $\omega_C = 245$ THz; however, we set different frequencies for different emitters because we are considering the inhomogenous case. Specifically, we set the frequency of the first emitter to be $\omega_{E,1} = 245.1$ THz, and then each successive emitter has a higher frequency by the same amount so that $\omega_{E, i} = \omega_{E,i-1} + 100$ GHz. Furthermore, we set $\kappa$ and $\gamma$ to be the same as before, i.e.,  24.5 GHz and 0.4 GHz, respectively, and we set  $g_i = 100$ GHz, for all $i \in \{1,\ldots,4\}$. We then selected 200 equally spaced times from the time interval $[0, 0.25]$ ns, ran the Split $J$-Matrix algorithm for each of these selected times, and calculated the populations of the cavity and four emitters. For each run of this algorithm, we employed $n=50$ steps. We plot the simulation results in Figure~\ref{fig:Staple_j_plots}, where the top plot corresponds to the population plots produced using the Split $J$-Matrix algorithm and the bottom plot corresponds to the population plots produced using QuTiP.
    
    \begin{figure}[ht]
        \centering
        \includegraphics[width=\columnwidth]{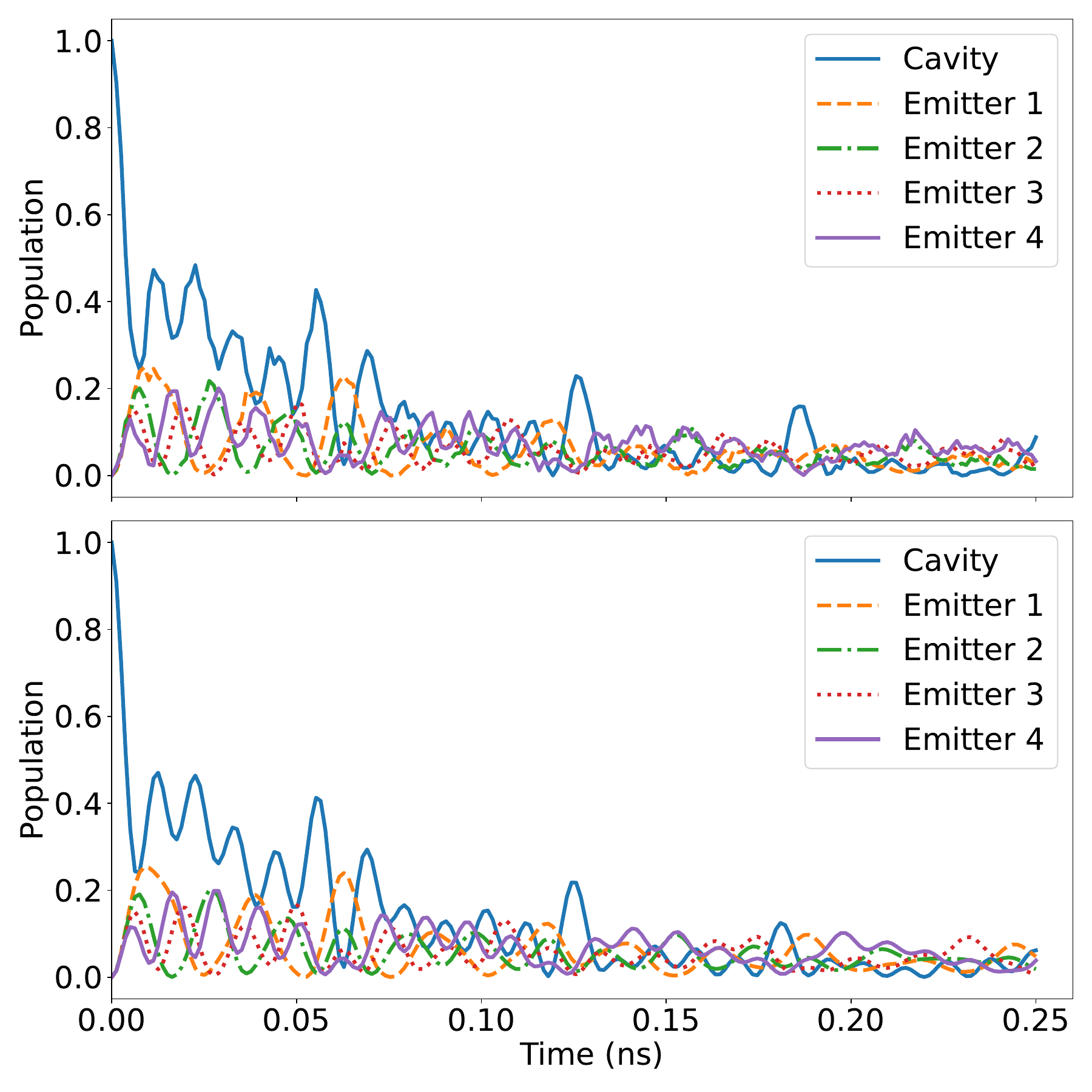}
        \caption{Population of an off-resonant inhomogeneous $N=4$ emitter system initialized with one excitation between $t = 0$ and $t = 0.25$ ns.
        The cavity frequency is $\omega_C = 245$ THz and emitter frequencies are $(\omega_{E,i}) = (245.1, 245.2, 245.3, 245.4)$ THz.
        System parameters are $(\kappa, \gamma, g_i) = (24.5, 0.4, 100)$ GHz.
        The top plot represents the result of the Split $J$-matrix quantum algorithm run on the QASM     simulator, while the bottom plot represents the classical solution simulated in QuTiP.
        The simulation used 1000 shots, giving statistical shot noise of approximately $1/\sqrt{1000} \approx 0.03$.}
        \label{fig:Staple_j_plots}
    \end{figure}
     
    
    Our most important result is that our algorithms expand the parameter space for simulations of the TC model. To give an indication of this, we model the population of a non-resonant $N=9$ emitter system. To this end, we consider a cavity with frequency $\omega_C = 245$ THz. 
    The emitter frequencies are $\omega_{E,i}-\omega_C=(100, -400, -100, 400, 100, 100, 400, -200, -500 )$ GHz.
    Again, we set $\kappa = 24.5$ GHz, $\gamma = 0.4$ GHz, and $g =100$ GHz. To begin with, we initialize the system with three excitations so that the dynamics of all nine emitters are easier to visualize as the system decays. In Figure~\ref{fig:9_emitters_j_graph}, for generating the top plot, we ran the Split $J$-Matrix algorithm for 150 selected times from the time interval $[0, 0.25]$ ns. For each run of the algorithm, we employed $n=45$ steps. Note that the plots in Figure~~\ref{fig:9_emitters_j_graph} are cut off at a population of one excitation because almost all systems spend the entire time in this range.
    
    \begin{figure} [ht]
        \centering
        \includegraphics[width=\columnwidth]{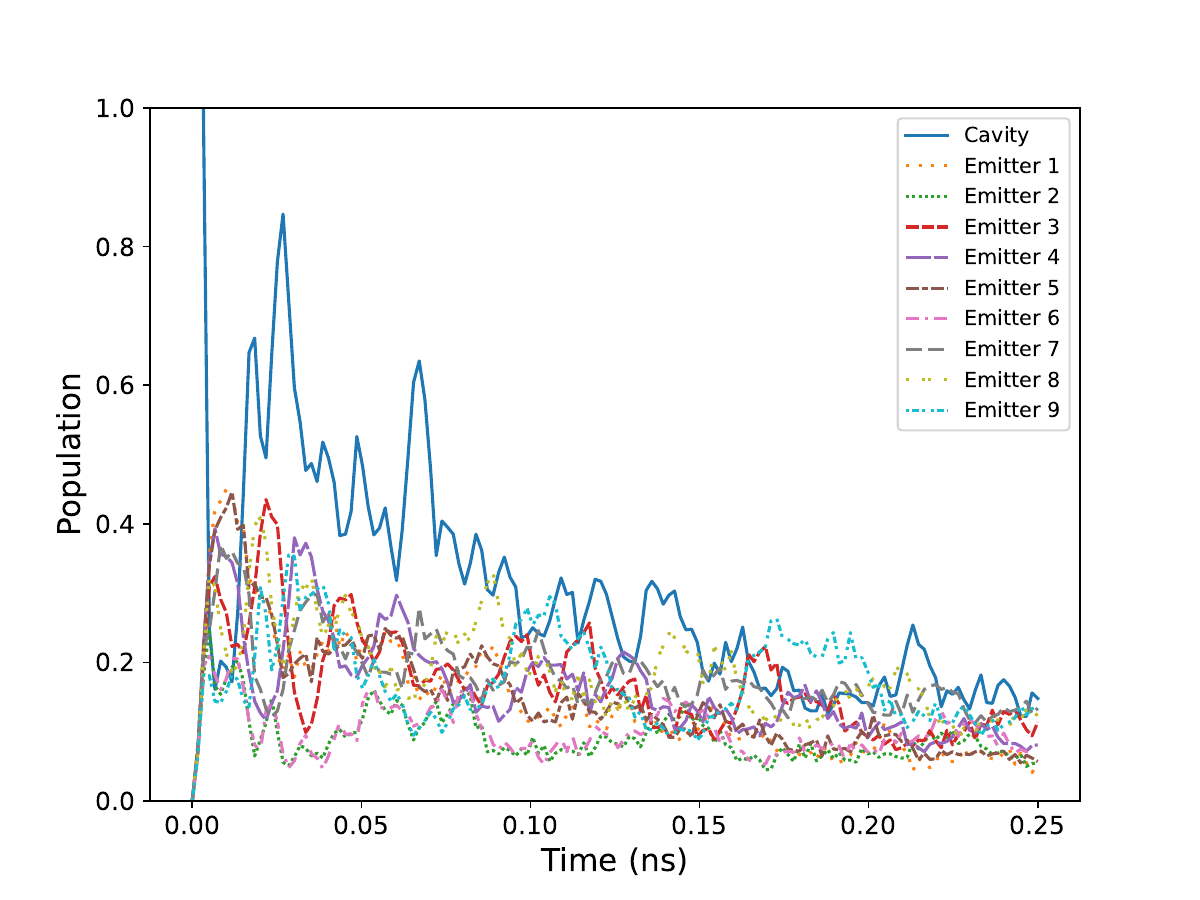}
        \caption{Population of an $N=9$ emitter system initialized with three excitations between $t = 0$ and $t = 0.25$ ns.
        The cavity frequency is $\omega_C = 245$ THz and the emitter frequencies are $\omega_{E,i}-\omega_C=\{100, -400, -100, 0, 100, 100, 400, -200, -500\}$ GHz.
        System parameters are $(\kappa, \gamma, g_i) = (24.5, 0.4, 100)$ GHz. The plot was generated by the Split $J$-Matrix algorithm run on the QASM simulator with 1000 shots, giving statistical shot noise of approximately $1/\sqrt{1000} \approx 0.03$.}
        \label{fig:9_emitters_j_graph}
    \end{figure}

    Next, we consider a system consisting of a cavity and two emitters ($N=2$) with $\omega_C = 245$ GHz, $\omega_{E, 1} = \omega_C + 0.4$ GHz, and $\omega_{E, 2} = \omega_C + 1.3$ GHz. For simulating this system, we make use of a rotating frame, in which 
    \begin{equation}\label{eqn:how_to_elimin_identity}
        e^{-iHt} = e^{-i(H-aI)t}
    \end{equation}
    up to a global phase, where $a$ is some real number and $I$ is the identity matrix. Consequently, simulating the system with $\omega_C = 0$ GHz, $\omega_{E, 1} = 0.4$ GHz, and $\omega_{E, 2} = 1.3$ GHz yields the same results as simulating the aforementioned system with higher values of $\omega_C$, $\omega_{E, 1}$, and $\omega_{E, 2}$. The rotating frame is crucial for employing the WML algorithm (Algorithm~\ref{algo:wml}) to simulate the system even more effectively. This is because it significantly reduces the value of $c$, which directly depends on the values of $\omega_C$, $\omega_{E, 1}$, and $\omega_{E, 2}$. This reduction in the value of $c$ leads to a significant decrease in the runtime of the WML algorithm, which is proportional to $c^2$, as proved in Theorem~\ref{thm:gate-comp-WML}.
    \textcolor{green}{For color center systems, the emitter frequencies $\omega_{E, i}$ are quite high; the $c^2$ time complexity of WML thus makes it more amenable to apply it to systems of few emitters, or in which the emitters have identical or nearly identical frequencies.}
    For the WML algorithm, we create plots by employing a hybrid algorithm in which we use the Split $J$-matrix algorithm for implementing the fixed interaction map $e^{\mathcal{M}\Delta}$ defined in Step~\ref{protstep:wml_algo_step4} of Algorithm~\ref{algo:wml}. Finally, in Figure~\ref{fig:WML_2_emitter}, we plot the population plots at 19 evenly spaced times from the time interval [0, 3] ns.
    

    
    \begin{figure} [ht]
        \centering
        \includegraphics[width=\columnwidth]{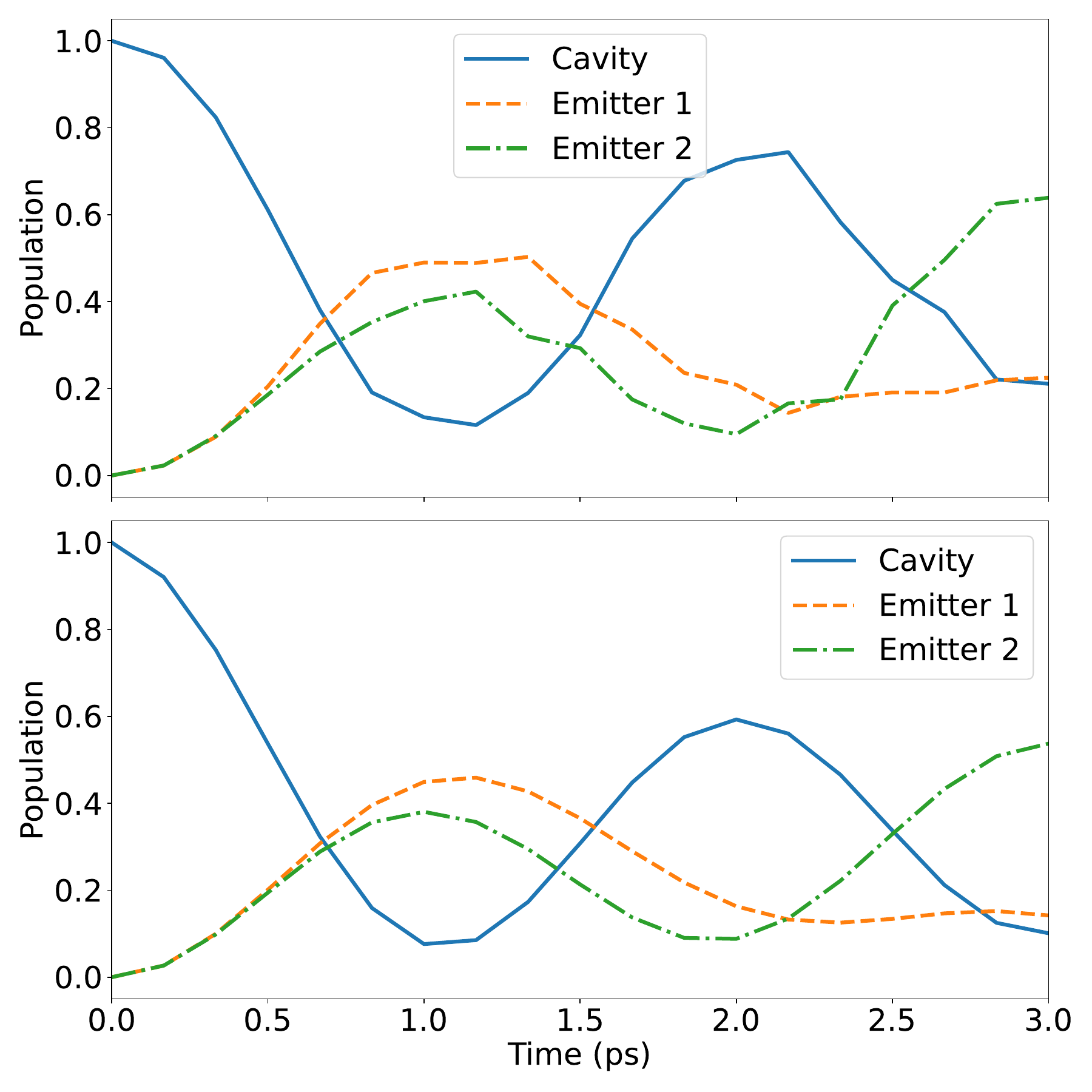}
        \caption{Population of an $N=2$ emitter system between $t = 0$ and $t = 3$ ns.
        The cavity frequency is $\omega_C = 245$ GHz and the emitter frequencies are $\omega_{E,i} - \omega_C = (0.4, 1.3)$ GHz.
        System parameters are $(\kappa, \gamma, g) = (160, 19.6, 1000)$ MHz.
        The top plot was generated by a hybrid algorithm run on the QASM simulator and the bottom plot by QuTiP.
        The simulation used 1000 shots, giving statistical shot noise of approximately $1/\sqrt{1000} \approx 0.03$.}
        \label{fig:WML_2_emitter}
    \end{figure}
    
    We next employ the hybrid algorithm to model the dynamics of a non-resonant $N=4$ emitter system. In this system, the emitter frequencies are $\omega_{E,i} - \omega_C = (0.2, .5, .75, 1)$ GHz, and the system parameters are $(\kappa, \gamma, g) = (160, 22.5, 800)$ MHz.
    The results of this simulation, between 0 and 2 ns, are shown in Figure~\ref{fig:WML_4_emitter}, where for generating the top plot, we evaluate the system populations at 11 evenly spaced times.
    
    \begin{figure}[ht]
        \centering
        \includegraphics[width=\columnwidth]{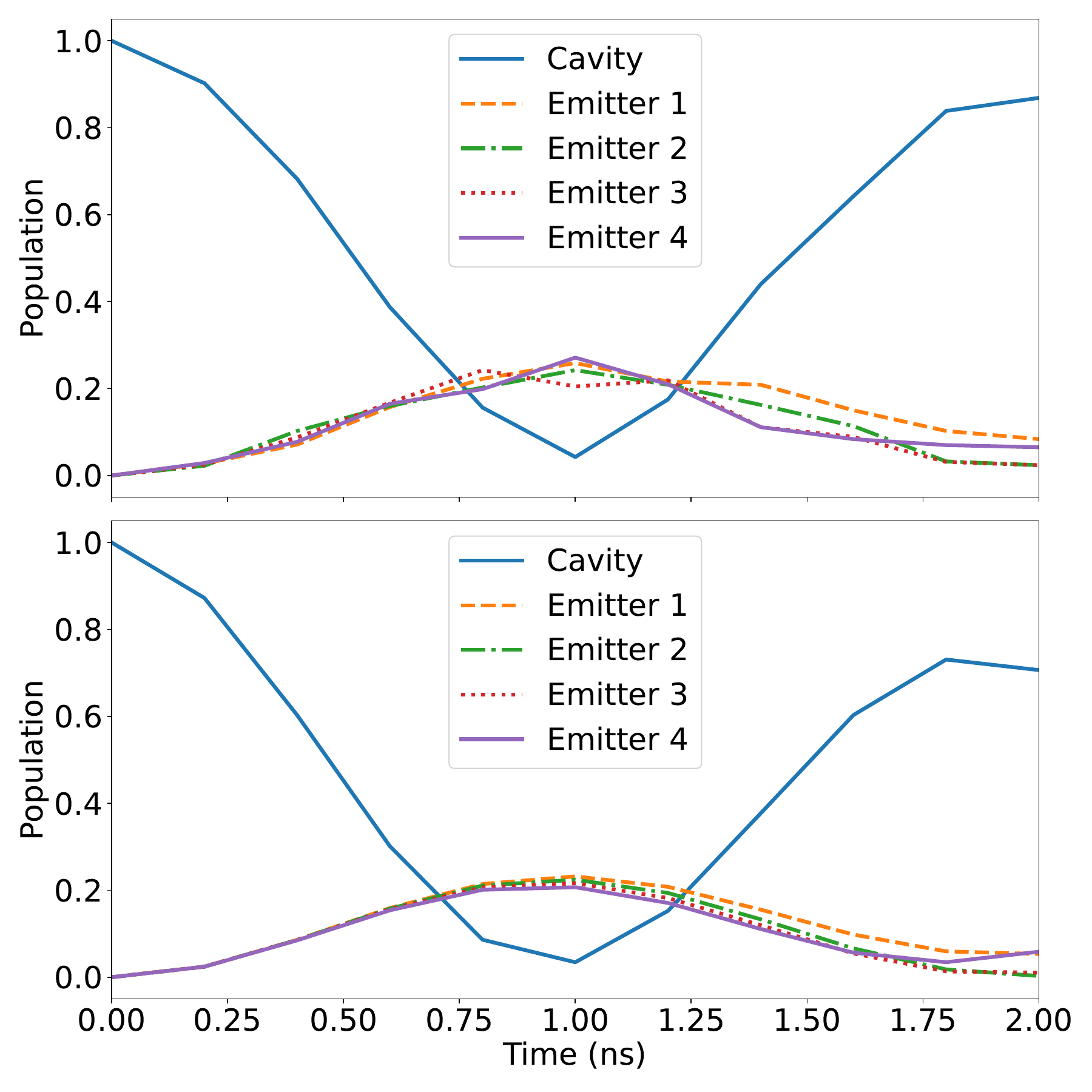}
        \caption{Population of an $N=4$ emitter system between $t=0$  and $t=2$ ns.
        The cavity frequency is $\omega_C = 245$ GHz and the emitter frequencies are $\omega_{E,i} - \omega_C = (0.2, .5, .75, 1)$ GHz, and the system parameters are $(\kappa, \gamma, g) = (160, 22.5, 800)$ MHz.
        The top plot was generated by a hybrid algorithm run on the QASM simulator and the bottom by QuTiP.
         The simulation used 1000 shots, giving statistical shot noise of approximately $1/\sqrt{1000} \approx 0.03$.}
        \label{fig:WML_4_emitter}
    \end{figure}

\subsection{\texorpdfstring{$g^{(2)}(0)$}{g2(0)} Coherence}

\label{subsec:g_2_coherence_plots}

    Estimating the $g^{(2)}(0)$ coherence of the cavity, as defined in~\eqref{eqn:g_2_coherence}, accurately is a challenging task when using a sampling algorithm. This is because, in the steady-state regime, the numerator ${\operatorname{Tr}[a^\dag a^\dag aa\rho]}$ and the denominator ${\operatorname{Tr}[a^\dag a\rho]}$ of~\eqref{eqn:g_2_coherence} tend to be very close to zero, and thus many samples are needed to sample sufficiently many non-zero values. Estimating  $g^{(2)}(0)$ by estimating its numerator and denominator separately requires numerous samples, and an estimate of the number of samples required to approximate quantities like $g^{(2)}(0)$ can be found in~\cite{FracHoeff}. In this paper, we use the median of means method~\cite{medianofmeans} to estimate  $g^{(2)}(0)$. Although the median of means method also requires that we separately estimate both the numerator and denominator of $g^{(2)}(0)$, it uses binning to obtain slightly better convergence. 
    
    We aim to approximate $g^{(2)}(0)$ within $0.1$ of the QuTiP value. To demonstrate that our simulations are able to approximate $g^{(2)}(0)$ within these bounds, we consider the following examples: A resonant single-emitter system, where $\omega_C = \omega_{E,1} = 245$ THz, $\kappa$ = 24.5 GHz, $\gamma$ = 0.4 GHz, and $g_1$ = 100 GHz, is initialized with one excitation and has attached to it a coherent drive of strength $E_P=\kappa/5$. The value for the $g^{(2)}(0)$ coherence, estimated using the classical solver of QuTiP, is 0.1895. Our estimate of $g^{(2)}(0)$, in Figure~\ref{fig:median_o_means_plot_1}, is within 0.1 of this value.

    
    \begin{figure} [ht]
        \centering
        \includegraphics[width=\columnwidth]{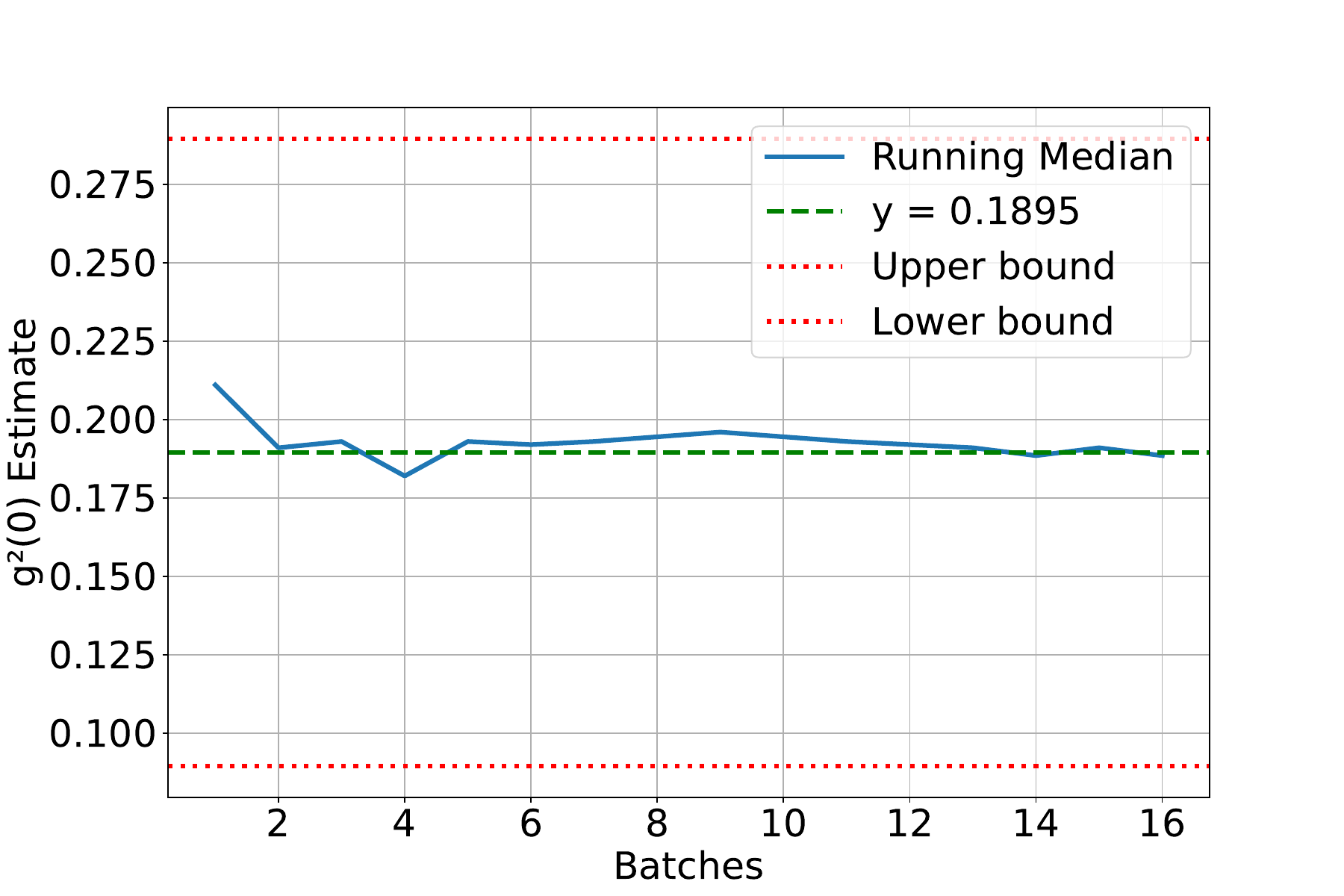}
        \caption{Coherence in a driven resonant cavity with a single emitter showing the running median estimate for the $g^{(2)}(0)$ coherence after each batch mean. The median of means approach is used to estimate the QuTiP value of 0.1895.
        The emitter is resonant with the cavity: $\omega_C = \omega_{E,1} = 245$ THz.
        The system parameters are $(\kappa, \gamma, g) = (24.5, 0.4, 100)$ GHz, and the cavity is subjected to a pump of strength $E_P = \kappa/5$. The plot was generated by the Split $J$-Matrix algorithm run on the QASM simulator with 1000 shots per batch. Based on QuTiP simulations, the error in a single $g^{(2)}$ simulation is about 0.19, so the expected statistical shot noise in each batch is about 0.19/$\sqrt{1000} \approx$ 0.006.
        }
        \label{fig:median_o_means_plot_1}
    \end{figure}

    Now, we demonstrate that the WML algorithm (Algorithm~\ref{algo:wml}) can also be used to estimate the $g^{(2)}(0)$ coherence of a non-resonant system with one emitter.
    The cavity frequency is $\omega_C = 245$ THz.
    The emitter frequency is $\omega_{E,1}-\omega_c = 180$ MHz and the system parameters are $\{\kappa, \gamma, g, E_P\} = \{1.8, 0.1, 0.2, \kappa/2\}$.
    For this system, we divide the total number of shots into 20 batches of 1500 shots each, and we plot the running median of these 20 batches in Figure~\ref{fig:median_o_means_plot_3}.

    \begin{figure} [ht]
        \centering
        \includegraphics[width=\linewidth]{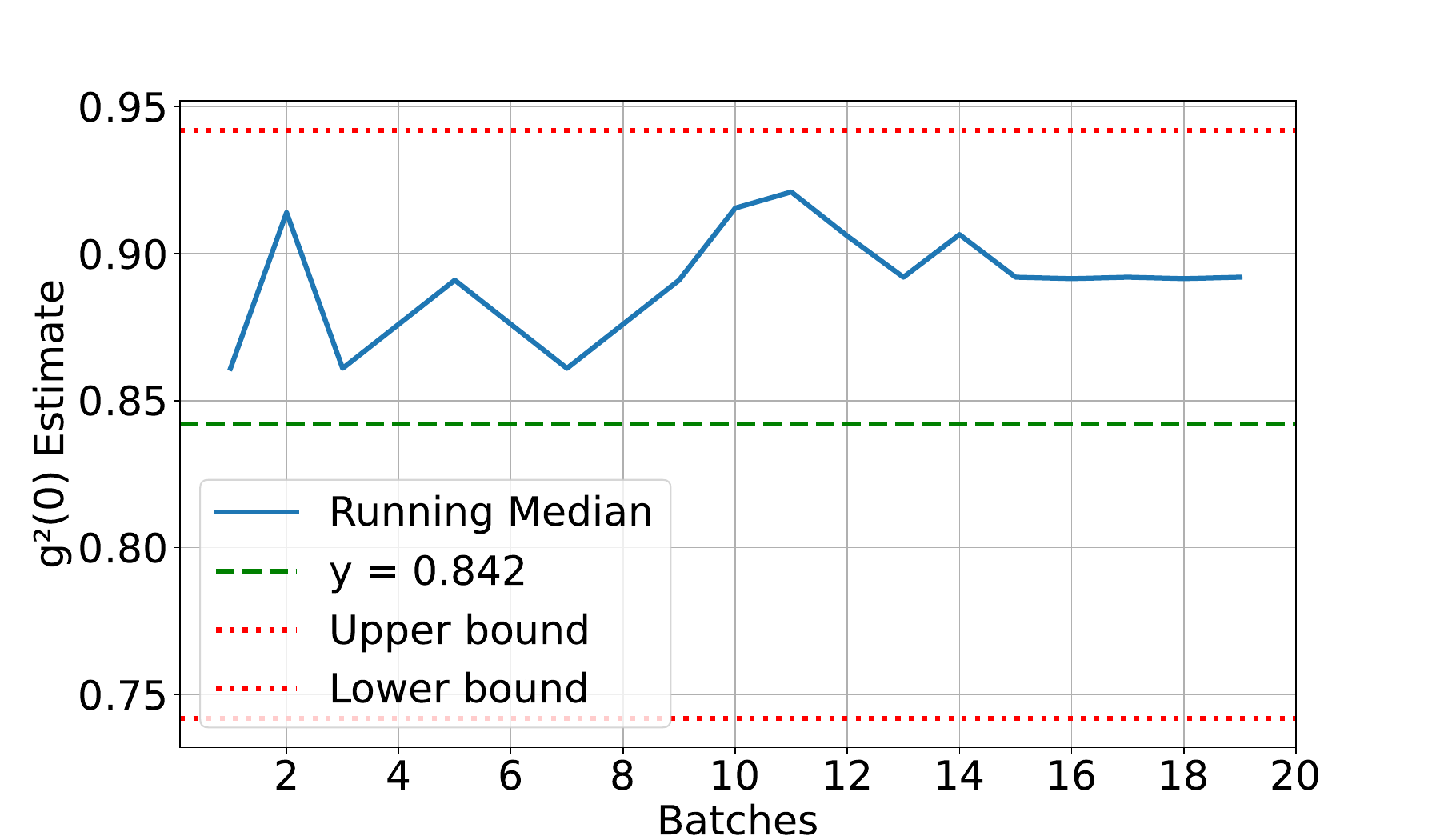}
        \caption{Coherence in a driven non-resonant cavity with a single emitter showing the running median estimate for the $g^{(2)}(0)$ coherence after each batch mean. The median of means approach is used to estimate the QuTiP value of 0.842.
        The emitter frequency is $\omega_{E,1}-\omega_c = 180$ MHz and the system parameters are $(\kappa, \gamma, g) = (1.8, 0.1, 0.2)$ GHz, and the cavity is subjected to a pump of strength $E_P = \kappa/2$. The plot was generated by the hybrid algorithm run on the QASM simulator using 1500 shots per batch. Based on QuTiP simulations, the error in a single $g^{(2)}$ simulation is about 0.04, so the expected statistical shot noise in each batch is about 0.04/$\sqrt{1500}\,\approx$\,0.001.}
        \label{fig:median_o_means_plot_3}
    \end{figure}
    
    Finally, consider a larger system with eight emitters, a size difficult to simulate numerically on a typical classical computer.
    We set the frequencies of these eight emitters as follows: $\omega_{E,i}-\omega_C = (20, 50, 75, 40, 15, 30, 57, 15)$ GHz.
    Furthermore, the other system parameters are $(\kappa, \gamma, g) = (2.83, 0.8, 10)$, and the cavity is subjected to a coherent drive of strength $E_P = \kappa/2$.
   For this system, we divide the total number of shots into 13 batches of 3000 shots each. We plot the running median of these 13 batches in Figure~\ref{fig:median_o_means_plot_2}; the median quickly converges to $g^{(2)}(0) \sim 0.867$.

    \begin{figure} [ht]
        \centering
        \includegraphics[width=\columnwidth]{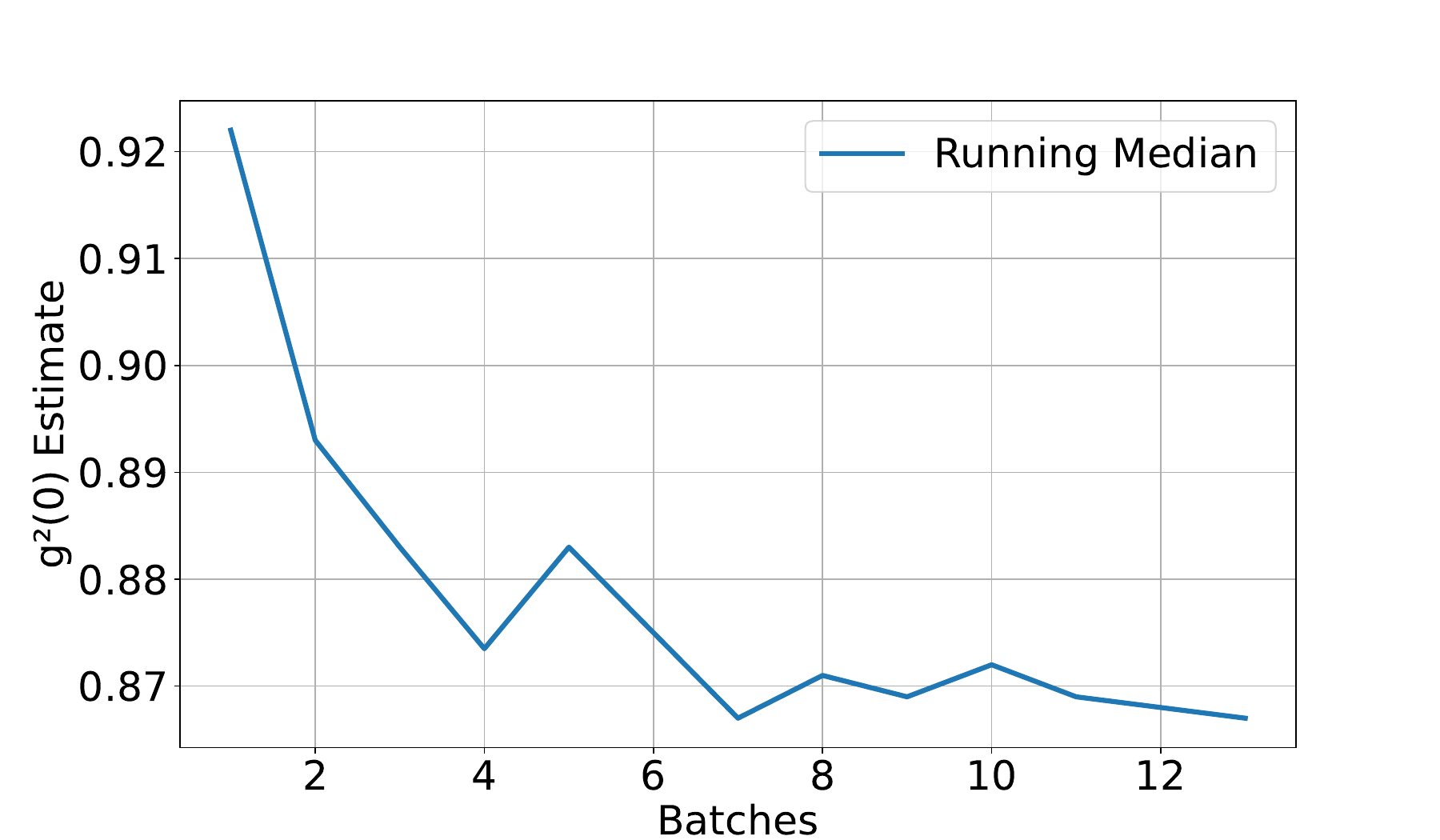}
        \caption{Coherence in a driven inhomogeneous systems of a cavity with eight emitters showing the running estimate for the $g^{(2)}(0)$ coherence after each batch mean.
        Emitters frequencies are $\omega_{E,i}-\omega_C = (20, 50, 75, 40, 15, 30, 57, 15)$ GHz.
        The system parameters are $(\kappa, \gamma, g) = (2.83, 0.8, 10)$ GHz, and the cavity is subjected to a pump of strength $E_P = \kappa/2$. The plot was generated by the Split $J$-Matrix algorithm run on the QASM simulator with 3000 shots per batch. Because this large system is difficult to simulate by typical classical means with QuTiP, we do not estimate the shot noise.}
        \label{fig:median_o_means_plot_2}
    \end{figure}

\section{Discussion}

\label{sec:Discussion}

    In this section, we discuss important considerations when deciding which of the algorithms, i.e., the $J$-Matrix algorithm (see Section~\ref{subsec:background_j_matrix}), the Split $J$-Matrix algorithm (see Section~\ref{sec:methods_split_j_matrix}), and the WML algorithm~\cite{WMLpartI,WMLpartII}, to use when simulating a system of interest using a quantum computer.

    The choice between the $J$-Matrix algorithm and the Split $J$-Matrix algorithm depends on the size and number of Lindblad operators in the system of interest. The standard $J$-Matrix algorithm may be better for situations where the Lindblad operators are not local operators and thus the matrix encoding these operators cannot be split into smaller operators acting on subsystems. The Split $J$-Matrix algorithm is well suited for systems with multiple Lindblad operators, each acting on a constant number of qubits, like in the open TC model. 
    
    The main difference between the WML~\cite{WMLpartI,WMLpartII} and the Split $J$-Matrix algorithm is the input model. WML assumes sample access to program states that encode the Hamiltonian and Lindblad operators of a given model. This is an easy assumption to satisfy if the program states required can be efficiently prepared. On the other hand, if we are provided with classical descriptions of the Lindblad operators, then we can use the Split $J$-Matrix algorithm since we can obtain a classical description of the unitary operators required.
    \textcolor{green}{Another difference between the two algorithms is that the WML algorithm's runtime is sensitive to the quantity $c$ (defined in \eqref{algo:wml}).
    As described in Sec.~\ref{sec:Results}, this means that systems featuring many emitters with different frequencies are more expensive to simulate than systems with few emitters, or emitters with lower frequency variation.
    In contrast, while the simulation time depends on system size for the Split $J$-Matrix method, its performance degrades much more gradually as these factors increase.}

    
    In the context of the open TC model, we discuss gates involving three or more qubits in each algorithm, which must be decomposed into one- and two-qubit gates. The decomposition is specific to the quantum computer in question. Several auxiliary qubits are used in the LCU-based WML algorithm to simulate the fixed interaction. This subroutine uses a series of controlled-unitary gates. In order to apply these gates, four control qubits and three target qubits are used, for a total of seven qubits. On the other hand, the largest gates in the Split $J$-Matrix algorithm act on only three qubits  (i.e., the cavity $J$-Matrix gate and the cavity-emitter interaction Hamiltonian gate). 

    
    The simulations presented in this work can handle up to three excitations in the cavity. This can be achieved by changing the way the Pauli-$X$ gates are applied to initialize the cavity qubits. To simulate $R$ excitations in the cavity, our techniques require $\log_2(R+1)$ qubits corresponding to the cavity. The dimension of the cavity annihilation operator increases linearly in $R$, and the constants $c$ and $\lambda_{\max}$ scale quadratically in $R$, as shown in Appendix~\ref{app:constants}. The gate complexity of both the WML and Split $J$-Matrix algorithms then scale according to Theorem~\ref{thm:gate-comp-WML} and Theorem~\ref{thm:split_j_matrix}, respectively. Additionally, if we assume that the number of emitters, $N$, is proportional to the number of excitations in the cavity, $R$, we can see that our techniques scale polynomially with the number of excitations. This is a marked improvement over commonly used classical simulation techniques, which scale exponentially with number of excitations (see Appendix~\ref{app:classical-complexity} for details). The final modification we should consider is when emitters can hold more than a single excitation at a time, e.g., modeling three-level atoms. This would  require using multiple qubits per emitter. The gate complexities then scale according to the dimension of annihilation operators of the emitters. The number of qubits required in these algorithms scales linearly in $N$, another substantial improvement over the exponential scaling of typical solvers.

\section{Conclusion}

\label{sec:Conclusion}

    The key contributions of our paper are three-fold. First, we implemented two open quantum simulation algorithms---the WML algorithm and the Split $J$-Matrix algorithm---to model the behavior of the open TC model. Our results show that these quantum algorithms can model open model dynamics accurately. Furthermore, our \textcolor{green}{theoretical findings broaden the parameter regimes for simulating the open TC model (non-resonant and inhomogeneous) using our quantum algorithms.} Second, we proposed two efficient LCU-based protocols for implementing the fixed interaction channel of the WML algorithm. This resolves one of the key open questions of prior studies~\cite{WMLpartI, WMLpartII}. Third, we investigated the gate complexity of our algorithms. We discovered that the gate complexities of the WML and Split $J$-Matrix algorithms scale quadratically and cubically with respect to the number of emitters in the system, respectively,  while the number of qubits scales linearly as $O(N)$.
    
    Looking ahead, one open question is how the simulations would perform if the algorithms were run on hardware with bosonic modes instead of qubits. Finally, it would be interesting to extend these algorithms to model the Tavis--Cummings--Hubbard model, in which multiple cavities are coupled to each other and to emitters.

\section*{Author Contributions}

\noindent
\textbf{Author Contributions}: The following describes the
different contributions of the authors of this work, using
roles defined by the CRediT (Contributor Roles Taxonomy) project~\cite{CRediT}:

\noindent\textbf{AS:} Formal analysis, Investigation, Methodology, Project administration, Software, Visualization, Writing - original draft, Writing - review \& editing.

\noindent\textbf{DP:} Conceptualization, Formal Analysis, Investigation, Methodology, Project Administration, Supervision, Validation, Visualization, Writing – original draft, Writing – review \& editing

\noindent\textbf{AP:} Formal analysis, Supervision, Methodology, Investigation, Writing – original draft, Writing – review \& editing

\noindent\textbf{AR:} Methodology, Writing - original draft, Writing - review \& editing.

\noindent\textbf{RB:} Methodology, Software, Writing - original draft, Writing - review \& editing.

\noindent\textbf{MR:} Conceptualization, Methodology, Funding acquisition, Writing – review \& editing.

\noindent\textbf{MMW:} Conceptualization, Formal Analysis, Funding acquisition, Methodology,  Validation,  Writing – review \& editing.
    
\begin{acknowledgments}
    We  thank Valla Fatemi for helpful discussions at Cornell Quantum Day.
    
    AP acknowledges support from the National Science Centre Poland (Grant~No.~2022/46/E/ST2/00115). MR acknowledges support from NSF CAREER (Award 2047564). DP and MMW acknowledge support from
    the Air Force Office of Scientific Research  under agreement no.~FA2386-24-1-4069.
    The U.S.~Government is authorized to reproduce and
    distribute reprints for Governmental purposes notwithstanding any copyright
    notation thereon. The views and conclusions contained herein are those of the
    authors and should not be interpreted as necessarily representing the official
    policies or endorsements, either expressed or implied, of the United States Air Force.

\end{acknowledgments}

\bibliography{ref}

\appendix
\onecolumngrid


\section{Decomposition of the Lindblad Operator~\texorpdfstring{$M$}{M}} 

\label{app:WML_fixed_inter_unitaries}

    In Section~\ref{sec:wml-practical}, we presented a decomposition of the Lindblad operator $M$, defined in~\eqref{eqn:fixed_interaction_WML}, as a linear combination of unitaries for the case of the system of interest being a single qubit. In this appendix, we extend this to the case where the system of interest consists of multiple qubits.
    
    To begin with, let $A$ denote the system register consisting of $q$ qubits. Similarly, let $B$ and $C$ jointly denote the program register, each consisting of $q$ qubits. As stated previously in~\eqref{eq:intro-M}, the Lindbladian $\mathcal{M}$ acts jointly on the system register and program register, and it consists of a single Lindblad operator $M$, which is given as follows:
    \begin{equation}
        M = \frac{1}{2^{q/2}}\left(I_{A}\otimes |\Gamma\rangle\!\langle\Gamma|_{BC}\right)\left(\operatorname{SWAP}_{AB}\otimes I_C\right).
    \end{equation}
    Now, consider the fact that the multi-qubit operators $I_{A}, |\Gamma\rangle\!\langle\Gamma|_{BC},$ and $\operatorname{SWAP}_{AB}$ can be decomposed as tensor products of operators that each act on only one or two qubits:
    \begin{align}
        I_{A} & = I_{A_1} \otimes I_{A_2} \otimes \cdots \otimes I_{A_q},\\
        |\Gamma\rangle\!\langle\Gamma|_{BC} & = |\Gamma\rangle\!\langle\Gamma|_{B_1C_1} \otimes |\Gamma\rangle\!\langle\Gamma|_{B_2C_2} \otimes \cdots \otimes |\Gamma\rangle\!\langle\Gamma|_{B_qC_q} ,\\
        \operatorname{SWAP}_{AB} & = \operatorname{SWAP}_{A_1B_1} \otimes \operatorname{SWAP}_{A_2B_2} \otimes \cdots \otimes \operatorname{SWAP}_{A_qB_q},
    \end{align}
    where the registers $\{A_i\}_{i\in[q]}$, $\{B_i\}_{i\in[q]}$, and $\{C_i\}_{i\in[q]}$ are all single-qubit registers and $A \coloneqq A_1 \otimes \cdots \otimes A_q $, $B \coloneqq B_1 \otimes \cdots \otimes B_q$, and  $C \coloneqq C_1 \otimes \cdots \otimes C_q$. Using the above equalities, we decompose $M$ as follows:
    \begin{align}
        M & = \underbrace{\frac{1}{2^{1/2}}\left(I_{A_1}\otimes |\Gamma\rangle\!\langle\Gamma|_{B_1C_1}\right)\left(\operatorname{SWAP}_{A_1B_1}\otimes I_{C_1}\right)}_{\eqqcolon M_1} \otimes \underbrace{\frac{1}{2^{1/2}}\left(I_{A_2}\otimes |\Gamma\rangle\!\langle\Gamma|_{B_2C_2}\right)\left(\operatorname{SWAP}_{A_2B_2}\otimes I_{C_2}\right)}_{\eqqcolon M_2}\notag\\
        & \qquad \otimes \cdots \otimes \underbrace{\frac{1}{2^{1/2}}\left(I_{A_q}\otimes |\Gamma\rangle\!\langle\Gamma|_{B_qC_q}\right)\left(\operatorname{SWAP}_{A_qB_q}\otimes I_{C_q}\right) }_{\eqqcolon M_q}\\
        & = M_1 \otimes M_2 \otimes \cdots \otimes M_q.\label{eq:decomp-M-M_i}
    \end{align}
    
    It is straightforward to see that we can obtain the linear-combination expression for $M$ if we can get the linear-combination expression for each $M_i$; therefore, we now focus on obtaining the latter. For all $i\in [q]$, $\operatorname{SWAP}_{A_iB_i}$ and $|\Gamma\rangle\!\langle\Gamma|_{B_iC_i}$ can be written in terms of Pauli strings as follows:
    \begin{align}
        \operatorname{SWAP}_{A_iB_i}&= \textcolor{green}{\frac{1}{2} \left(\right.}I_{A_i} \otimes I_{B_i}+ X_{A_i}\otimes X_{B_i}+Y_{A_i}\otimes Y_{B_i}+Z_{A_i}\otimes Z_{B_i} \textcolor{green}{\left.\right)},\\
        |\Gamma\rangle\!\langle\Gamma|_{B_iC_i} &= \textcolor{green}{\frac{1}{2} \left(\right.}I_{B_i}\otimes I_{C_i}+ X_{B_i}\otimes X_{C_i}-Y_{B_i}\otimes Y_{C_i}+Z_{B_i}\otimes Z_{C_i}\textcolor{green}{\left.\right)}.
    \end{align}
    Ignoring the system labels for simplicity, we can rewrite each $M_i$ as follows:
    \begin{align}
        M_i & = \frac{1}{\textcolor{green}{4\cdot}2^{1/2}}\Big(\left.I\otimes I\otimes I + X\otimes X\otimes I + Y\otimes Y \otimes I \right.+ Z\otimes Z\otimes I +I\otimes X\otimes X + X\otimes I\otimes X \notag \\
        &\quad + Y\otimes iZ \otimes X - Z\otimes iY\otimes X -I\otimes Y\otimes Y + X\otimes iZ\otimes Y - Y\otimes I\otimes Y - Z\otimes iX\otimes Y \notag \\
        &\quad + I\otimes Z\otimes Z + X\otimes iY\otimes Z - Y\otimes iX \otimes Z+ Z\otimes I \otimes Z \Big).\label{eq:M_i-exp-full}
    \end{align}
    
    Observe that there are 16 terms in the above linear-combination expression. This implies that there are $16^q$ or $2^{4q}$ terms in the linear-combination expression for $M$.

\section{Wave Matrix Lindbladization \texorpdfstring{$e^{\mathcal{M}\Delta}$}{exp(M Delta)} Channel Protocol 1}~\label{app:WML_channel}

    In this section, we outline a protocol for the implementation of the fixed interaction channel $e^{\mathcal{M}\Delta}$ of the WML algorithm using symmetries inherent to the operator $M$. For simplicity, the following protocol details the steps for the implementation of $e^{\mathcal{M}\Delta}$ channel with three qubits, $A$, $B$, and $C$, as input, where $A$ is the system of interest, and $BC$ contains the program state. We detail how we employ the LCU-based algorithm proposed in~\cite{cleve2019efficient} to realize this fixed interaction. Using the LCU method, we can produce a quantum map $\mathcal{M}_\Delta$ that approximates $e^{\mathcal{M}\Delta}$, where $\mathcal{M}_\Delta$ can be written in the following manner:
    \begin{equation}\label{eqn:lcu_approx_wml_appendix}
        \mathcal{M}_\Delta(O) = \sum^1_{j=0}A_j O A_j^\dagger,
    \end{equation}
    with $A_0 = I -\frac{\Delta}{2}M^\dagger M$ and $A_1 = \sqrt{\Delta} M$. Since we will express $M$ and $M^\dagger M$ as a linear combination of unitaries, we can use the LCU method~\cite{cleve2019efficient} to implement this approximation. First, recall the definition of $M$ from~\eqref{eqn:fixed_interaction_WML}
    \begin{align}
        M  & =\frac{1}{\sqrt{2}}\left(  I_{A}\otimes|\Gamma\rangle\!\langle\Gamma
        |_{BC}\right)  \left(  \text{SWAP}_{AB}\otimes I_{C}\right)  \\
        & =\frac{1}{\textcolor{green}{4}\sqrt{2}}\left(  I_{A}\otimes I_{B}\otimes I_{C}+I_{A}\otimes
        X_{B}\otimes X_{C}+I_{A}\otimes Z_{B}\otimes Z_{C}-I_{A}\otimes Y_{B}\otimes
        Y_{C}\right)  \times\nonumber\\
        & \qquad\left(  I_{A}\otimes I_{B}\otimes I_{C}+X_{A}\otimes X_{B}\otimes
        I_{C}+Z_{A}\otimes Z_{B}\otimes I_{C}+Y_{A}\otimes Y_{B}\otimes I_{C}\right)
        \\
        & =\frac{1}{\textcolor{green}{4}\sqrt{2}}\sum_{i,j,k,\ell\in\{0,1\}}\left(  Z_{B}^{i}X_{B}^{j}\otimes
        Z_{C}^{i}X_{C}^{j}\right)  \left(  -1\right)  ^{k\cdot\ell}\left(  Z_{A}%
        ^{k}X_{A}^{\ell}\otimes Z_{B}^{k}X_{B}^{\ell}\right)  .
    \end{align}
    Then it follows that
    \begin{align}
        M^{\dag}M  & =\left(\text{SWAP}_{AB}\otimes I_{C}\right)  \frac{1}{\sqrt{2}}\left(  I_{A}\otimes|\Gamma\rangle\!\langle\Gamma|_{BC}\right)  \frac{1}{\sqrt{2}}\left(  I_{A}\otimes|\Gamma\rangle\!\langle\Gamma|_{BC}\right)\left(\text{SWAP}_{AB}\otimes I_{C}\right)  \\
        & =\left(  \text{SWAP}_{AB}\otimes I_{C}\right)  \left(  I_{A}\otimes
        |\Gamma\rangle\!\langle\Gamma|_{BC}\right)  \left(  \text{SWAP}_{AB}\otimes
        I_{C}\right)  \\
        & =|\Gamma\rangle\!\langle\Gamma|_{AC}\otimes I_{B}\\
        & = \textcolor{green}{\frac{1}{2}\left(\right.} I_{A}\otimes I_{B}\otimes I_{C}+X_{A}\otimes I_{B}\otimes X_{C}%
        +Z_{A}\otimes I_{B}\otimes Z_{C}-Y_{A}\otimes I_{B}\otimes Y_{C} \textcolor{green}{\left.\right)}\\
        & = \textcolor{green}{\frac{1}{2}} \sum_{k,\ell \in \{0,1\}}\left(  Z_{A}^{k}X_{A}^{\ell}\otimes Z_{C}^{k}X_{C}^{\ell
        }\right)  .
    \end{align}
    We would like to implement a quantum channel that has the following two Kraus operators to approximate $e^{\mathcal{M}\Delta}$:
    \begin{align} 
        A_{0}  & =I-\frac{\Delta}{2}M^{\dag}M,\\
        A_{1}  & =\sqrt{\Delta}M.
    \end{align}
    To do so, we can use linear combination of unitaries (LCU) methods~\cite{cleve2019efficient}. Consider that the first Kraus operator can be written as%
    \begin{align}
        A_{0}  & =I_{A}\otimes I_{B}\otimes I_{C}-\frac{\Delta}{\textcolor{green}{4}}\left(  I_{A}\otimes I_{B}\otimes I_{C}+X_{A}\otimes I_{B}\otimes X_{C}+Z_{A}\otimes I_{B}\otimes Z_{C}-Y_{A}\otimes I_{B}\otimes Y_{C}\right)  \\
        & =I_{A}\otimes I_{B}\otimes I_{C}+\frac{\Delta}{\textcolor{green}{4}}\left(  -I_{A}\otimes I_{B}\otimes I_{C}-X_{A}\otimes I_{B}\otimes X_{C}-Z_{A}\otimes I_{B}\otimes Z_{C}+Y_{A}\otimes I_{B}\otimes Y_{C}\right)  .
    \end{align}

\begin{figure}
    \centering
    \includegraphics[width=0.9\linewidth]{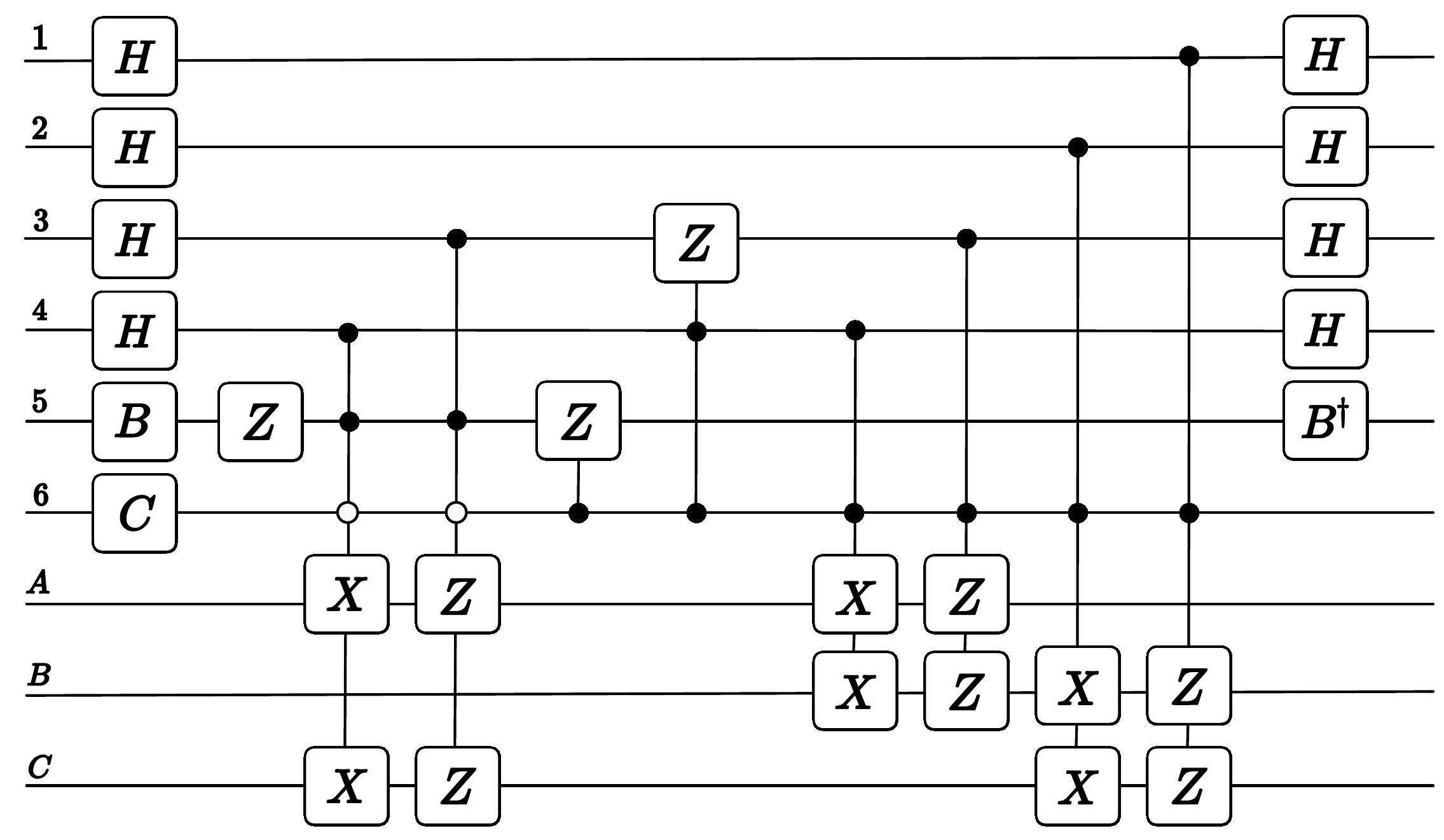}
    \caption{Circuit diagram for Protocol~1 for approximately implementing the channel $e^{\mathcal{M}\Delta}$.}
    \label{fig:protocol-1}
\end{figure}

    An LCU algorithm for implementing $e^{\mathcal{M}\Delta}$ is depicted in Figure~\ref{fig:protocol-1}. Let us verify that the constructed circuit in Figure~\ref{fig:protocol-1} is indeed correct. We define unitaries $B$ and $C$ as follows: 
    
    \begin{align}
        \textcolor{green}{B_5|0\rangle} & = \textcolor{green}{\frac{1}{\sqrt{1+\Delta}}\left(|0\rangle+\sqrt{\Delta}|1\rangle\right)},\\
        \textcolor{green}{C_6|0\rangle} & = \textcolor{green}{\frac{1}{\sqrt{(1+\Delta)^2+8\Delta}}\left((1+\Delta)|0\rangle+ 2\sqrt{2\Delta}|1\rangle\right)}.
    \end{align}
    The remaining gates, in order of their appearance, are defined as 
    \begin{align}
        C_{456}X_{A}X_{C} &  \coloneqq |1\rangle\!\langle1|_{4}\otimes|1\rangle\!\langle1|_{5}\otimes|0\rangle\!\langle 0|_{6}\otimes X_{A}\otimes X_{C} +\left(I_{456}-|1\rangle\!\langle1|_{4}\otimes|1\rangle\!\langle1|_{5}\otimes|0\rangle\!\langle 0|_{6}\right)  \otimes I_{A}\otimes I_{B},\\
        C_{356}Z_{A}Z_{C} &  \coloneqq |1\rangle\!\langle1|_{3}\otimes|1\rangle\!\langle1|_{5}\otimes|0\rangle\!\langle 0|_{6}\otimes Z_{A}\otimes Z_{C} +\left(  I_{356}-|1\rangle\!\langle1|_{3}\otimes|1\rangle\!\langle1|_{5}\otimes|0\rangle\!\langle 0|_{6}\right)  \otimes I_{A}\otimes I_{B},\\
        C_{46}Z_3 &  \coloneqq |1\rangle\!\langle1|_{4}\otimes|1\rangle\!\langle 1|_{6}\otimes Z_3 +\left(I_{46}-|1\rangle\!\langle1|_{4}\otimes|1\rangle\!\langle 1|_{6}\right)  \otimes I_{3},\\
        C_{6}Z_5 &  \coloneqq |1\rangle\!\langle 1|_{6}\otimes Z_5 +\left(I_{6}-|1\rangle\!\langle1|_{6}\right)  \otimes I_{5},\\
        C_{46}X_{A}X_{B} &  \coloneqq |1\rangle\!\langle1|_{4}\otimes|1\rangle\!\langle 1|_{6}\otimes X_{A}\otimes X_{B} +\left(I_{46}-|1\rangle\!\langle1|_{4}\otimes|1\rangle\!\langle 1|_{6}\right)  \otimes I_{A}\otimes I_{B},\\
        C_{36}Z_{A}Z_{B} &  \coloneqq |1\rangle\!\langle1|_{3}\otimes|1\rangle\!\langle 1|_{6}\otimes Z_{A}\otimes Z_{B} +\left(  I_{36}-|1\rangle\!\langle1|_{3}\otimes|1\rangle\!\langle 1|_{6}\right)  \otimes I_{A}\otimes I_{B},\\
        C_{16}X_{B}X_{C} &  \coloneqq |1\rangle\!\langle1|_{1}\otimes|1\rangle\!\langle 1|_{6}\otimes X_{B}\otimes X_{C} +\left(I_{16}-|1\rangle\!\langle1|_{1}\otimes|1\rangle\!\langle 1|_{6}\right)  \otimes I_{B}\otimes I_{C},\\
        C_{26}Z_{B}Z_{C} &  \coloneqq |1\rangle\!\langle1|_{2}\otimes|1\rangle\!\langle 1|_{6}\otimes Z_{B}\otimes Z_{C} +\left(  I_{16}-|1\rangle\!\langle1|_{2}\otimes|1\rangle\!\langle 1|_{6}\right)  \otimes I_{B}\otimes I_{C}.
    \end{align}
    The state after the application of the gates $H_1 H_2 H_3 H_4 Z_5 B_5 C_6$ is%
    \begin{align}
        & |+\rangle_{1}|+\rangle_{2}|+\rangle_{3}|+\rangle_{4}\left(  |0\rangle
        _{5}-\textcolor{green}{\sqrt{\Delta}}|1\rangle_{5}\right)  \left(  \left(  1+\textcolor{green}{\Delta}\right)
        |0\rangle_{6}+\textcolor{green}{2\sqrt{2\Delta}}|1\rangle_{6}\right)  |\psi\rangle_{ABC}%
        \nonumber\\
        & =\left(1+\textcolor{green}{\Delta}\right)  |+\rangle_{1}|+\rangle_{2}|+\rangle_{3}%
        |+\rangle_{4}|0\rangle_{5}|0\rangle_{6}|\psi\rangle_{ABC} +\textcolor{green}{2\sqrt{2\Delta}}|+\rangle_{1}|+\rangle_{2}|+\rangle_{3}|+\rangle_{4}\left(  |0\rangle_{5}-\textcolor{green}{\sqrt{\Delta}}|1\rangle_{5}\right)  |1\rangle_{6}|\psi\rangle_{ABC}\nonumber\\
        & \qquad-\textcolor{green}{\sqrt{\Delta}\left(  1+\Delta\right)}  |+\rangle_{1}|+\rangle
        _{2}|+\rangle_{3}|+\rangle_{4}|1\rangle_{5}|0\rangle_{6}|\psi\rangle_{ABC}.
    \end{align}
    The first term in the superposition is never modified by the circuit. We
    just need to handle the other terms. Let
    \begin{equation}
        |\Delta_{-}\rangle_{5}\equiv|0\rangle_{5}-\textcolor{green}{\sqrt{\Delta}}|1\rangle_{5}.
    \end{equation}
    The second term can be written as follows%
    \begin{align}
        & |+\rangle_{1}|+\rangle_{2}|+\rangle_{3}|+\rangle_{4}|\Delta_{-}\rangle
        _{5}|1\rangle_{6}|\psi\rangle_{ABC}=\frac{1}{4}\sum_{i,j,k,\ell}|i\rangle_{1}|j\rangle_{2}|k\rangle_{3}%
        |\ell\rangle_{4}|\Delta_{-}\rangle_{5}|1\rangle_{6}|\psi\rangle_{ABC}.
    \end{align}
    After the application of $C_{46}Z_3$, we get
    \begin{align}
        & \rightarrow\frac{1}{4}\sum_{i,j,k,\ell}\left(  -1\right)  ^{k\cdot\ell}|i\rangle_{1}|j\rangle_{2}|k\rangle_{3}|\ell\rangle_{4}|\Delta_{-}\rangle_{5}|1\rangle_{6}|\psi\rangle_{ABC}.
    \end{align}
    After the application of $C_{46}X_A X_B$ and $C_{36}Z_A Z_B$, we get
    \begin{align}
        &\rightarrow\frac{1}{4}\sum_{i,j,k,\ell}\left(  -1\right)  ^{k\cdot\ell}|i\rangle_{1}|j\rangle_{2}|k\rangle_{3}|\ell\rangle_{4}|\Delta_{-}\rangle_{5}|1\rangle_{6}\left(  Z_{A}^{k}X_{A}^{\ell}\otimes Z_{B}^{k}X_{B}^{\ell}\right)  |\psi\rangle_{ABC}
    \end{align}
    After the application of $C_{26}X_B X_C$ and $C_{16}Z_B Z_C$, we get
    \begin{align}
        &\rightarrow\frac{1}{4}\sum_{i,j,k,\ell}|i\rangle_{1}|j\rangle_{2}|k\rangle_{3}|\ell\rangle_{4}|\Delta_{-}\rangle_{5}|1\rangle_{6}\left(Z_{B}^{i}X_{B}^{j}\otimes Z_{C}^{i}X_{C}^{j}\right)\left(-1\right)^{k\cdot\ell} \left(Z_{A}^{k}X_{A}^{\ell}\otimes Z_{B}^{k}X_{B}^{\ell}\right)  |\psi\rangle_{ABC}.
    \end{align}
    The third term can be expressed as follows
    \begin{align}
        &|+\rangle_{1}|+\rangle_{2}|+\rangle_{3}|+\rangle_{4}|1\rangle_{5}|0\rangle_{6}|\psi\rangle_{ABC}=|+\rangle_{1}|+\rangle_{2}\frac{1}{2}\sum_{k,\ell}|k\rangle_{3}|\ell\rangle_{4}|1\rangle_{5}|0\rangle_{6}|\psi\rangle_{ABC}.
    \end{align}
    After the application of $C_{456}X_A X_C$ and $C_{356}Z_A Z_C$, we get
    \begin{align}
        &\rightarrow|+\rangle_{1}|+\rangle_{2}\frac{1}{2}\sum_{k,\ell}|k\rangle_{3}|\ell\rangle_{4}|1\rangle_{5}|0\rangle_{6}\left(  Z_{A}^{k}X_{A}^{\ell}\otimes Z_{C}^{k}X_{C}^{\ell}\right)  |\psi\rangle_{ABC}.
    \end{align}
    So this means that the overall superposition goes to
    \begin{align}
        & =\left(  1+\textcolor{green}{\Delta}\right)  |+\rangle_{1}|+\rangle_{2}|+\rangle_{3}|+\rangle_{4}|0\rangle_{5}|0\rangle_{6}|\psi\rangle_{ABC}\nonumber\\
        &\qquad+\textcolor{green}{2\sqrt{2\Delta}}\frac{1}{4}\sum_{i,j,k,\ell}|i\rangle_{1}|j\rangle_{2}|k\rangle_{3}|\ell\rangle_{4}|\Delta_{-}\rangle_{5}|1\rangle_{6}\left(Z_{B}^{i}X_{B}^{j}\otimes Z_{C}^{i}X_{C}^{j}\right)  \left(  -1\right)^{k\cdot\ell}\left(  Z_{A}^{k}X_{A}^{\ell}\otimes Z_{B}^{k}X_{B}^{\ell}\right)  |\psi\rangle_{ABC}\nonumber\\
        &\qquad-\textcolor{green}{\sqrt{\Delta}\left(  1+\Delta\right)}  |+\rangle_{1}|+\rangle_{2}\frac{1}{2}\sum_{k,\ell}|k\rangle_{3}|\ell\rangle_{4}|1\rangle_{5}|0\rangle_{6}\left(  Z_{A}^{k}X_{A}^{\ell}\otimes Z_{C}^{k}X_{C}^{\ell}\right)  |\psi\rangle_{ABC}\\
        &=\left(  1+\textcolor{green}{\Delta}\right)  |+\rangle_{1}|+\rangle_{2}|+\rangle_{3}|+\rangle_{4}|0\rangle_{5}|0\rangle_{6}|\psi\rangle_{ABC}\nonumber\\
        &\qquad+\textcolor{green}{\sqrt{\frac{\Delta}{2}}}\sum_{i,j,k,\ell}|i\rangle_{1}|j\rangle_{2}|k\rangle_{3}|\ell\rangle_{4}|\Delta_{-}\rangle_{5}|1\rangle_{6}\left(Z_{B}^{i}X_{B}^{j}\otimes Z_{C}^{i}X_{C}^{j}\right)  \left(  -1\right)^{k\cdot\ell}\left(  Z_{A}^{k}X_{A}^{\ell}\otimes Z_{B}^{k}X_{B}^{\ell}\right)  |\psi\rangle_{ABC}\nonumber\\
        &\qquad-\textcolor{green}{\frac{\sqrt{\Delta}}{2}\left(  1+\Delta\right)} |+\rangle_{1}|+\rangle_{2}\sum_{k,\ell}|k\rangle_{3}|\ell\rangle_{4}|1\rangle_{5}|0\rangle_{6}\left(  Z_{A}^{k}X_{A}^{\ell}\otimes Z_{C}^{k}X_{C}^{\ell}\right)  |\psi\rangle_{ABC}.
    \end{align}
    Now applying the projection onto $\langle+|_{1}\langle+|_{2}\langle+|_{3}\langle+|_{4}$, the state then becomes
    \begin{multline}
        \left(  1+\textcolor{green}{\Delta}\right)  |0\rangle_{5}|0\rangle_{6}|\psi\rangle_{ABC}%
        +\textcolor{green}{\sqrt{\frac{\Delta}{32}}}|\Delta_{-}\rangle_{5}|1\rangle_{6}\sum_{i,j,k,\ell}\left(
        Z_{B}^{i}X_{B}^{j}\otimes Z_{C}^{i}X_{C}^{j}\right)  \left(  -1\right)
        ^{k\cdot\ell}\left(  Z_{A}^{k}X_{A}^{\ell}\otimes Z_{B}^{k}X_{B}^{\ell
        }\right)  |\psi\rangle_{ABC}\\
        -\frac{1}{2}\sqrt{\frac{\Delta}{2}}\left(1+\textcolor{green}{\Delta}\right)
        |1\rangle_{5}|0\rangle_{6}\sum_{k,\ell}\left(Z_{A}^{k}X_{A}^{\ell}\otimes
        Z_{C}^{k}X_{C}^{\ell}\right)|\psi\rangle_{ABC}\\ =\left(1+\textcolor{green}{\Delta}\right) |0\rangle_{5}|0\rangle_{6}|\psi\rangle_{ABC}+\sqrt{\Delta}|\Delta_{-}\rangle_{5}|1\rangle_{6}M|\psi\rangle_{ABC}-\frac{1}{2}\sqrt{\frac{\Delta}{2}}\left(1+\textcolor{green}{\Delta}\right)  |1\rangle_{5}|0\rangle_{6}M^{\dag}M|\psi\rangle_{ABC}.
    \end{multline}
    Now, apply $C_6 Z_5$. The state then becomes%
    \begin{equation}
        \left(  1+\textcolor{green}{\Delta}\right)  |0\rangle_{5}|0\rangle_{6}|\psi\rangle_{ABC}+\sqrt{\Delta}|\Delta_{+}\rangle_{5}|1\rangle_{6}M|\psi\rangle_{ABC}-\frac{1}{2}\sqrt{\frac{\Delta}{2}}\left(  1+\textcolor{green}{\Delta}\right)  |1\rangle_{5}|0\rangle_{6}M^{\dag}M|\psi\rangle_{ABC},
    \end{equation}
    where%
    \begin{equation}
    |\Delta_{+}\rangle\equiv|0\rangle_{5}+\textcolor{green}{\sqrt{\Delta}}|1\rangle_{5}%
    \end{equation}
    Now apply the projection onto $\langle0|_{5}+\langle1|_{5}\textcolor{green}{\sqrt{\Delta}}$,
    which gives%
    \begin{align}
    & \left(  1+\textcolor{green}{\Delta}\right)  |0\rangle_{6}|\psi\rangle_{ABC}+\left(
    1+\textcolor{green}{\Delta}\right)  \sqrt{\Delta}|1\rangle_{6}M|\psi\rangle_{ABC}-\frac{\Delta
    }{2}\left(  1+\Delta\right)  |0\rangle_{6}M^{\dag}M|\psi\rangle_{ABC}%
    \nonumber\\
    & =\left(  1+\textcolor{green}{\Delta}\right)  |0\rangle_{6}\left(  I-\frac{\Delta}{2}M^{\dag
    }M\right)  |\psi\rangle_{ABC}+\left(  1+\textcolor{green}{\Delta}\right)  |1\rangle_{6}%
    \sqrt{\Delta}M|\psi\rangle_{ABC}\\
    & \propto|0\rangle_{6}\left(  I-\frac{\Delta}{2}M^{\dag}M\right)  |\psi
    \rangle_{ABC}+|1\rangle_{6}\sqrt{\Delta}M|\psi\rangle_{ABC}.
    \end{align}
    This is the final correct state, so that we realize the quantum map in \eqref{eqn:lcu_approx_wml_appendix} after tracing over register 6.

\section{Wave Matrix Lindbladization \texorpdfstring{$e^{\mathcal{M}\Delta}$}{exp(M Delta)} Channel Protocol 2}
\label{app:reduce-aux-overhead}

    In this appendix, we demonstrate how to reduce the auxillary-qubit overhead by reducing the number of unitaries required to express each $M_i$ as defined in~\eqref{eq:M_i-exp-full}. The key idea for reducing the number of terms in the linear-combination expression given in~\eqref{eq:M_i-exp-full} is to map this linear combination of unitaries to a linear combination of different unitaries. We achieve this by combining certain unitaries in a manner that ensures the resulting combination remains a unitary operator. To this end, we define the following unitaries:
    \begin{equation}
        M = \textcolor{green}{\frac{1}{2\sqrt{2}}}\left(U_{1,0} + U_{1,1} + U_{1,2} + U_{1,3}\right),
    \end{equation}
    where
    \begin{align}
        2\,U_{1,0}&\coloneqq X\otimes X \otimes I - I\otimes Y \otimes Y + Z\otimes Z \otimes I - Y\otimes I \otimes Y ,\\
        2\,U_{1,1} &\coloneqq I\otimes X \otimes X + Y \otimes XY \otimes X + X\otimes I \otimes X + Z \otimes XZ \otimes X ,\\
        2\,U_{1,2} &\coloneqq Y\otimes ZY \otimes Z + Z\otimes I\otimes Z + I \otimes Z\otimes Z + X \otimes ZX \otimes Z,\\
        2\,U_{1,3} &\coloneqq I\otimes I\otimes I + Y \otimes Y \otimes I -X\otimes YX \otimes Y - Z\otimes YZ\otimes Y.
    \end{align}
    Now, we evaluate $M^\dag M$ as follows: 
    \begin{align}
        M^\dag& M =\frac{1}{2}\left(\left(\operatorname{SWAP}_{12}\otimes I\right )\left(I_{1}\otimes |\Gamma\rangle\!\langle\Gamma|_{23}\right)
         \left(I_{1}\otimes |\Gamma\rangle\!\langle\Gamma|_{23}\right)\left(\operatorname{SWAP}_{12}\otimes I\right)\right)\\
         &=\left(\operatorname{SWAP}_{12}\otimes I\right )\left(I_{1}\otimes |\Gamma\rangle\!\langle\Gamma|_{23}\right)\left(\operatorname{SWAP}_{12}\otimes I\right)\\
         &=\sqrt{2}(\operatorname{SWAP}_{12} \otimes I) M= \textcolor{green}{\frac{1}{2}}(\operatorname{SWAP}_{12} \otimes I)\left(U_{1,0} + U_{1,1} + U_{1,2} + U_{1,3}\right)\\
         &= \textcolor{green}{\frac{1}{2}}\left(U_{0,0} + U_{0,1} + U_{0,2} + U_{0,3}\right),
    \end{align} 
    where $U_{0,i}=(\operatorname{SWAP}_{12} \otimes I)U_{1,i}$. This represents the case in which $M$ is applied on a one-qubit system.

    Observe that $M$ is now a linear combination of only four unitaries. From~\eqref{eq:decomp-M-M_i}, and the above equality, we have that there are now $4^q$ (or $2^{2q}$) terms in the linear-combination expression of $M$. This is a quadratic improvement over the previous expression of $M$, which contained $16^q$ (or $2^{4q}$) terms. This quadratic improvement halves the number of required auxiliary qubits.  Although this is a constant improvement, it is important in the actual implementation of the algorithm.
    
    Using this new linear combination of unitaries, we describe a new protocol to implement a quantum channel that approximates $e^{\mathcal{M}\Delta}$
    \begin{align} 
        A_{0}  & =I-\frac{\Delta}{2}M^{\dag}M =I+ \textcolor{green}{\frac{\Delta}{4}}\left(U_{0,0} + U_{0,1} + U_{0,2} + U_{0,3}\right),\\
        A_{1}  & =\sqrt{\Delta}M = \textcolor{green}{\frac{\sqrt{\Delta}}{2\sqrt{2}}}\left(U_{1,0} + U_{1,1} + U_{1,2} + U_{1,3}\right)
    \end{align}
    To do so, we can use linear combination of unitaries (LCU) methods~\cite{cleve2019efficient}. The unitaries required to implement the Kraus operators $A_{0}$ and $A_{1}$ are defined are follows
    \begin{align}
        C_{1234}U(0,0)_{ABC} &  \coloneqq |00\rangle\!\langle00|_{12}\otimes|1\rangle\!\langle1|_{3}\otimes|0\rangle\!\langle 0|_{4}\otimes U_{0,0} +\left(I_{1234}-|00\rangle\!\langle00|_{12}\otimes|0\rangle\!\langle0|_{3}\otimes|0\rangle\!\langle 0|_{4}\right)  \otimes I_{ABC},\\
        C_{1234}U(0,1)_{ABC} &  \coloneqq |01\rangle\!\langle01|_{12}\otimes|1\rangle\!\langle 1|_{3}\otimes|0\rangle\!\langle 0|_{4}\otimes U_{0,1} +\left(I_{1234}-|01\rangle\!\langle01|_{12}\otimes|0\rangle\!\langle0|_{3}\otimes|0\rangle\!\langle 0|_{4}\right)  \otimes I_{ABC},\\
        C_{1234}U(0,2)_{ABC} &  \coloneqq |10\rangle\!\langle10|_{12}\otimes|1\rangle\!\langle1|_{3}\otimes|0\rangle\!\langle 0|_{4}\otimes U_{0,2} +\left(I_{1234}-|10\rangle\!\langle10|_{12}\otimes|0\rangle\!\langle0|_{3}\otimes|0\rangle\!\langle 0|_{4}\right)  \otimes I_{ABC},\\
        C_{1234}U(0,3)_{ABC} &  \coloneqq |11\rangle\!\langle11|_{12}\otimes|1\rangle\!\langle1|_{3}\otimes|0\rangle\!\langle 0|_{4}\otimes U_{0,3} +\left(I_{1234}-|11\rangle\!\langle11|_{12}\otimes|0\rangle\!\langle0|_{3}\otimes|0\rangle\!\langle 0|_{4}\right)  \otimes I_{ABC},\\
        C_{124}U(1,0)_{ABC} &  \coloneqq |00\rangle\!\langle00|_{12}\otimes|1\rangle\!\langle 1|_{4}\otimes U_{1,0} +\left(I_{124}-|00\rangle\!\langle00|_{12}\otimes|1\rangle\!\langle 1|_{4}\right)  \otimes I_{ABC},\\
        C_{124}U(1,1)_{ABC} &  \coloneqq |01\rangle\!\langle01|_{12}\otimes|1\rangle\!\langle 1|_{4}\otimes U_{1,1} +\left(I_{124}-|01\rangle\!\langle01|_{12}\otimes|1\rangle\!\langle 1|_{4}\right)  \otimes I_{ABC},\\
        C_{124}U(1,2)_{ABC} &  \coloneqq |10\rangle\!\langle10|_{12}\otimes|1\rangle\!\langle 1|_{4}\otimes U_{1,2} +\left(I_{124}-|10\rangle\!\langle10|_{12}\otimes|1\rangle\!\langle 1|_{4}\right)  \otimes I_{ABC},\\
        C_{124}U(1,3)_{ABC} &  \coloneqq |11\rangle\!\langle11|_{12}\otimes|1\rangle\!\langle 1|_{4}\otimes U_{1,3} +\left(I_{124}-|11\rangle\!\langle11|_{12}\otimes|1\rangle\!\langle 1|_{4}\right)  \otimes I_{ABC}.
    \end{align}
\newline

\textbf{Channel Protocol 2 ($e^{\mathcal{M}\Delta}$) ---} Collect the four unitary gates for the $A_0$ Kraus operator and the four unitary gates corresponding to the non-identity part of the $A_1$ Kraus operator and  calculate $\Delta$. In Figure~\ref{fig:sample channel}, we show a sample $e^{\mathcal{M}\Delta}$ channel quantum circuit being applied when an emitter term is sampled.

\begin{figure}
    \centering
    \includegraphics[width=\columnwidth]{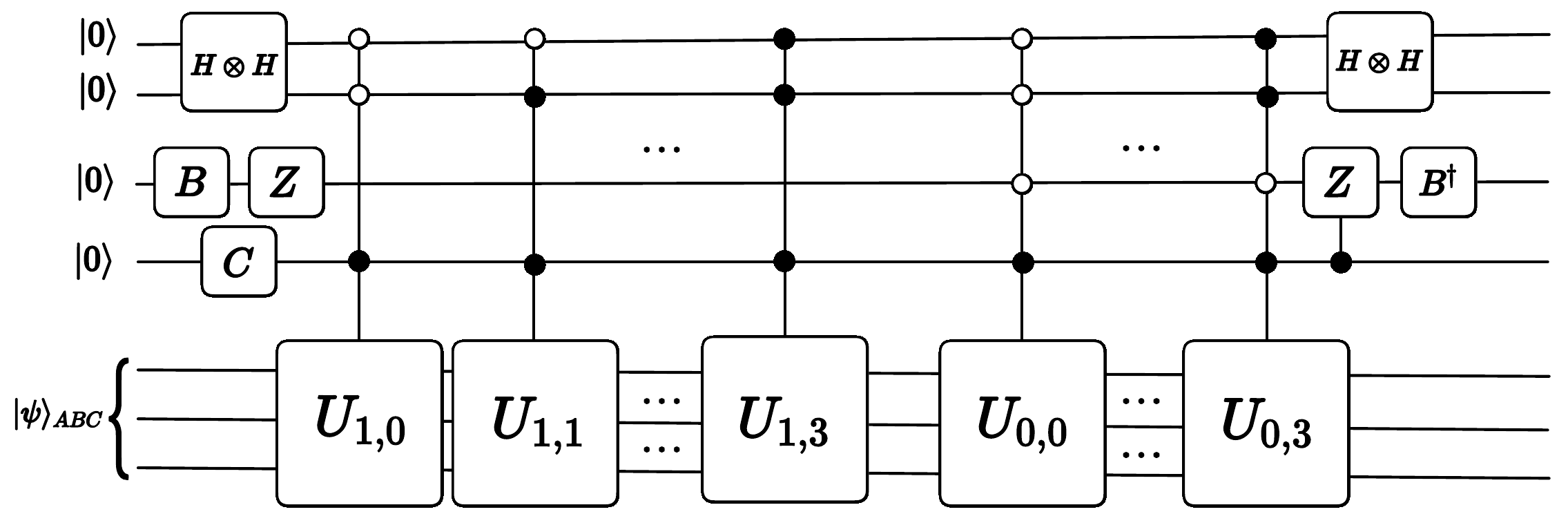}
    \caption{Circuit diagram for Protocol~2 for approximately implementing the channel $e^{\mathcal{M}\Delta}$.}
    \label{fig:sample channel}
\end{figure} 

    Apply the following gates
    \begin{align}
        \textcolor{green}{B_3|0\rangle} & = \textcolor{green}{\frac{1}{\sqrt{1+\Delta}}\left(|0\rangle+ \sqrt{\Delta}|1\rangle\right),}\\
        \textcolor{green}{C_4|0\rangle} & = \textcolor{green}{\frac{1}{\sqrt{(1+\Delta)^2+(\sqrt{2\Delta})^2}}\left(\left(1+\Delta\right)|0\rangle+\sqrt{2\Delta}|1\rangle\right).}
    \end{align}
    
    The state after the application of the gates $H_1 H_2 Z_3 B_3 C_4$ is%
    \begin{align}
        & |+\rangle_{1}|+\rangle_{2}\left(|0\rangle_3 -  \textcolor{green}{\sqrt{\Delta}}|1\rangle_3\right) \left(\textcolor{green}{\left(1+\Delta\right)}|0\rangle_4+\textcolor{green}{\sqrt{2\Delta}}|1\rangle_4\right)  |\psi\rangle_{ABC}%
        \nonumber\\
        & =\left(1+\textcolor{green}{\Delta}\right)  |+\rangle_{1}|+\rangle_{2}|0\rangle_{3}|0\rangle_{4}|\psi\rangle_{ABC} + \textcolor{green}{\sqrt{2\Delta}}|+\rangle_{1}|+\rangle_{2}\left(|0\rangle_3- \textcolor{green}{\sqrt{\Delta}}|1\rangle_3\right)|1\rangle_{4}|\psi\rangle_{ABC}-\textcolor{green}{\sqrt{\Delta}\left(1+\Delta\right)}|+\rangle_{1}|+\rangle
        _{2}|1\rangle_{3}|0\rangle_{4}|\psi\rangle_{ABC}.
    \end{align}
    
    Let us evaluate the second term, 
    \begin{equation}
        |+\rangle_{1}|+\rangle_{2}\left(|0\rangle_3- \textcolor{green}{\sqrt{\Delta}}|1\rangle_3\right)|1\rangle_{4}|\psi\rangle_{ABC}.
    \end{equation}
    
    Upon applying unitaries $C_{124}U(1,0)_{ABC}$, $C_{124}U(1,1)_{ABC}$, $C_{124}U(1,2)_{ABC}$, and $C_{124}U(1,3)_{ABC}$, the state is: 
    \begin{equation}
        \left(|0\rangle_3- \textcolor{green}{\sqrt{\Delta}}|1\rangle_3\right)\textcolor{green}{\otimes\frac{1}{2}}\bigg(\ket{001}_{124}\otimes U_{1,0}|\psi\rangle_{ABC}+\ket{011}_{124}\otimes U_{1,1}|\psi\rangle_{ABC}
        +\ket{101}_{124}\otimes U_{1,2}|\psi\rangle_{ABC}+\ket{111}_{124}\otimes U_{1,3}|\psi\rangle_{ABC}\bigg).
    \end{equation}
    
    Let us evaluate the third term. 
    \begin{equation}
        |+\rangle_{1}|+\rangle_{2}|1\rangle_{3}|0\rangle_{4}|\psi\rangle_{ABC}
    \end{equation}
    
    Apply $C_{1234}U(0,0)_{ABC}$, $C_{1234}U(0,1)_{ABC}$, $C_{1234}U(0,2)_{ABC}$, and $C_{1234}U(0,3)_{ABC}$. The state is 
    \begin{equation}
        \ket{10}_{34}\otimes\textcolor{green}{\frac{1}{2}}\bigg(\ket{00}_{12}\otimes U_{0,0}|\psi\rangle_{ABC}+\ket{01}_{12}\otimes U_{0,1}|\psi\rangle_{ABC} +\ket{10}_{12}\otimes U_{0,2}|\psi\rangle_{ABC}+\ket{11}_{12}\otimes U_{0,3}|\psi\rangle_{ABC}\bigg)
    \end{equation}
    
    Applying the projection onto $\langle+|_{1}\langle+|_{2}$, the resulting state is 
    \begin{multline}
        \left(1+\textcolor{green}{\Delta}\right)|0\rangle_{3}|0\rangle_{4}|\psi\rangle_{ABC}
         -\frac{\sqrt{\Delta}}{\textcolor{green}{4}}\left(1+\textcolor{green}{\Delta}\right)\ket{10}_{34}\otimes \bigg(U_{0,0}+ U_{0,1}+ U_{0,2}+ U_{0,3}\bigg)|\psi\rangle_{ABC}\notag\\
        +\textcolor{green}{\frac{\sqrt{\Delta}}{2\sqrt{2}}}\left(|0\rangle_3- \textcolor{green}{\sqrt{\Delta}}|1\rangle_3\right)\ket{1}_{4}\otimes\bigg(U_{1,0}+ U_{1,1}+U_{1,2}+ U_{1,3}\bigg)|\psi\rangle_{ABC}.
    \end{multline}
    
    Then upon applying $C_4 Z_3\coloneqq |1\rangle\!\langle 1|_{4}\otimes Z_3 +\left(I_{4}-|1\rangle\!\langle1|_{4}\right)  \otimes I_{3}$ and the projection onto $\langle0|_{3}+\langle1|_{3}\textcolor{green}{\sqrt{\Delta}}$, the final state is
    \begin{equation}
        \propto\left(I
         -\frac{\Delta}{\textcolor{green}{2}}M^\dagger M\right)|\psi\rangle_{ABC}\ket{0}_{4}
        +\textcolor{green}{\sqrt{\Delta}}M|\psi\rangle_{ABC}\ket{1}_{4}.
    \end{equation}

\section{Proof of Theorem~\ref{thm:gate-comp-WML}}\label{app:wml_complexity}
In what follows, we prove Theorem~\ref{thm:gate-comp-WML} by breaking the analysis into two parts. To make sense of these two parts, consider the following:
\begin{align}
    \frac{1}{2}\left \Vert e^{\mathcal{L}t} - \left(\mathcal{A}_{\operatorname{WML}}^{\operatorname{(LCU)}}\right)^{\!\circ n}\right\Vert_{\diamond} & = \frac{1}{2}\left \Vert e^{\mathcal{L}t} - \left(\mathcal{A}_{\operatorname{WML}}^{\operatorname{(ideal)}}\right)^{\!\circ n} + \left(\mathcal{A}_{\operatorname{WML}}^{\operatorname{(ideal)}}\right)^{\!\circ n} - \left(\mathcal{A}_{\operatorname{WML}}^{\operatorname{(LCU)}}\right)^{\!\circ n}\right\Vert_{\diamond}\\
    & \leq \frac{1}{2}\left \Vert e^{\mathcal{L}t} - \left(\mathcal{A}_{\operatorname{WML}}^{\operatorname{(ideal)}}\right)^{\!\circ n} \right\Vert_{\diamond} + \frac{1}{2}\left \Vert \left(\mathcal{A}_{\operatorname{WML}}^{\operatorname{(ideal)}}\right)^{\!\circ n}- \left(\mathcal{A}_{\operatorname{WML}}^{\operatorname{(LCU)}}\right)^{\!\circ n}\right\Vert_{\diamond}\\
    & = \frac{1}{2}\left \Vert \left(e^{\mathcal{L}\tau}\right)^{\!\circ n} - \left(\mathcal{A}_{\operatorname{WML}}^{\operatorname{(ideal)}}\right)^{\!\circ n} \right\Vert_{\diamond} + \frac{1}{2}\left \Vert \left(\mathcal{A}_{\operatorname{WML}}^{\operatorname{(ideal)}}\right)^{\!\circ n}- \left(\mathcal{A}_{\operatorname{WML}}^{\operatorname{(LCU)}}\right)^{\!\circ n}\right\Vert_{\diamond}\\
    & \leq \frac{n}{2}\left \Vert e^{\mathcal{L}\tau} - \mathcal{A}_{\operatorname{WML}}^{\operatorname{(ideal)}}\right\Vert_{\diamond} + \frac{n}{2}\left \Vert \mathcal{A}_{\operatorname{WML}}^{\operatorname{(ideal)}}- \mathcal{A}_{\operatorname{WML}}^{\operatorname{(LCU)}}\right\Vert_{\diamond}\label{eq:dia-dist-two-parts}.
\end{align}
To achieve a final error of at most $\varepsilon$, we can ensure that each of the two terms on the right-hand side of the inequality is bounded from above by $\frac{\varepsilon}{2}$: 
\begin{gather}
    \frac{n}{2}\left \Vert e^{\mathcal{L}\tau} - \mathcal{A}_{\operatorname{WML}}^{\operatorname{(ideal)}}\right\Vert_{\diamond}  \leq \frac{\varepsilon}{2},\\
     \frac{n}{2}\left \Vert \mathcal{A}_{\operatorname{WML}}^{\operatorname{(ideal)}}- \mathcal{A}_{\operatorname{WML}}^{\operatorname{(LCU)}}\right\Vert_{\diamond} \leq \frac{\varepsilon}{2}.
\end{gather}
To simplify the subsequent analysis, we divide it into two parts. In the first part, we analyze the initial inequality, which resolves the sample complexity of the algorithm, i.e., $n$. In the second part, we analyze the second inequality, which resolves the gate complexity of the algorithm.

\subsection{Sample Complexity}
Consider the following:
\begin{equation}
     \left \Vert e^{\mathcal{L}\tau} - \mathcal{A}_{\operatorname{WML}}^{\operatorname{(ideal)}}\right\Vert_{\diamond} 
     = \left \Vert e^{\mathcal{L}\tau} -\left(
     \begin{array}[c]{c}\sum_{j: c_j>0} \frac{c_j}{c} \operatorname{Tr}_{2}\circ~ e^{\mathcal{N}_1 c\tau} \circ \mathcal{P}_{1, j} +  \sum_{j: c_j<0} \frac{(-c_j)}{c} \operatorname{Tr}_{2}\circ~ e^{\mathcal{N}_2 c\tau} \circ \mathcal{P}_{1, j} \\ + \sum_{k} \frac{\left \Vert L_k\right\Vert_2^2}{c} \operatorname{Tr}_{23}\circ~ e^{\mathcal{M} c\tau} \circ \mathcal{P}_{2, k}
     \end{array} \right) \right\Vert_{\diamond}\label{eq:dim_exp}
\end{equation}
Expanding the second term on the right-hand side of the above equation, we obtain:
\begin{align}
    & \sum_{j: c_j>0} \frac{c_j}{c} \operatorname{Tr}_{23}\circ~ e^{\mathcal{N}_1 c\tau} \circ \mathcal{P}_{1, j} \notag\\
    & = \sum_{j: c_j>0} \frac{c_j}{c} \operatorname{Tr}_{2}\circ\left( \mathcal{I} + c \tau \mathcal{N}_1 + \sum_{r=2}^{\infty}\frac{ c^r \tau^r}{r!} \mathcal{N}^{r}_1  \right) \circ \mathcal{P}_{1, j}\\
    & = \sum_{j: c_j>0} \frac{c_j}{c} \mathcal{I} + \sum_{j: c_j>0} c_j \tau\operatorname{Tr}_{2}\circ~ \mathcal{N}_1 \circ \mathcal{P}_{1, j}+ \sum_{j: c_j>0}\sum_{r=2}^{\infty} \frac{c_jc^{r-1} \tau^r}{r!}   \operatorname{Tr}_{2}\circ~\mathcal{N}^{r}_1  \circ \mathcal{P}_{1, j}\\
    & = \sum_{j: c_j>0} \frac{c_j}{c} \mathcal{I} + \sum_{j: c_j>0} c_j \tau  \mathcal{H}_j + \sum_{j: c_j>0}\sum_{r=2}^{\infty} \frac{c_jc^{r-1} \tau^r}{r!}   \operatorname{Tr}_{2}\circ~\mathcal{N}^{r}_1  \circ \mathcal{P}_{1, j}\label{eq:H_first_term}.
\end{align}
Similarly, we obtain the following expression for the third term:
\begin{equation}
    \sum_{j: c_j<0} \frac{(-c_j)}{c} \operatorname{Tr}_{2}\circ~ e^{\mathcal{N}_2 c\tau} \circ \mathcal{P}_{1, j} 
    =  \sum_{j: c_j<0} \frac{(-c_j)}{c} \mathcal{I} + \sum_{j: c_j<0} (-c_j) \tau  \mathcal{H}_j + \sum_{j: c_j<0}\sum_{r=2}^{\infty} \frac{(-c_j)c^{r-1} \tau^r}{r!}   \operatorname{Tr}_{2}\circ~\mathcal{N}^{r}_2  \circ \mathcal{P}_{1, j}\label{eq:H_sec_term},
\end{equation}
and the following expression for the fourth term:
\begin{equation}
    \sum_{k} \frac{\left \Vert L_k\right\Vert_2^2}{c} \operatorname{Tr}_{23}\circ~ e^{\mathcal{M} c\tau} \circ \mathcal{P}_{2, k}
   = \sum_{k} \frac{\left \Vert L_k\right\Vert_2^2}{c} \mathcal{I} + \sum_{k} \left \Vert L_k\right\Vert_2^2 \tau  \mathcal{L}_k + \sum_{k}\sum_{r=2}^{\infty} \frac{\left \Vert L_k\right\Vert_2^2c^{r-1} \tau^r}{r!}   \operatorname{Tr}_{23}\circ~\mathcal{M}^{r}  \circ \mathcal{P}_{2, k}\label{eq:lind_term}.
\end{equation}
Combining~\eqref{eq:H_first_term},~\eqref{eq:H_sec_term}, and~\eqref{eq:lind_term} and rearranging, we get
\begin{align}
    & \left(\underbrace{\sum_{j: c_j>0} \frac{c_j}{c} + \sum_{j: c_j<0} \frac{(-c_j)}{c} + \sum_{k} \frac{\left \Vert L_k\right\Vert_2^2}{c}}_{=1}\right) \mathcal{I} + \tau \left(\underbrace{\sum_{j: c_j>0} c_j   \mathcal{H}_j + \sum_{j: c_j<0} (-c_j)  \mathcal{H}_j + \sum_{k} \left \Vert L_k\right\Vert_2^2  \mathcal{L}_k}_{=\mathcal{L}}\right)  \notag\\
    & \qquad + \sum_{j: c_j>0}\sum_{r=2}^{\infty} \frac{c_jc^{r-1} \tau^r}{r!}   \operatorname{Tr}_{2}\circ~\mathcal{N}^{r}_1  \circ \mathcal{P}_{1, j} + \sum_{j: c_j<0}\sum_{r=2}^{\infty} \frac{(-c_j)c^{r-1} \tau^r}{r!}   \operatorname{Tr}_{2}\circ~\mathcal{N}^{r}_2  \circ \mathcal{P}_{1, j}\notag\\
    & \qquad + \sum_{k}\sum_{r=2}^{\infty} \frac{\left \Vert L_k\right\Vert_2^2c^{r-1} \tau^r}{r!}   \operatorname{Tr}_{23}\circ~\mathcal{M}^{r}  \circ \mathcal{P}_{2, k}\\
    & = \mathcal{I} + \tau \mathcal{L} + \sum_{j: c_j>0}\sum_{r=2}^{\infty} \frac{c_jc^{r-1} \tau^r}{r!}   \operatorname{Tr}_{2}\circ~\mathcal{N}^{r}_1  \circ \mathcal{P}_{1, j} + \sum_{j: c_j<0}\sum_{r=2}^{\infty} \frac{(-c_j)c^{r-1} \tau^r}{r!}   \operatorname{Tr}_{2}\circ~\mathcal{N}^{r}_2  \circ \mathcal{P}_{1, j}\notag\\
    & \qquad + \sum_{k}\sum_{r=2}^{\infty} \frac{\left \Vert L_k\right\Vert_2^2c^{r-1} \tau^r}{r!}   \operatorname{Tr}_{23}\circ~\mathcal{M}^{r}  \circ \mathcal{P}_{2, k}.
\end{align}
By substituting the right-hand side of the above equation into~\eqref{eq:dim_exp} and expanding the term $e^{\mathcal{L}\tau}$ using its Taylor series, the first two terms of the Taylor series get canceled. As a result, we get the following:
\begin{align}
    & \left \Vert e^{\mathcal{L}\tau} - \mathcal{A}_{\operatorname{WML}}^{\operatorname{(ideal)}}\right\Vert_{\diamond} \notag\\
     & = \left \Vert \sum_{r=2}^{\infty} \frac{\tau^r}{r!}\mathcal{L}^r -\left(
     \begin{array}[c]{c}\sum_{j: c_j>0}\sum_{r=2}^{\infty} \frac{c_jc^{r-1} \tau^r}{r!}   \operatorname{Tr}_{2}\circ~\mathcal{N}^{r}_1  \circ \mathcal{P}_{1, j} + \sum_{j: c_j<0}\sum_{r=2}^{\infty} \frac{(-c_j)c^{r-1} \tau^r}{r!}   \operatorname{Tr}_{2}\circ~\mathcal{N}^{r}_2  \circ \mathcal{P}_{1, j}\\
    + \sum_{k}\sum_{r=2}^{\infty} \frac{\left \Vert L_k\right\Vert_2^2c^{r-1} \tau^r}{r!}   \operatorname{Tr}_{23}\circ~\mathcal{M}^{r}  \circ \mathcal{P}_{2, k}
     \end{array} \right) \right\Vert_{\diamond}\\
     & \leq \sum_{r=2}^{\infty} \frac{\tau^r}{r!}\left \Vert \mathcal{L}^r \right\Vert_{\diamond} + 
    \sum_{j: c_j>0}\sum_{r=2}^{\infty} \frac{c_jc^{r-1} \tau^r}{r!}\left \Vert   \operatorname{Tr}_{2}\circ~\mathcal{N}^{r}_1  \circ \mathcal{P}_{1, j}\right\Vert_{\diamond} + \sum_{j: c_j<0}\sum_{r=2}^{\infty} \frac{(-c_j)c^{r-1} \tau^r}{r!}   \left \Vert \operatorname{Tr}_{2}\circ~\mathcal{N}^{r}_2  \circ \mathcal{P}_{1, j}\right\Vert_{\diamond}\notag\\
    & \qquad + \sum_{k}\sum_{r=2}^{\infty} \frac{\left \Vert L_k\right\Vert_2^2c^{r-1} \tau^r}{r!}   \left \Vert \operatorname{Tr}_{23}\circ~\mathcal{M}^{r}  \circ \mathcal{P}_{2, k}\right\Vert_{\diamond}\\
    & \leq \sum_{r=2}^{\infty} \frac{\tau^r}{r!}\left \Vert \mathcal{L}^r \right\Vert_{\diamond} + 
    \sum_{j: c_j>0}\sum_{r=2}^{\infty} \frac{c_jc^{r-1} \tau^r}{r!}\left \Vert   \mathcal{N}^{r}_1\right\Vert_{\diamond} + \sum_{j: c_j<0}\sum_{r=2}^{\infty} \frac{(-c_j)c^{r-1} \tau^r}{r!}   \left \Vert \mathcal{N}^{r}_2\right\Vert_{\diamond}\notag\\
    & \qquad + \sum_{k}\sum_{r=2}^{\infty} \frac{\left \Vert L_k\right\Vert_2^2c^{r-1} \tau^r}{r!}   \left \Vert \mathcal{M}^{r} \right\Vert_{\diamond}\\
    & \leq \sum_{r=2}^{\infty} \frac{\tau^r}{r!}\left \Vert \mathcal{L} \right\Vert_{\diamond}^r + 
    \sum_{j: c_j>0}\sum_{r=2}^{\infty} \frac{c_jc^{r-1} \tau^r}{r!}\left \Vert   \mathcal{N}_1\right\Vert_{\diamond}^r + \sum_{j: c_j<0}\sum_{r=2}^{\infty} \frac{(-c_j)c^{r-1} \tau^r}{r!}   \left \Vert \mathcal{N}_2\right\Vert_{\diamond}^r\notag\\
    & \qquad + \sum_{k}\sum_{r=2}^{\infty} \frac{\left \Vert L_k\right\Vert_2^2c^{r-1} \tau^r}{r!}   \left \Vert \mathcal{M} \right\Vert_{\diamond}^r.\label{eq:first_sim_exp}
\end{align}
The first inequality follows from the triangle inequality. The second inequality follows from the following two facts: The diamond norm is submultiplicative under composition of maps, i.e., for maps $\mathcal{Q}$ and $\mathcal{R}$, it holds that $\left \Vert\mathcal{Q} \circ \mathcal{R} \right \Vert_{\diamond} \leq \left \Vert\mathcal{Q} \right \Vert_{\diamond}\left \Vert\mathcal{R} \right \Vert_{\diamond}$, and 2) the diamond norm for a quantum channel is equal to one, i.e., for a quantum channel $\mathcal{Q}$, it holds that $\left \Vert\mathcal{Q} \right \Vert_{\diamond} = 1$. Finally, the third inequality also follows from the submultiplicativity of the diamond norm under composition of maps.

Now, consider the following:
\begin{align}
    \left \Vert \mathcal{L} \right\Vert_{\diamond} & = \left \Vert \sum_{j: c_j>0} c_j   \mathcal{H}_j + \sum_{j: c_j<0} (-c_j)  \mathcal{H}_j + \sum_{k} \left \Vert L_k\right\Vert_2^2  \widehat{\mathcal{L}}_k \right\Vert_{\diamond}\\
    & \leq \sum_{j: c_j>0} c_j    \left \Vert\mathcal{H}_j\right\Vert_{\diamond} + \sum_{j: c_j<0} (-c_j)   \left \Vert\mathcal{H}_j\right\Vert_{\diamond} + \sum_{k} \left \Vert L_k\right\Vert_2^2   \left \Vert \widehat{\mathcal{L}}_k\right\Vert_{\diamond} \\
    & \leq \sum_{j: c_j>0} c_j    (2) + \sum_{j: c_j<0} (-c_j)   (2) + \sum_{k} \left \Vert L_k\right\Vert_2^2 (2)\\
    & = 2 \left(\sum_{j: c_j>0} c_j + \sum_{j: c_j<0} (-c_j)+ \sum_{k} \left \Vert L_k\right\Vert_2^2\right)\\
    & = 2c.\label{eq:bound_on_L}
\end{align}
The first inequality follows from the triangle inequality. The second inequality holds due to the following:
\begin{align}
    \left \Vert\mathcal{H}_j\right\Vert_{\diamond} & = \sup_{\omega}\left \Vert\mathcal{H}_j(\omega)\right\Vert_{1} = \sup_{\omega} \left \Vert (-i) [\sigma_j, \omega]\right\Vert_{1} \leq 2,\\
    \left \Vert\widehat{\mathcal{L}}_k(\cdot)\right\Vert_{\diamond} & = \sup_{\omega}\left \Vert\widehat{\mathcal{L}}_k(\omega)\right\Vert_{1} = \sup_{\omega}\left \Vert \widehat{L}_k \omega \widehat{L}_k^{\dagger} - \frac{1}{2} \left\{\widehat{L}_k^{\dagger} \widehat{L}_k, \omega\right\}\right\Vert_{1} \leq 2,
\end{align}
where $\widehat{\mathcal{L}}_k\coloneqq   \mathcal{L}_k/\left \Vert L_k\right\Vert_2^2$.

Now, similar to bounding $\left \Vert \mathcal{L}' \right\Vert_{\diamond}$, we bound $\left \Vert   \mathcal{N}_1\right\Vert_{\diamond}$, $\left \Vert   \mathcal{N}_2\right\Vert_{\diamond}$, and $\left \Vert   \mathcal{M}\right\Vert_{\diamond}$ from above:
\begin{align}
    \left \Vert   \mathcal{N}_1\right\Vert_{\diamond} & = \sup_{\omega}\left \Vert   \mathcal{N}_1 (\omega)\right\Vert_{1} = \sup_{\omega} \left \Vert   (-i)[\operatorname{SWAP}, \omega]\right\Vert_{1}\\
    & = \sup_{\omega} \left \Vert\left (\operatorname{SWAP} \omega - \omega\operatorname{SWAP}\right)\right\Vert_{1}\leq 2\left \Vert\operatorname{SWAP}\right\Vert \leq 2\label{eq:bound_on_N_1},\\
    \left \Vert   \mathcal{N}_2\right\Vert_{\diamond} & = \sup_{\omega}\left \Vert   \mathcal{N}_1 (\omega)\right\Vert_{1} = \sup_{\omega} \left \Vert   (-i)[\operatorname{-SWAP}, \omega]\right\Vert_{1}\\
    & = \sup_{\omega} \left \Vert\left (\operatorname{SWAP} \omega - \omega\operatorname{SWAP}\right)\right\Vert_{1}\leq 2\left \Vert\operatorname{SWAP}\right\Vert \leq 2 \label{eq:bound_on_N_2},
    \\
    \left \Vert   \mathcal{M}\right\Vert_{\diamond} & = \sup_{\omega}\left \Vert   \mathcal{M} (\omega)\right\Vert_{1} = \sup_{\omega}\left \Vert   M\omega M^{\dagger} - \frac{1}{2} \{M^{\dagger}M, \omega\}\right\Vert_{1}\\
    & \leq 2 \left \Vert M \right \Vert^2 \leq 2 Q.\label{eq:bound_on_M}
\end{align}
Here, the last inequality follows due to the following:
\begin{align}
    \left \Vert M \right \Vert & = \left \Vert \frac{1}{\sqrt{Q}}\left(I_{1}\otimes |\Gamma\rangle\!\langle\Gamma|_{23}\right)\left(\operatorname{SWAP}_{12}\otimes I_3\right) \right \Vert\\
    & = \left \Vert \sqrt{Q}\left(I_{1}\otimes |\Phi\rangle\!\langle\Phi|_{23}\right)\left(\operatorname{SWAP}_{12}\otimes I_3\right) \right \Vert\\
    & \leq \sqrt{Q},
\end{align}
where the last inequality follows due to the submultiplicativity of operator norm under composition and tensor product.

Using the bounds~\eqref{eq:bound_on_L},~\eqref{eq:bound_on_N_1},~\eqref{eq:bound_on_N_2}, and~\eqref{eq:bound_on_M} in~\eqref{eq:first_sim_exp}, we get
\begin{align}
    & \left \Vert e^{\mathcal{L}\tau} - \mathcal{A}_{\operatorname{WML}}^{\operatorname{(ideal)}}\right\Vert_{\diamond}\notag\\
    & \leq \sum_{r=2}^{\infty} \frac{\tau^r}{r!}(2c)^r + 
    \sum_{j: c_j>0}\sum_{r=2}^{\infty} \frac{c_jc^{r-1} \tau^r}{r!}2^r + \sum_{j: c_j<0}\sum_{r=2}^{\infty} \frac{(-c_j)c^{r-1} \tau^r}{r!}   2^r + \sum_{k}\sum_{r=2}^{\infty} \frac{\left \Vert L_k\right\Vert_2^2c^{r-1} \tau^r}{r!} (2Q)^r\\
    & \leq \sum_{r=2}^{\infty} \frac{(2c\tau)^r}{r!} + 
    \sum_{j: c_j>0}\sum_{r=2}^{\infty} \frac{c_jc^{r-1} (2\tau)^r}{r!} + \sum_{j: c_j<0}\sum_{r=2}^{\infty} \frac{(-c_j)c^{r-1} (2\tau)^r}{r!}+ \sum_{k}\sum_{r=2}^{\infty} \frac{\left \Vert L_k\right\Vert_2^2c^{r-1} (2\tau)^r}{r!} (1 + Q^r - 1)\\
    & = \sum_{r=2}^{\infty} \frac{(2c\tau)^r}{r!} + 
    \sum_{r=2}^{\infty} \frac{c^{r-1} (2\tau)^r}{r!} \left (\underbrace{\sum_{j: c_j>0} c_j  + \sum_{j: c_j<0}(-c_j)+  \sum_{k}\left \Vert L_k\right\Vert_2^2}_{=c}\right) + \sum_{k}\sum_{r=2}^{\infty} \frac{\left \Vert L_k\right\Vert_2^2c^{r-1} (2\tau)^r}{r!} (Q^r - 1)\\
   & = \sum_{r=2}^{\infty} \frac{(2c\tau)^r}{r!} + 
    \sum_{r=2}^{\infty} \frac{ (2c\tau)^r}{r!} + \sum_{k}\sum_{r=2}^{\infty} \frac{\left \Vert L_k\right\Vert_2^2c^{r-1} (2\tau)^r}{r!} (Q^r - 1)\\
    & = 2\sum_{r=2}^{\infty} \frac{(2c\tau)^r}{r!} + \sum_{k}\sum_{r=2}^{\infty} \frac{\left \Vert L_k\right\Vert_2^2c^{r-1} (2\tau)^r}{r!} (Q^r - 1)\\
    & \leq 2\sum_{r=2}^{\infty} \frac{(2c\tau)^r}{r!} + \sum_{r=2}^{\infty} \frac{c \cdot c^{r-1} (2\tau)^r}{r!} Q^r\\
    & \leq 2\sum_{r=2}^{\infty} \frac{(2c\tau)^r}{r!} +\sum_{r=2}^{\infty} \frac{(2cQ\tau)^r}{r!}.
\end{align}
The third inequality follows from the fact that $\sum_k\left \Vert L_k\right\Vert_2^2 \leq c$ and $Q^r - 1 \leq Q^r$. Now, substituting $\tau = \frac{t}{n}$ in the above inequality and dividing by two on both sides for normalizing the diamond distance, we get
\begin{align}
    & \frac{1}{2}\left \Vert e^{\mathcal{L}\tau} - \mathcal{A}_{\operatorname{WML}}^{\operatorname{(ideal)}}\right\Vert_{\diamond}\leq \sum_{r=2}^{\infty} \frac{1}{r!} \left (\frac{2ct}{n}\right)^r + \frac{1}{2}\sum_{r=2}^{\infty} \frac{1}{r!} \left (\frac{2cQt}{n}\right)^r.
\end{align}
To bound the right-hand side of the inequality from above for $n \geq 2cQt$, we utilize the fact that for all $0 \leq x \leq 1$, $\sum_{r=2}^{\infty} \frac{x^r}{r!} \leq x^2$:
\begin{align}
    \frac{1}{2}\left \Vert e^{\mathcal{L}\tau} - \mathcal{A}_{\operatorname{WML}}^{\operatorname{(ideal)}}\right\Vert_{\diamond} & \leq \frac{(2ct)^2}{n^2} + \frac{1}{2}\frac{(2cQt)^2}{n^2} \\
    & \leq \frac{4(cQt)^2}{n^2},
\end{align}
where the last inequality follows due to the fact that $Q\geq 2$.

Now, we use the above inequality to further bound the first term of~\eqref{eq:dia-dist-two-parts} from above:
\begin{align}
    n \cdot \frac{1}{2}\left \Vert e^{\mathcal{L}\tau} - \mathcal{A}_{\operatorname{WML}}^{\operatorname{(ideal)}}\right\Vert_{\diamond} \leq n \cdot \frac{4(cQt)^2}{n^2} = \frac{4(cQt)^2}{n}.
\end{align}
If we want the final error to be less than $\frac{\varepsilon}{2}$, then we need 
\begin{equation}
    n \geq \frac{8(cQt)^2}{\varepsilon} = O\!\left(\frac{c^2t^2}{\varepsilon}\right),
\end{equation}
where we use that fact that $Q = 2^q = 2^{O(1)} = O(1)$. This resolves the sample complexity of the WML algorithm.

\subsection{Gate Complexity}
Substituting~\eqref{eq:practical-wml-iterate} and~\eqref{eq:ideal-WML-iterate} into~\eqref{eq:dia-dist-two-parts}, the first two terms of~\eqref{eq:practical-wml-iterate} and~\eqref{eq:ideal-WML-iterate} cancel out, leaving us with the following expression:
\begin{align}
     \frac{n}{2}\left \Vert \mathcal{A}_{\operatorname{WML}}^{\operatorname{(ideal)}} - \mathcal{A}_{\operatorname{WML}}^{\operatorname{(LCU)}}\right\Vert_{\diamond} & = \frac{n}{2}\left \Vert 
      \sum_{k} \frac{\left \Vert L_k\right\Vert_2^2}{c} \operatorname{Tr}_{23}\circ~ e^{\mathcal{M} c\tau} \circ \mathcal{P}_{2, k} - \sum_{k} \frac{\left \Vert L_k\right\Vert_2^2}{c} \operatorname{Tr}_{23}\circ~\mathcal{R}_{c\tau} \circ \mathcal{P}_{2, k} \right\Vert_{\diamond}\\
      & \leq  \frac{n}{2}\sum_{k} \frac{\left \Vert L_k\right\Vert_2^2}{c}\left \Vert 
      \operatorname{Tr}_{23}\circ~ e^{\mathcal{M} c\tau} \circ \mathcal{P}_{2, k} - \operatorname{Tr}_{23}\circ~\mathcal{R}_{c\tau} \circ \mathcal{P}_{2, k} \right\Vert_{\diamond}\\
      & \leq  \frac{n}{2}\left \Vert 
      \operatorname{Tr}_{23}\circ~ e^{\mathcal{M} c\tau} \circ \mathcal{P}_{2, k} - \operatorname{Tr}_{23}\circ~\mathcal{R}_{c\tau} \circ \mathcal{P}_{2, k} \right\Vert_{\diamond}\\
      & \leq  \frac{n}{2}\left \Vert 
      \operatorname{Tr}_{23}\circ\left(e^{\mathcal{M} c\tau} - \mathcal{R}_{c\tau} \right) \circ \mathcal{P}_{2, k}\right\Vert_{\diamond}\\
      & \leq \frac{n}{2}\left \Vert 
      e^{\mathcal{M} c\tau} - \mathcal{R}_{c\tau} \right\Vert_{\diamond},\label{eq:second-part-simplify}
\end{align}
where the first inequality follows from the triangle inequality, the second inequality follows due to the following fact:
\begin{equation}
    \sum_k\frac{\left \Vert L_k\right\Vert_{2}^{2}}{c} \leq 1,
\end{equation}
and the last inequality follows from the following two facts: The diamond norm is submultiplicative under composition of maps, i.e., for all maps $\mathcal{Q}$ and $\mathcal{R}$, it holds that $\left \Vert\mathcal{Q} \circ \mathcal{R} \right \Vert_{\diamond} \leq \left \Vert\mathcal{Q} \right \Vert_{\diamond}\left \Vert\mathcal{R} \right \Vert_{\diamond}$, and 2) the diamond norm for a quantum channel is equal to one, i.e., for all quantum channels $\mathcal{Q}$, it holds that $\left \Vert\mathcal{Q} \right \Vert_{\diamond} = 1$. Now, if we want the final error in~\eqref{eq:second-part-simplify} to be at most $\frac{\varepsilon}{2}$, then it suffices to have the following:
\begin{equation}
    \frac{1}{2}\left \Vert 
      e^{\mathcal{M} c\tau} - \mathcal{R}_{c\tau} \right\Vert_{\diamond} \leq \frac{\varepsilon}{2n}.
\end{equation}

Recall that $\mathcal{R}_{c\tau}$ is an LCU-based quantum algorithm proposed in~\cite{cleve2019efficient} for simulating Lindbladian channels. In our case, the channel of interest is $e^{\mathcal{M} c\tau}$. The algorithm $\mathcal{R}_{c\tau}$ assumes an input model where the Lindblad operators are expressed as linear combinations of Pauli strings. Therefore, before applying the algorithm, we need to first express the Lindblad operators of the Lindbladian $\mathcal{M}$ into this required form, which we have in Appendix~\ref{app:WML_fixed_inter_unitaries}.

Observe that there are 16 terms in~\eqref{eq:M_i-exp-full}. This implies that there are $16^q$ or $2^{4q}$ terms in the linear-combination expression for $M$. This resolves the number of terms in the linear combination expression of $M$.

Additionally, note that a coefficient $\alpha_i$ in the linear-combination expression for $M$ is either $+1/2^{q/2}$ or $-1/2^{q/2}$, which is clear to see from~\eqref{eq:decomp-M-M_i} and~\eqref{eq:M_i-exp-full}. Using this fact, we resolve the quantity $\left \Vert \mathcal{M} \right\Vert_{\operatorname{Pauli}}$ in the following way:
\begin{align}
    \left \Vert \mathcal{M} \right\Vert_{\operatorname{Pauli}} & \coloneqq \left(\sum_{i=0}^{2^{4q} - 1} \alpha_i \right)^2 \leq \left(\sum_{i=0}^{2^{4q} - 1}\frac{1}{2^{q/2}} \right)^2\\
    & =\left(\frac{1}{2^{q/2}} 2^{4q} \right)^2 = 2^{7q}.
\end{align}

Using the development above and Theorem~1 of~\cite{cleve2019efficient}, we can say that the gate complexity $G$ of the algorithm $\mathcal{R}_{c\tau}$ for implementing the channel $e^{\mathcal{M}c\tau}$ such that $\left \Vert e^{\mathcal{M} c\tau} - \mathcal{R}_{c\tau} \right\Vert_{\diamond} \leq \varepsilon/n$ holds is given as follows:
\begin{equation}
    G = O\!\left( 2^{15q} c\tau  \frac{\left(\ln(2^{15q}nc\tau /\varepsilon) + q\right)\ln(n c\tau /\varepsilon)}{\ln \ln(Nc\tau /\varepsilon)}\right) = O\!\left(\frac{\ln^2(n/\varepsilon)}{\ln\ln(n /\varepsilon)}\right) = O\!\left(\frac{\ln^2(ct/\varepsilon)}{\ln\ln(ct /\varepsilon)}\right),
\end{equation}
where the second equality holds because $q = O(1)$ and $c\tau \leq 1$. This implies that the total gate complexity of the full algorithm is 
\begin{equation}
    n \cdot G = O\!\left(\frac{c^2 t^2 \ln^2(ct/\varepsilon)}{\varepsilon \ln\ln(ct /\varepsilon)}\right).
\end{equation}
This is the expression for the gate complexity of the LCU-based WML algorithm claimed in the statement of Theorem~\ref{thm:gate-comp-WML}, and thus concludes its proof.

\section{Proof of Theorem~\ref{thm:split_j_matrix}}\label{app:split-J-matrix-analysis}

To analyze the performance of the Split-$J$ Matrix algorithm, we need to bound the following quantity from above:
\begin{equation}
    \left\Vert e^{\mathcal{L}t} - \left(e^{\mathcal{H}t/n}\circ \left(\prod_{q=1}^{Q} e^{\mathcal{H}'_{q}t/n}\right) \circ \mathcal{J}_{1}(t/n)\circ\cdots\circ\mathcal{J}_{K}(t/n) \right)^{\circ n}\right\Vert_{\diamond}.
\end{equation}
Using the fact that the diamond distance obeys subadditivity under composition, we get
 \begin{align}
    & \left\Vert e^{\mathcal{L}t} - \left(e^{\mathcal{H}t/n}\circ \left(\prod_{q=1}^{Q} e^{\mathcal{H}'_{q}t/n}\right) \circ \mathcal{J}_{1}(t/n)\circ\cdots\circ\mathcal{J}_{K}(t/n) \right)^{\circ n}\right\Vert_{\diamond}\notag\\
    & \leq n\left\Vert e^{\mathcal{L}t/n} - e^{\mathcal{H}t/n}\circ \left(\prod_{q=1}^{Q} e^{\mathcal{H}'_{q}t/n}\right) \circ \mathcal{J}_{1}(t/n)\circ\cdots\circ\mathcal{J}_{K}(t/n) \right\Vert_{\diamond}\\
    & = n\Big\Vert e^{\mathcal{L}t/n} - e^{\mathcal{H}t/n}\circ e^{\mathcal{H}'t/n} \circ e^{\mathcal{N}_1 t/n}\circ\cdots\circ e^{\mathcal{N}_K t/n} \notag\\
    & \qquad + e^{\mathcal{H}t/n}\circ e^{\mathcal{H}'t/n} \circ e^{\mathcal{N}_1 t/n}\circ\cdots\circ e^{\mathcal{N}_K t/n}  - e^{\mathcal{H}t/n}\circ \left(\prod_{q=1}^{Q} e^{\mathcal{H}'_{q}t/n}\right) \circ \mathcal{J}_{1}(t/n)\circ\cdots\circ\mathcal{J}_{K}(t/n) \Big\Vert_{\diamond}\\
    & \leq n\left\Vert e^{\mathcal{L}t/n} - e^{\mathcal{H}t/n}\circ e^{\mathcal{H}'t/n} \circ e^{\mathcal{N}_1 t/n}\circ\cdots\circ e^{\mathcal{N}_K t/n}\right\Vert_{\diamond} \notag\\
    & \qquad + n\left\Vert e^{\mathcal{H}t/n}\circ e^{\mathcal{H}'t/n} \circ e^{\mathcal{N}_1 t/n}\circ\cdots\circ e^{\mathcal{N}_K t/n}  - e^{\mathcal{H}t/n}\circ \left(\prod_{q=1}^{Q} e^{\mathcal{H}'_{q}t/n}\right) \circ \mathcal{J}_{1}(t/n)\circ\cdots\circ\mathcal{J}_{K}(t/n) \right\Vert_{\diamond}\\
    & \leq n\left\Vert e^{\mathcal{L}t/n} - e^{\mathcal{H}t/n}\circ e^{\mathcal{H}'t/n} \circ e^{\mathcal{N}_1 t/n}\circ\cdots\circ e^{\mathcal{N}_K t/n}\right\Vert_{\diamond} \notag\\
    & \qquad + n\left \Vert e^{\mathcal{N}_1 t/n}\circ\cdots\circ e^{\mathcal{N}_K t/n}  - \mathcal{J}_{1}(t/n)\circ\cdots\circ\mathcal{J}_{K}(t/n) \right\Vert_{\diamond} + n\left\Vert  e^{\mathcal{H}'t/n} - \left(\prod_{q=1}^{Q} e^{\mathcal{H}'_{q}t/n}\right) \right\Vert_{\diamond}\\
    & = n\left\Vert e^{\mathcal{L}t/n} - e^{\mathcal{H}t/n}\circ e^{\mathcal{H}'t/n} \circ e^{\mathcal{N}_1 t/n}\circ\cdots\circ e^{\mathcal{N}_K t/n}\right\Vert_{\diamond} \notag\\
    & \qquad + n\left \Vert e^{\mathcal{N}_1 t/n}\circ\cdots\circ e^{\mathcal{N}_K t/n}  - \mathcal{J}_{1}(t/n)\circ\cdots\circ\mathcal{J}_{K}(t/n) \right\Vert_{\diamond} + O\!\left(\frac{Q^2 \lambda_{\max}^2t^2}{n}\right) ,
    \label{eq:two-terms-split-J-matrix-analysis}
\end{align}
where we obtain the second inequality by using the triangle inequality, the third inequality by using the subadditivity under composition property of the diamond distance, and the last equality follows from the standard error analysis for the first-order Trotter for Hamiltonian simulation~\cite[Equation~4]{Childs2019}, with $\lambda_{\max}$ defined in \eqref{eq:split-lambda}.

\subsection{Bounding the First Term of~\eqref{eq:two-terms-split-J-matrix-analysis}}

Consider the following:
\begin{align}
    & n\left\Vert e^{\mathcal{L}t/n} - e^{\mathcal{H}t/n}\circ e^{\mathcal{H}'t/n} \circ e^{\mathcal{N}_1 t/n}\circ\cdots\circ e^{\mathcal{N}_K t/n}\right\Vert_{\diamond}\notag\\
    & = n\left\Vert e^{\mathcal{L}t/n} - e^{\mathcal{H}t/n}\circ e^{\mathcal{H}'t/n} \circ e^{\left(\mathcal{N}_1 + \cdots + \mathcal{N}_K \right) t/n}\right\Vert_{\diamond}\\
    & = n\left\Vert e^{\mathcal{L}t/n} - e^{\mathcal{H}t/n}\circ e^{\mathcal{H}'t/n} \circ e^{\mathcal{N} t/n}\right\Vert_{\diamond}\\
    & = n\left\Vert e^{\left(\mathcal{H}+  \mathcal{H}' + \mathcal{N}\right)t/n} - e^{\mathcal{H}t/n}\circ e^{\mathcal{H}'t/n} \circ e^{\mathcal{N} t/n}\right\Vert_{\diamond}\\
    & = n\Bigg\Vert \mathcal{I} + \left(\mathcal{H}+  \mathcal{H}' + \mathcal{N}\right)\frac{t}{n} + \sum_{r=2}^{\infty} \frac{\left(\mathcal{H}+  \mathcal{H}' + \mathcal{N}\right)^r}{r!}\left(\frac{t}{n}\right)^r \notag\\
    & \qquad\qquad\qquad\qquad\qquad  - \left(\mathcal{I} + \left(\mathcal{H}+  \mathcal{H}' + \mathcal{N}\right)\frac{t}{n} + \sum_{r=2}^{\infty}\sum_{\substack{r_1, r_2, r_3=0:\\r_1 +r_2+r_3=r}}^{\infty} \frac{\mathcal{H}^{r_1}\mathcal{H}'^{r_2}\mathcal{N}^{r_3}}{r_1!r_2!r_3!}\left(\frac{t}{n}\right)^{r}\right)\Bigg\Vert_{\diamond}\\
    & = n\left\Vert \sum_{r=2}^{\infty}\left( \frac{\left(\mathcal{H}+  \mathcal{H}' + \mathcal{N}\right)^r}{r!}\left(\frac{t}{n}\right)^r - \sum_{\substack{r_1, r_2, r_3=0:\\r_1 +r_2+r_3=r}}^{\infty} \frac{\mathcal{H}^{r_1}\mathcal{H}'^{r_2}\mathcal{N}^{r_3}}{r_1!r_2!r_3!}\left(\frac{t}{n}\right)^{r}\right)\right\Vert_{\diamond}\\
    & \leq n \sum_{r=2}^{\infty}\left\Vert \frac{\left(\mathcal{H}+  \mathcal{H}' + \mathcal{N}\right)^r}{r!}\left(\frac{t}{n}\right)^r - \sum_{\substack{r_1, r_2, r_3=0:\\r_1 +r_2+r_3=r}}^{\infty} \frac{\mathcal{H}^{r_1}\mathcal{H}'^{r_2}\mathcal{N}^{r_3}}{r_1!r_2!r_3!}\left(\frac{t}{n}\right)^{r}\right\Vert_{\diamond}\\
    & \leq n \sum_{r=2}^{\infty}\left(\left\Vert \frac{\left(\mathcal{H}+  \mathcal{H}' + \mathcal{N}\right)^r}{r!}\left(\frac{t}{n}\right)^r \right\Vert_{\diamond} + \left \Vert \sum_{\substack{r_1, r_2, r_3=0:\\r_1 +r_2+r_3=r}}^{\infty} \frac{\mathcal{H}^{r_1}\mathcal{H}'^{r_2}\mathcal{N}^{r_3}}{r_1!r_2!r_3!}\left(\frac{t}{n}\right)^{r}\right\Vert_{\diamond}\right)\\
    & \leq n \sum_{r=2}^{\infty} \left(\frac{\left(\left\Vert\mathcal{H}+  \mathcal{H}' + \mathcal{N}\right\Vert_{\diamond}\right)^r}{r!}\left(\frac{t}{n}\right)^r  + \sum_{\substack{r_1, r_2, r_3=0:\\r_1 +r_2+r_3=r}}^{\infty} \frac{\left\Vert\mathcal{H}\right\Vert_{\diamond}^{r_1}\left \Vert\mathcal{H}'\right\Vert_{\diamond}^{r_2}\left\Vert\mathcal{N}\right\Vert_{\diamond}^{r_3}}{r_1!r_2!r_3!}\left(\frac{t}{n}\right)^{r}\right)\\
    & \leq n \sum_{r=2}^{\infty}\left( \frac{\left(3 \left\Vert\mathcal{L}\right\Vert_{\max}\right)^r}{r!}\left(\frac{t}{n}\right)^r  + \sum_{\substack{r_1, r_2, r_3=0:\\r_1 +r_2+r_3=r}}^{\infty} \frac{\left\Vert\mathcal{L}\right\Vert_{\max}^{r_1}\left\Vert\mathcal{L}\right\Vert_{\max}^{r_2}\left\Vert\mathcal{L}\right\Vert_{\max}^{r_3}}{r_1!r_2!r_3!}\left(\frac{t}{n}\right)^{r}\right)\\
    & = n \sum_{r=2}^{\infty}\left( \frac{3^r \left\Vert\mathcal{L}\right\Vert_{\max}^r}{r!}\left(\frac{t}{n}\right)^r  + \sum_{\substack{r_1, r_2, r_3=0:\\r_1 +r_2+r_3=r}}^{\infty} \frac{\left\Vert\mathcal{L}\right\Vert_{\max}^{r_1+r_2+r_3}}{r_1!r_2!r_3!}\left(\frac{t}{n}\right)^{r}\right)\\
    & = n \sum_{r=2}^{\infty}\left( \frac{3^r \left\Vert\mathcal{L}\right\Vert_{\max}^r}{r!}\left(\frac{t}{n}\right)^r  + \sum_{\substack{r_1, r_2, r_3=0:\\r_1 +r_2+r_3=r}}^{\infty} \frac{\left\Vert\mathcal{L}\right\Vert_{\max}^{r}}{r_1!r_2!r_3!}\left(\frac{t}{n}\right)^{r}\right)\\
    & = n \sum_{r=2}^{\infty}\left( \frac{3^r \left\Vert\mathcal{L}\right\Vert_{\max}^r}{r!}\left(\frac{t}{n}\right)^r  + \left\Vert\mathcal{L}\right\Vert_{\max}^{r}\left(\frac{t}{n}\right)^{r} \sum_{\substack{r_1, r_2, r_3=0:\\r_1 +r_2+r_3=r}}^{\infty} \frac{1}{r_1!r_2!r_3!}\right)\\
    & \leq n \sum_{r=2}^{\infty}\left( 3^r \left\Vert\mathcal{L}\right\Vert_{\max}^r\left(\frac{t}{n}\right)^r  + \left\Vert\mathcal{L}\right\Vert_{\max}^{r}\left(\frac{t}{n}\right)^{r} \sum_{\substack{r_1, r_2, r_3=0:\\r_1 +r_2+r_3=r}}^{\infty} 1\right),
\end{align}
where $\left\Vert\mathcal{L}\right\Vert_{\max} \coloneqq 2K\max\{\left\Vert H\right\Vert, \left\Vert H_{I, 1}\right\Vert, \ldots, \left\Vert H_{I, J}\right\Vert, \left\Vert L_1\right\Vert^2, \ldots, \left\Vert L_K\right\Vert^2 \}$. The first equality holds due to the fact that the Lindbladians $\mathcal{N}_1, \ldots, \mathcal{N}_K$ commute with each other. The first and second inequalities follow due to the triangle inequality. The third inequality follows due to the submultiplicativity of the diamond norm and the triangle inequality.
The number of ways to pick $r_1, r_2, r_3\geq 0$ such that $r_1 + r_2 + r_3 = r$ is given by $\binom{r+2}{2}$, and for $r\geq 2$, this number can be bounded from above by $\left(2^r-1\right)3^r$. Using this fact in the above inequality, we get
\begin{align}
    n\left\Vert e^{\mathcal{L}t/n} - e^{\mathcal{H}t/n}\circ e^{\mathcal{H}'t/n} \circ e^{\mathcal{N}_1 t/n}\circ\cdots\circ e^{\mathcal{N}_K t/n}\right\Vert_{\diamond} & \leq n \sum_{r=2}^{\infty}\left( 3^r \left\Vert\mathcal{L}\right\Vert_{\max}^r\left(\frac{t}{n}\right)^r  + \left\Vert\mathcal{L}\right\Vert_{\max}^{r}\left(\frac{t}{n}\right)^{r} \left(2^r-1\right)3^r\right)\\
    & = n \sum_{r=2}^{\infty}6^r \left\Vert\mathcal{L}\right\Vert_{\max}^r\left(\frac{t}{n}\right)^r \\
    & = n \frac{36 \left\Vert\mathcal{L}\right\Vert_{\max}^2\left(\frac{t}{n}\right)^2}{1-6 \left\Vert\mathcal{L}\right\Vert_{\max}\left(\frac{t}{n}\right)}\\
    & \leq 72 \left\Vert\mathcal{L}\right\Vert_{\max}^2\frac{t^2}{n}\label{eq:sJm-first-bound},
\end{align}
where, for the last inequality, we assume that $6 \left\Vert\mathcal{L}\right\Vert_{\max}\left(\frac{t}{n}\right) \leq 1/2$.

\subsection{Bounding the Second Term of~\eqref{eq:two-terms-split-J-matrix-analysis}}

Consider the following:
\begin{equation}
    n\left\Vert e^{\mathcal{N}_1 t/n}\circ\cdots\circ e^{\mathcal{N}_K t/n}  - \mathcal{J}_{1}(t/n)\circ\cdots\circ\mathcal{J}_{K}(t/n) \right\Vert_{\diamond} \leq n\sum_{k=1}^{K}\left\Vert e^{\mathcal{N}_{k} t/n}- \mathcal{J}_{k}(t/n)\right\Vert_{\diamond}. \label{eqn:multi_split_j}
\end{equation}
Now we can analyze the individual term $\left\Vert e^{\mathcal{N}_{k} t/n}- \mathcal{J}_{k}(t/n)\right\Vert_{\diamond}$. For this, we must first understand the action of $\mathcal{J}_{k}(t)$.
        Note that, up to $O(t^2)$, Hamiltonian simulation can be expressed as follows:
        \begin{equation}\label{eqn:hamilton_to_second_order}
            e^{-iHt}\rho e^{iHt}=\rho - it[H,\rho]+\frac{(it)^2}{2}[H,[H,\rho]]+\cdots.
        \end{equation}
        Given the definition of $\mathcal{J}_{k}(t)$, we can use~\eqref{eqn:hamilton_to_second_order} and get
        \begin{align}
            &e^{-iJ_{k}\sqrt{t}}\left(\rho\otimes|0\rangle\!\langle 0|_{A}\right) e^{iJ_{k}\sqrt{t}}\notag\\
            &=\left(\rho\otimes|0\rangle\!\langle 0|_{A}\right)- i\sqrt{t}[J_{k},\left(\rho\otimes|0\rangle\!\langle 0|_{A}\right)]+\frac{(i\sqrt{t})^2}{2}[J_{k},[J_{k},\left(\rho\otimes|0\rangle\!\langle 0|_{A}\right)]] + \cdots\\
        	&=\left(\rho\otimes|0\rangle\!\langle 0|_{A}\right)- i\sqrt{t} \left[\left(L^\dagger_{k}\otimes |0\rangle\!\langle 1|_{A} + L_{k}\otimes |1\rangle\!\langle 0|_{A}\right),\left(\rho\otimes|0\rangle\!\langle 0|_{A}\right)\right]\notag\\
        	&\qquad+\frac{(i\sqrt{t})^2}{2}\left[\left(L^\dagger_{k}\otimes |0\rangle\!\langle 1|_{A} + L_{k}\otimes |1\rangle\!\langle 0|_{A}\right),\left[\left(L^\dagger_{k}\otimes |0\rangle\!\langle 1|_{A} + L_{k}\otimes |1\rangle\!\langle 0|_{A}\right),\left(\rho\otimes|0\rangle\!\langle 0|_{A}\right)\right]\right] + \cdots\label{eqn:any_lind_no_trace}.
        \end{align}
        Now consider the individual terms in~\eqref{eqn:any_lind_no_trace}. The first commutator can be simplified as follows:
        \begin{align}
        	&\left[\left(L^\dagger_{k}\otimes |0\rangle\!\langle 1|_{A} + L_{k}\otimes |1\rangle\!\langle 0|_{A}\right),\left(\rho\otimes|0\rangle\!\langle 0|_{A}\right)\right]\notag\\
        	&=  L_{k}\otimes |1\rangle\!\langle 0|_{A}\left(\rho\otimes|0\rangle\!\langle 0|_{A}\right)- \left(\rho\otimes|0\rangle\!\langle 0|_{A}\right) L^\dagger_{k}\otimes |0\rangle\!\langle 1|_{A}\\
        	&=  L_{k}\rho\otimes |1\rangle\!\langle 0|_{A}-\rho L^\dagger_{k}\otimes |0\rangle\!\langle 1|_{A}.
        \end{align}
        The second commutator can be simplified as follows:
        \begin{align}
        	&\left[\left(L^\dagger_{k}\otimes |0\rangle\!\langle 1|_{A} + L_{k}\otimes |1\rangle\!\langle 0|_{A}\right),\left[\left(L^\dagger_{k}\otimes |0\rangle\!\langle 1|_{A} + L_{k}\otimes |1\rangle\!\langle 0|_{A}\right),\left(\rho\otimes|0\rangle\!\langle 0|_{A}\right)\right]\right] \notag \\
        	&=\left[\left(L^\dagger_{k}\otimes |0\rangle\!\langle 1|_{A} + L_{k}\otimes |1\rangle\!\langle 0|_{A}\right),L_{k}\rho\otimes |1\rangle\!\langle 0|_{A}-\rho L^\dagger_{k}\otimes |0\rangle\!\langle 1|_{A}\right]\\
        	&=L^\dagger_{k}L_{k}\rho\otimes |0\rangle\!\langle 0|_{A} + \rho L^\dagger_{k}L_{k} \otimes |0\rangle\!\langle 0|_{A} - L_{k}\rho L^\dagger_{k}\otimes |1\rangle\!\langle 1|_{A} - L_{k}\rho L^\dagger_{k}\otimes |1\rangle\!\langle 1|_{A} 
        \end{align}
        After substituting the appropriate terms in~\eqref{eqn:any_lind_no_trace} and tracing out the auxiliary system,  we get
        \begin{align}
        	\mathcal{J}_{k}(t)(\rho) & = \operatorname{Tr}_{A}\!\left[e^{-iJ_{k}\sqrt{t}}\left(\rho\otimes|0\rangle\!\langle 0|_{A}\right)e^{iJ_{k}\sqrt{t}}\right] \\
            & =\rho + t\left(L^\dagger_{k} \rho L_{k} -\frac{1}{2}\left\{L_{k}L^\dagger_{k},\rho\right\}\right) + \operatorname{Tr}_{A}\!\left[\sum_{m=3}^{\infty}\frac{t^{m/2}}{m! n^m }\tilde{\mathcal{J}}_{k}^{m}(\rho)\right]\\
            & = \rho + \mathcal{N}_k(\rho) t + \operatorname{Tr}_{A}\!\left[\sum_{m=3}^{\infty}\frac{t^{m/2}}{m! n^m }\tilde{\mathcal{J}}_{k}^{m}(\rho)\right],
        \end{align}
        where $\tilde{\mathcal{J}}_{k}^{m}(\rho)=[J_{k},\tilde{\mathcal{J}}_{k}^{m-1}(\rho)]$ and $\tilde{\mathcal{J}}_{k}(\rho\otimes|0\rangle\!\langle 0|_{A})=[J_{k},\rho\otimes|0\rangle\!\langle 0|_{A}]$. 
        Note that terms with $\sqrt{t}$ and $(\sqrt{t})^3$ do not contribute to the above equation. Using the Taylor expansion of both $e^{\mathcal{N}_{k} t/n}$ and $\mathcal{J}_{k}(t/n)$, and plugging these back into $\left\Vert e^{\mathcal{N}_{k} t/n}- \mathcal{J}_{k}(t/n)\right\Vert_{\diamond}$, we get
        \begin{align}
            \left\Vert e^{\mathcal{N}_{k} t/n}- \mathcal{J}_{k}(t/n)\right\Vert_{\diamond}= \left\Vert\mathcal{I}+\frac{t}{n}\mathcal{N}_{k}+\sum_{m=2}^{\infty}\frac{t^m}{m!n^m}\mathcal{N}_{k}^{m}-\mathcal{I}-\frac{t}{n}\mathcal{N}_{k}-\operatorname{Tr}_{A}\!\left(\sum_{m=3}^{\infty}\frac{t^{m/2}}{m! n^m }\tilde{\mathcal{J}}_{k}^{m}\right)\right\Vert_{\diamond}
        \end{align}
        Then,
        \begin{align}
            \left\Vert e^{\mathcal{N}_{k} t/m}- \mathcal{J}_{k}(t/m)\right\Vert_{\diamond}&=\left\Vert\sum_{m=2}^{\infty}\frac{t^m}{m!n^m}\mathcal{N}_{k}^{m}-\operatorname{Tr}_{A}\!\left(\sum_{m=3}^{\infty}\frac{t^{m/2}}{n!n^m}\tilde{\mathcal{J}}_{k}^{m}\right)\right\Vert_{\diamond}\\
            &\leq\left\Vert\sum_{m=2}^{\infty}\frac{t^m}{m!n^m }\mathcal{N}_{k}^{m}\right\Vert_{\diamond}+\left\Vert\sum_{m=3}^{\infty}\frac{t^{m/2}}{m!n^m}\operatorname{Tr}_{A}\tilde{\mathcal{J}}_{k}^{m}\right\Vert_{\diamond}\\
            &\leq\sum_{m=2}^{\infty}\frac{t^m}{m!n^m }\left\Vert\mathcal{N}_{k}^{m}\right\Vert_{\diamond}+\sum_{m=3}^{\infty}\frac{t^{m/2}}{m!n^m}\left\Vert\operatorname{Tr}_{A}\tilde{\mathcal{J}}_{k}^{m}\right\Vert_{\diamond}\\
            &\leq\sum_{m=2}^{\infty}\frac{t^m}{m!n^m }\left\Vert\mathcal{N}_{k}\right\Vert^{m}_{\diamond}+\sum_{m=3}^{\infty}\frac{t^{m/2}}{m!n^m}\left\Vert\operatorname{Tr}_{A}\tilde{\mathcal{J}}_{k}^{m}\right\Vert_{\diamond}.\label{eq:alomst_final_sum_split_j}
        \end{align}
        The first and second inequalities hold due to the triangle inequality, and the final inequality holds due to the sub-multiplicativity of the diamond norm. Note that, $\left\Vert\mathcal{N}_{k}\right\Vert_{\diamond}\leq 2\left\Vert L_{k}\right\Vert^2 \leq 2\lambda_{\max}$. Now consider $\Vert\operatorname{Tr}_{A}\tilde{\mathcal{J}}_{k}^{m}\Vert_{\diamond}$. If $m$ is odd,  $\Vert\operatorname{Tr}_{A}\tilde{\mathcal{J}}_{k}^{m}\Vert_{\diamond}= 0$ because, after the partial trace on $A$, these terms do not contribute to $\mathcal{J}_{k}$. If $m$ is even, 
        \begin{equation}
            \left\Vert\operatorname{Tr}_{A}\tilde{\mathcal{J}}_{k}^{m}\right\Vert_{\diamond}=\left\Vert\operatorname{Tr}_{A}\!\left(\tilde{\mathcal{J}}_{k}\circ \tilde{\mathcal{J}}_{k}\right)^{m'}\right\Vert_{\diamond}\leq \left(4\left\Vert L_{k}\right\Vert^{2}\right)^{m'}\leq  (4\lambda_{\max})^{m'},
        \end{equation}
        where $m' \coloneqq m/2$. Substituting these bounds on $\left\Vert\mathcal{N}_{k}\right\Vert_{\diamond}$ and  $\left\Vert\operatorname{Tr}_{A}\tilde{\mathcal{J}}_{k}^{m}\right\Vert_{\diamond}$ in~\eqref{eq:alomst_final_sum_split_j}, we get
        \begin{align}
            \left\Vert e^{\mathcal{N}_{k} t/n}- \mathcal{J}_{k}(t/n)\right\Vert_{\diamond} &\leq \sum_{m=2}^{\infty}\frac{2^m\lambda_{\max}^m t^m}{m!n^m}+\sum_{m'=2}^{\infty}\frac{4^{m'}\lambda_{\max}^{m'} t^{m'}}{(2m')!n^{2m'}}\\
            &\leq \sum_{m=2}^{\infty}\frac{2^m\lambda_{\max}^m t^m}{m!n^m}+\sum_{m'=2}^{\infty}\frac{4^{m'}\lambda_{\max}^{m'} t^{m'}}{m'!n^{m'}}\\
            & \leq \frac{4 \lambda_{\max}^2 t^2}{n^2}e^{2\lambda_{\max} t/n} + \frac{16\lambda_{\max}^2 t^2}{n^2}e^{4\lambda_{\max} t/n}\\
            & \leq \frac{4\lambda_{\max}^2 t^2}{n^2}e^{4\lambda_{\max} t/n} + \frac{16\lambda_{\max}^2 t^2}{n^2}e^{4\lambda_{\max} t/n}\\
            &= \frac{20\lambda_{\max}^2 t^2}{n^2}e^{4\lambda_{\max} t/n}. \label{eq:single_split_j}
        \end{align}
        Substituting~\eqref{eq:single_split_j} in~\eqref{eqn:multi_split_j}, we get
        \begin{align}
            \left\Vert e^{\mathcal{N}t} - \left(\mathcal{J}_{1}(t)\circ\cdots\circ\mathcal{J}_{K}(t)\right)\right\Vert_{\diamond}&\leq n\sum_{k=1}^{K}\frac{20\lambda_{\max}^2 t^2}{n^2}e^{4\lambda_{\max} t/n}\\
            &\leq \sum_{k=1}^{K}\frac{20\lambda_{\max}^2 t^2}{n}e^{4\lambda_{\max} t/n}\\
            &= \frac{20 K\lambda_{\max}^2 t^2}{n}e^{4\lambda_{\max} t/n}.
        \end{align}
        When $n$ is large enough that $e^{4\lambda_{\max} t/n}\approx 1$, we get  
        \begin{align}
            \left\Vert e^{\mathcal{N}t} - \left(\mathcal{J}_{1}(t)\circ\cdots\circ\mathcal{J}_{K}(t)\right)\right\Vert_{\diamond}\leq 20K\lambda_{\max}^2\frac{ t^2}{n}\label{eq:sJm-sec-bound}.
        \end{align}

\subsection{Final Bound and Gate Complexity}

       From~\eqref{eq:two-terms-split-J-matrix-analysis},~\eqref{eq:sJm-first-bound}, the above inequality, and normalizing the diamond distance on the left-hand side of~\eqref{eq:two-terms-split-J-matrix-analysis}, we finally get
       \begin{align}
           & \frac{1}{2}\left\Vert e^{\mathcal{L}t} - \left(e^{\mathcal{H}t/n}\circ \left(\prod_{q=1}^{Q} e^{\mathcal{H}'_{q}t/n}\right)  \circ \mathcal{J}_{1}(t/n)\circ\cdots\circ\mathcal{J}_{K}(t/n) \right)^{\circ n}\right\Vert_{\diamond} \notag\\
           & \leq \frac{1}{2}\left(72 \left\Vert\mathcal{L}\right\Vert_{\max}^2\frac{t^2}{n} + 20K\lambda_{\max}^2\frac{ t^2}{n} + O\!\left(\frac{Q^2 \lambda_{\max}^2t^2}{n}\right) \right).
       \end{align}
        If we require that our final simulation error is at most $\varepsilon$, then 
        \begin{align}
            n& \geq \frac{1}{2}\left(72 \left\Vert\mathcal{L}\right\Vert_{\max}^2\frac{t^2}{n} + 20K\lambda_{\max}^2\frac{ t^2}{n} + O\!\left(\frac{Q^2 \lambda_{\max}^2t^2}{n}\right) \right)\\
            & =O\!\left(\frac{K^2 \lambda_{\max}^2t^2}{\varepsilon}\right)+  O\!\left(\frac{Q^2 \lambda_{\max}^2t^2}{\varepsilon}\right)\\
            & = O\!\left(\frac{ (K^2 + Q^2)\lambda_{\max}^2 t^2}{\varepsilon}\right),
        \end{align}
where, in the first equality, we use the fact that $\left\Vert\mathcal{L}\right\Vert_{\max} = 2K\lambda_{\max}$.

Given $n$, we can now directly compute the gate complexity of the Split $J$-Matrix algorithm from the channel form of this algorithm given by~\eqref{eq:sJm-channel}. Note that the unitary $e^{-iH\tau}$, where $H$ is a local Hamiltonian, acting on a constant number of qubits, and $\tau$ is some time, can be implemented using $O(1)$ number of one- and two-qubit gates. With this in mind, we can determine the number of one- and two-qubit gates required to implement the different components of the Split $J$-Matrix channel \eqref{eq:sJm-channel} as follows:
\begin{align}
    e^{\mathcal{H}t/n} = \prod_{p=1}^{P} e^{\mathcal{H}_{p}t/n} & \longrightarrow O(P) \text{ gates}\\
    \prod_{q=1}^{Q} e^{\mathcal{H}'_{q}t/n} & \longrightarrow O(Q) \text{ gates}\\
    \mathcal{J}_{1}(t/n)\circ\cdots\circ\mathcal{J}_{K}(t/n) & \longrightarrow O(K) \text{ gates}.
\end{align}
Therefore, the total gate complexity of the Split $J$-Matrix algorithm is
\begin{align}
    O(n(P+Q+K)) = O\left(\frac{ (P+Q+K)(K^2 + Q^2)\lambda_{\max}^2 t^2}{\varepsilon}\right).
\end{align}

This concludes the proof of Theorem~\ref{thm:split_j_matrix}.

\section{Calculating \texorpdfstring{$c$}{Lg} in~\eqref{def:WML_c} as a Function of \texorpdfstring{$N$}{Lg} and \texorpdfstring{$R$}{Lg}}\label{app:constants}
    Recall that $N$ is the number of emitters included in the open Tavis--Cummings model. Let $R$ be the number of excitations allowed within the cavity. To express $c$, as defined in~\eqref{def:WML_c}, in terms of $N$ and $R$, we first calculate $\Vert a \Vert_{2}^{2}$. Note that when $R$ excitations are allowed in the cavity, $a$ can be expressed in the Fock basis as follows:
    \begin{equation}\label{def:general_annhil}
        a  = |0\rangle\!\langle 1|+\sqrt{2}\,|1\rangle\!\langle 2|+\sqrt{3}\,|2\rangle\!\langle 3| + \cdots + \sqrt{R}\,|R-1\rangle\!\langle R|.
    \end{equation}
    Using~\eqref{def:general_annhil} and~\eqref{eqn:p-norm}, we get 
    \begin{equation}
        \Vert a \Vert_{2}^{2} = \operatorname{Tr}\!\left[a^\dagger a\right] = \operatorname{Tr}\!\left[\sum_{r=1}^{R}r\,|r\rangle\!\langle r|\right] = \frac{R(R+1)}{2} \leq 2R^2
    \end{equation}
    Similarly, using~\eqref{eqn:annihil_emitter} and~\eqref{eqn:p-norm}, we get
    \begin{equation}
        \Vert \sigma_{j}^{-} \Vert_{2}^{2} = 1.
    \end{equation}
    From \ref{subsec:background_DME_WML}, we can infer the weights $c_j$ associated with each of the program states. Hence, we get
    \begin{equation}
        \sum_{j}c_{j}= R(R-1)\frac{\omega_{C}}{2} + \omega_{E}N + \sum_{r=1}^{R}2g\sqrt{r} + \sum_{r=1}^{R}2E_{P}\sqrt{r}.
    \end{equation}
    We can bound the above summation as follows:
    \begin{equation}
        \sum_{j}c_{j}\leq R^{2}\omega_{C} + \omega_{E}N + 2g R\sqrt{R} + 2E_{P}R\sqrt{R}.
    \end{equation}
    Therefore, for the Tavis--Cummings model, $c$ can be bound as follows:
    \begin{equation}
        c \leq (\kappa+ \omega_{C})R^{2} + (\gamma  + \omega_{E})N + 2(g  + E_{P})R\sqrt{R}.
    \end{equation}

\section{Computational Complexities of Basic Classical Methods \label{app:classical-complexity}}

In this appendix, we briefly describe the simplest classical methods of simulating open quantum systems, and we also give their associated space and time complexities when applied to the open Tavis--Cummings model.
We show that the two most basic approaches to simulating open systems dynamics (solving the Lindblad master equation in Liouville space and the wavefunction Monte Carlo method, which QuTiP employs) incur exponential space and time cost. For a comprehensive review of classical methods for simulating open systems dynamics, see~\cite{RevModPhys.93.015008}.            

\subsection{Lindblad Master Equation in Liouville Space}

As in the main text, we seek to solve~\eqref{eqn:open_TC_master}, which contains terms representing both unitary and dissipative time evolution.
For tractability, we truncate the Hilbert space of the cavity to $D$ dimensions, which allows for simulations containing up to $D-1$ photons.
With $N$ two-level quantum emitters coupled to the cavity, the dimensionality of the Hilbert space $\mathcal{H}$ of the entire Tavis--Cummings system is $2^N D$.

We again present the Lindblad master equation as given in~\eqref{eqn:open_TC_master}, rewritten as
\begin{equation}
    \dot{\rho} = \mathcal{L}(\rho),
\end{equation}
where $\mathcal{L}$ is a superoperator.
Let $L(\mathcal{H})$ be the space of linear operators from an input Hilbert space $\mathcal{H}$ to $\mathcal{H}$ itself, so that an operator $X\in L(\mathcal{H})$ takes states in $\mathcal{H}$ to states in $\mathcal{H}$.
A superoperator, then, is a function $\mathcal{N} \colon L(\mathcal{H}) \to L(\mathcal{H})$ which map operators to operators.
The Hilbert space of operators is also known as Liouville space, and its dimensionality is the square of the dimensionality of elements of $L(\mathcal{H})$: $(2^N D)^2 = 2^{2N} D^2$.
An operator $\rho \in L(\mathcal{H})$ can be converted to an element of the corresponding Liouville space by vectorizing it, i.e., ``column-stacking"~\cite{10.1063/1.1518555}. For example, for a $2 \times 2$ operator,
\begin{equation}
\rho = \begin{pmatrix} \rho_{11} & \rho_{12} \\ \rho_{21} & \rho_{22} \end{pmatrix} \quad \mapsto \quad \lvert \rho \rangle\rangle = \begin{pmatrix} \rho_{11} \\ \rho_{21} \\ \rho_{12} \\ \rho_{22} \end{pmatrix},
\end{equation}
where $|\rho\rangle\rangle$ denotes the vectorized form of $\rho$.
In making this transformation, operations on $\rho$ (superoperators) transform as follows:
\begin{equation}
    A \rho B \quad \mapsto \quad \left( B^{\intercal} \otimes A \right) \lvert \rho \rangle\rangle,
\end{equation}
where $B^{\intercal}$ indicates the transpose of $B$.
Importantly, the right-hand side can be written as a matrix acting on $|\rho\rangle\rangle$.

Using this transformation, the superoperator terms of the general Lindblad master equation~\eqref{eqn:general_master} can be written as
\begin{align}
    [H,\rho] = H\rho - \rho H \quad &\mapsto \quad \mathbb{H} \lvert \rho \rangle\rangle = ((I\otimes H) - (H^{\intercal}\otimes I)) \lvert \rho \rangle\rangle, \\
    \mathcal{L}_L (\rho) = L \rho L^\dagger - \frac{1}{2}\left\{ L^\dagger L, \rho \right\} \quad&\mapsto\quad \mathbb{L}_L|\rho\rangle\rangle = \left[(L^{*}\otimes L) - \frac{1}{2}(I\otimes L^\dagger L) - \frac{1}{2}(L^{\intercal}L^{*}\otimes I)\right] \lvert \rho \rangle\rangle,
\end{align}
where $L^*$ denotes the complex conjugate of $L$, and $I$ is the identity matrix.
so that the Lindblad master equation for an open Tavis--Cummings system~\eqref{eqn:open_TC_master} becomes
\begin{align}
    \lvert \dot{\rho} \rangle\rangle = \mathbb{L}~ \lvert \rho \rangle\rangle \equiv \left( -i\mathbb{H} + \mathbb{L}_a + \sum_{j=1}^N \mathbb{L}_{\sigma^-_j} \right) \lvert \rho \rangle\rangle.    
\end{align}
where, notably, $\mathbb{L}$ is a ($2^{2N} D^2$)-dimensional matrix.
The solution of this Liouville space master equation is then given by,
\begin{align}
    \lvert \rho(t) \rangle\rangle = e^{\mathbb{L}t} \lvert \rho (t=0) \rangle\rangle,
\end{align}

Steps to obtain $e^{\mathbb{L}t}$ are:
\begin{enumerate}
    \item Diagonalize $\mathbb{L}$ so that it can be written $\mathbb{L} = A \Lambda A^{-1}$, where $\Lambda = \operatorname{diag}(\lambda_1, \lambda_2, \dots)$ is the diagonal matrix of the eigenvalues of $\mathbb{L}$.
    \item Compute $e^{\mathbb{L}t} = A e^{\Lambda t} A^{-1} = A \operatorname{diag}(\{e^{\lambda_i t}\}) A^{-1}$.
\end{enumerate}
The matrix operations involved in computing $e^{\mathbb{L}t} \lvert \rho(t=0) \rangle\rangle$ are thus diagonalization, two matrix-matrix multiplications, and a matrix-vector multiplication.
If $d$ is the dimension of $\mathbb{L}$, these operations have time complexities $O(d^3)$, $O(d^3)$, and $O(d^2)$, respectively~\cite{golub2013matrix}.
Therefore, the first two operations dominate the time complexity, for an overall scaling of $O(d^3) = O(2^{6N} D^6)$.
The space resources required by this algorithm are proportional to the memory needed to record $\mathbb{L}$, which is $O(d^2) = O(2^{4N} D^4)$.
Thus, this basic method of classical simulation has time and space costs which are both exponential in the number of emitters, with the caveat that we have not considered any potential methods for taking advantage of the sparsity of $\mathbb{L}$.

\subsection{Wavefunction Monte Carlo Method}

The wavefunction Monte Carlo method, also known as the quantum jump or quantum trajectories method~\cite{PhysRevLett.68.580, PhysRevA.45.4879, carmichael2009open, breuer2002theory}, improves the computational complexity of evolving the vectorized density matrix in Liouville space by doing the following.
Consider the spectral decomposition of the density matrix $\rho = \sum_{i} \lambda_{i} \lvert \phi_{i} \rangle \! \langle \phi_{i} \rvert$. The key idea, then, is to evolve each of the pure states $\lvert \phi_{i} \rangle$ in this decomposition according to an ``effective" or ``conditional" Hamiltonian. 

The terms in a Lindbladian dissipator~\eqref{eqn:lindbladian} can be grouped into two types: $L \rho L^\dagger$ represents quantum jumps, where the system transitions between states, while $\{L^\dagger L, \rho\}$ represents the gradual loss of coherence in the system.
This allows us to define a non-Hermitian effective Hamiltonian:
\begin{equation}
    H_{\text{eff}} \coloneqq H - \frac{i}{2} \sum_{k}^K \gamma_{k} L^{\dag}_{k} L_{k},
\end{equation}
which enables us to rewrite the Lindblad master equation~\eqref{eqn:general_master} as 
\begin{equation}
    \dot{\rho} = -i\left(H_{\text{eff}}~\rho - \rho~H^\dagger_\text{eff}\right) + \sum_{k} \gamma_{k} L_{k} \rho L_{k}^{\dag}.
\end{equation}
The effective Hamiltonian combines the coherence decay terms with the unitary evolution, a factor of $i$ so that it generates decay, and the oscillations of the unitary evolution under $H$. Having defined this quantity, we proceed to approximate the time evolution of the open quantum system from time $t = 0$ to $t = t_{f}$ with the following steps~\cite{lidar2020lecturenotestheoryopen}:

\begin{enumerate}
\item Initialize the system in the state $\lvert \psi (0) \rangle$ and set $j=0$. For each time step $t_j \in [t_f]$:
\item Evolve the state under the effective Hamiltonian for a small time step $\tau$: $\lvert \tilde{\psi}(t_j+\tau) \rangle = e^{-iH_\text{eff}\tau} \lvert \psi(t_j) \rangle$.
\item Calculate jump probabilities for each jump operator $L_k$: $p_k = \frac{\gamma_k \|L_k \lvert \tilde{\psi}(t_j+\tau) \rangle \|^2}{\sum_k \gamma_k \|L_k \lvert \tilde{\psi}(t_j+\tau) \rangle \|^2}$.
\item On the basis of the probabilities, randomly select an $L_k$ and apply it: $|\tilde{\psi}(t_j + \tau)\rangle \mapsto \frac{L_k|\tilde{\psi}(t_j+\tau)\rangle}{\|L_k|\tilde{\psi}(t_j+\tau)\rangle \|}$. Set $t_{j+1} = t_j + \tau$.
\item Repeat Steps 2 through 4 $P$ times. 
\item For all $p \in [P]$, denote the outcome of the $p^{\text{th}}$ round as $\lvert \psi_{p}(t_{f}) \rangle$ and take the average over all such outcomes: $\rho(t_f) = \frac{1}{P}\sum_{p=1}^P \lvert \psi_p(t_f) \rangle \langle \psi_p(t_f) \rvert$. Upon convergence of $\rho(t_{f})$, this average gives an estimate of the solution.
\end{enumerate}

The space complexity of this method is proportional to the memory required to store the density matrix which, as before, has dimension $2^N D$, so the memory cost is $O(2^{2N} D^2)$.
Similar to the Liouville space method, the most expensive part of this method in terms of runtime is computing the matrix exponential $e^{-i H_\text{eff} \tau}$.
As discussed in the previous section, the cost of this operation scales as the cube of the dimension of $H_\text{eff}$, so the runtime cost is $O(2^{3N} D^3)$.
These costs are polynomially lower than those of the Liouville space method, yet still exponential.
The price of this slightly lower exponential scaling is that we trade off accuracy.
The direct method computes the exact density matrix, whereas the Monte Carlo method only approximates it; the error of approximation scales like the standard error of the mean: $1/\sqrt{P}$.

\end{document}